\documentclass[twoside]{IEEEtran}
\usepackage{amsmath,amsfonts,amssymb,amsthm}
\usepackage{url}

\newtheorem{assumption}{Assumption}[section]

\newtheorem{theorem}{Theorem}[section]

\newtheorem{lemma}{Lemma}[section]

\newtheorem{corollary}{Corollary}[section]

\newcounter{appendixcounter}

\newtheorem{theoremappendix}{Theorem}

\newtheorem{lemmaappendix}{Lemma}

\newtheorem{corollaryappendix}{Corollary}

\newtheorem{assumptionappendix}{Assumption}[appendixcounter]

\newtheorem*{remark}{Remark}

\begin{document}

\title{Analyticity of Entropy Rates of Continuous-State
Hidden Markov Models}

\author{Vladislav Z. B. Tadi\'{c} and
Arnaud Doucet
\thanks{
Vladislav Z. B. Tadi\'{c} is with School of Mathematics, University of Bristol,
Bristol BS8 1TW, United Kingdom
(e-mail: v.b.tadic@bristol.ac.uk).

Arnaud Doucet is with Department of Statistics,
University of Oxford, Oxford OX1 3LB, United Kingdom
(e-mail: doucet@stats.ox.ac.uk).
}}

\date{}

\maketitle
\maketitle

\begin{abstract}
The analyticity of the entropy and relative entropy rates of continuous-state hidden Markov models is studied here.
Using the analytic continuation principle
and the stability properties of the optimal filter,
the analyticity of these rates is established for analytically
parameterized models.
The obtained results hold under relatively mild conditions
and cover several useful classes of hidden Markov models.
These results are relevant for several theoretically and practically important problems arising in statistical inference, system identification
and information theory.
\end{abstract}


\section{Introduction}
Hidden Markov models are a powerful and versatile tool for statistical modeling
of complex time-series data and stochastic dynamic systems.
They can be described as a discrete-time Markov chain
observed through imperfect, noisy observations of its states.
Proposed in the seminal paper \cite{baum&petrie} over five decades ago,
hidden Markov models have found many applications in very diverse areas such as
acoustics and signal processing, image analysis and computer vision, automatic control,
economics and finance, computational biology, genetics and bioinformatics.
Owing to their theoretical and practical importance,
various aspects of hidden Markov models have been thoroughly studied in a number of papers and books
--- see, e.g.,
\cite{cappe&moulines&ryden}, \cite{{douc&moulines&stoffer}}, \cite{ephraim&merhav} and references therein.

The entropy and relative entropy rates of hidden Markov models can
be considered as an information-theoretic characterization
of the asymptotic properties of these models.
The entropy rate of a hidden Markov model can be interpreted as a measure of the average information
revealed by the model through noisy observations of the states.
The relative entropy rate between two hidden Markov models can be viewed as a measure of
discrepancy between these models.
The entropy rates of hidden Markov models and their analytical properties
have recently gained significant attention in the information theory community.
These properties and their links with statistical inference, system identification,
stochastic optimization and information theory have been studied extensively
in several papers \cite{han&marcus1} -- \cite{han}, \cite{holliday&goldsmith&glynn},
\cite{ordentlitch&weissman}, \cite{peres}, \cite{schoenhuht}, \cite{tadic2}.
However, to the best of our knowledge, the existing results on the analytical properties
of the entropy rates of hidden Markov models apply exclusively to scenarios
where the hidden Markov chain takes values in a finite state-space.
We establish here analytical properties of the entropy rates of
continuous-state hidden Markov models.
As indicated in \cite{tadic&doucet2}, 
such results can be very useful when analyzing algorithms for statistical inference in hidden Markov models.

In many applications, a hidden Markov model depends on an unknown parameter whose value
needs to be inferred from a set of state-observations. 
In online settings, the unknown parameter is typically estimated 
using the recursive maximum likelihood method \cite{poyiadjis&doucet&singh}, \cite{tadic2}.
In \cite{tadic2}, it has been shown that the convergence and convergence rate of recursive maximum likelihood estimation in finite-state hidden Markov models
is closely linked to the analyticity of the underlying (average) log-likelihood, i.e. of the underlying relative entropy rate.
In view of recent results on stochastic gradient search \cite{tadic&doucet3},
a similar link is expected to hold for continuous-state hidden Markov models.
However, to apply the results of \cite{tadic&doucet3} to
recursive maximum likelihood estimation in continuous-state hidden Markov models,
it is necessary to establish the analyticity of the average log-likelihood for these models.
Hence, one of the first and most important steps to carry out the asymptotic analysis of
recursive maximum likelihood estimation in continuous-state hidden Markov models
is to show the analyticity of the entropy rates of such models. 
The results presented here should provide a theoretical basis for this step.

In this paper, we study analytically parameterized continuous-state hidden Markov models
(i.e., the models whose state transition kernel and the observation conditional distribution
are analytic in the model parameters).
Using mixing conditions on the model dynamics,
we construct a geometrically ergodic analytic continuation of the state transition kernel and
an exponentially stable analytic continuation of the optimal filter.
Relying on these continuations and their asymptotic properties,
we demonstrate that the entropy and relative entropy rates are analytic in
the model parameters.
The obtained results hold under relatively mild conditions
and cover a broad and common class of state-space and continuous-state hidden Markov models.
Moreover, these results generalize
the existing results on the analyticity of entropy rates of finite-state hidden Markov models.
Additionally, the results presented here are relevant for several important
problems related to statistical inference, system identification and information theory.

The rest of this paper is organized as follows.
In Section \ref{section1}, the entropy rates of hidden Markov models are specified and the main 
results are presented.
Examples illustrating the main results are provided in Sections
\ref{section2} and \ref{section3}.
In Sections \ref{section2*} -- \ref{section3*},
the main results are proved.

\section{Main Results} \label{section1}

To define hidden Markov models and their entropy rates,
we use the following notations.
$(\Omega,{\cal F}, P)$ is a probability space.
$d_{x}\geq 1$ and $d_{y}\geq 1$ are integers,
while ${\cal X}\subseteq\mathbb{R}^{d_{x} }$ and ${\cal Y}\subseteq\mathbb{R}^{d_{y} }$
are Borel sets.
$P(x,{\rm d}x')$ is a transition kernel on ${\cal X}$,
while $Q(x,{\rm d}y)$ is a conditional probability measure on ${\cal Y}$ given $x\in{\cal X}$.
A hidden Markov model can be defined as
the ${\cal X}\times{\cal Y}$-valued stochastic process
$\{ (X_{n}, Y_{n} ) \}_{n\geq 0}$
which is defined on $(\Omega,{\cal F}, P)$ and satisfies
\begin{align*}
	&
	P\left( (X_{n+1}, Y_{n+1} )\in B
	|X_{0:n}, Y_{0:n} \right)
	\\
	&
	=
	\int I_{B}(x,y) Q(x,{\rm d}y) P(X_{n}, {\rm d}x )
\end{align*}
almost surely for $n\geq 0$ and any Borel set $B\subseteq{\cal X}\times{\cal Y}$.
$\{X_{n} \}_{n\geq 0}$ are the unobservable states,
while $\{Y_{n} \}_{n\geq 0}$ are the observations.
$Y_{n}$ can be interpreted as
a noisy measurement of state $X_{n}$.
States $\{X_{n} \}_{n\geq 0}$ form a Markov chain,
while $P(x,{\rm d}x')$ is their transition kernel.
Conditionally on $\{X_{n} \}_{n\geq 0}$, state-observations $\{Y_{n} \}_{n\geq 0}$ are mutually independent,
while $Q(X_{n}, {\rm d}y )$ is the conditional distribution of $Y_{n}$ given $X_{0:n}$.
For more details on hidden Markov models, see
\cite{cappe&moulines&ryden}, \cite{douc&moulines&stoffer} and references therein.

In addition to the model $\{ (X_{n}, Y_{n} ) \}_{n\geq 0}$,, we also consider a parameterized family of hidden Markov models.
To specify such a family,
we rely on the following notations.
Let $d\geq 1$ be an integer, while $\Theta\subset\mathbb{R}^{d}$ is an open set.
${\cal P}({\cal X})$ is the set of probability measures on ${\cal X}$.
$\mu({\rm d}x)$ and $\nu({\rm d}y)$ are measures on ${\cal X}$ and ${\cal Y}$ (respectively), while
$p_{\theta}(x'|x)$ and $q_{\theta}(y|x)$
are functions which map
$\theta\in\Theta$, $x,x'\in{\cal X}$, $y\in{\cal Y}$
to $[0,\infty )$ and satisfy
\begin{align*}
	\int_{\cal X} p_{\theta}(x'|x) \mu({\rm d}x')
	=
	\int_{\cal Y} q_{\theta}(y|x) \nu({\rm d}y)
	=
	1
\end{align*}
for all $\theta\in\Theta$, $x\in{\cal X}$.
A family of hidden Markov models can then be defined as a collection of
${\cal X}\times{\cal Y}$-valued stochastic processes
$\left\{ (X_{n}^{\theta,\lambda}, Y_{n}^{\theta,\lambda} ) \right\}_{n\geq 0}$
on $(\Omega, {\cal F}, P)$,
parameterized by $\theta\in\Theta$, $\lambda\in{\cal P}({\cal X})$ and satisfying
\begin{align*}
	&
	P\left( (X_{0}^{\theta,\lambda}, Y_{0}^{\theta,\lambda} ) \in B \right)
	=
	\int\int I_{B}(x,y) q_{\theta}(y|x) \lambda({\rm d}x),
	\\
	&
\begin{aligned}[b]
	&
	P\left(\left. (X_{n+1}^{\theta,\lambda}, Y_{n+1}^{\theta,\lambda} ) \in B\right|
	X_{0:n}^{\theta,\lambda}, Y_{0:n}^{\theta,\lambda} \right)
	\\
	&
	=
	\int\int I_{B}(x,y)
	q_{\theta}(y|x) p_{\theta}(x|X_{n}^{\theta,\lambda} ) \mu({\rm d}x) \nu({\rm d}y)
\end{aligned} 	
\end{align*}almost surely for $n\geq 0$ and any Borel set $B\subseteq{\cal X}\times{\cal Y}$.
$\{ X_{n}^{\theta,\lambda} \}_{n\geq 0}$ are the hidden states of this model, while $\{ Y_{n}^{\theta,\lambda} \}_{n\geq 0}$ are the corresponding observations.
$p_{\theta}(x'|x)$ is the transition density of the Markov chain
$\{ X_{n}^{\theta,\lambda} \}_{n\geq 0}$,
while $q_{\theta}(y|X_{n}^{\theta,\lambda} )$ is the conditional density
of $Y_{n}^{\theta,\lambda}$ given $X_{0:n}^{\theta,\lambda}$.
In the context of the identification of stochastic dynamical systems
and parameter estimation in time-series models,
$\{(X_{n},Y_{n})\}_{n\geq 0}$ is interpreted as the true system
(or true model),
while $\big\{ (X_{n}^{\theta,\lambda}, Y_{n}^{\theta,\lambda} ) \big\}_{n\geq 0}$
is viewed as a candidate model for $\{(X_{n},Y_{n})\}_{n\geq 0}$.

To define the entropy rates of hidden Markov models,
we introduce further notations.
$r_{\theta}(y,x'|x)$ is
the transition density  of
$\left\{ (X_{n}^{\theta,\lambda}, Y_{n}^{\theta,\lambda} ) \right\}_{n\geq 0}$, i.e.,
\begin{align*}
	r_{\theta}(y,x'|x)
	=
	q_{\theta}(y|x') p_{\theta}(x'|x)
\end{align*}
for $\theta\in\Theta$, $x,x'\in{\cal X}$, $y\in{\cal Y}$.
$q_{\theta}^{n}(y_{1:n}|\lambda)$ is the density  of
$Y_{1:n}^{\theta,\lambda}$, i.e.,
\begin{align*}
	q_{\theta}^{n}(y_{1:n}|\lambda)
	=&
	\int\cdots\int\int
	\left( \prod_{k=1}^{n} r_{\theta}(y_{k},x_{k}|x_{k-1} ) \right)
	\\
	&\cdot
	\mu({\rm d}x_{n})\cdots\mu({\rm d}x_{1})\lambda({\rm d}x_{0}),  	
\end{align*}
where $\lambda\in{\cal P}({\cal X})$,
$y_{1:n}=(y_{1},\dots,y_{n})\in{\cal Y}^{n}$, $n\geq 1$.
The (average) entropy
$h_{n}(\theta,\lambda)$ of $Y_{1:n}^{\theta,\lambda}$ is given by
\begin{align}\label{1.1}
	h_{n}(\theta,\lambda)
	=
	-E\left(\frac{1}{n} \log q_{\theta}^{n}\big(Y_{1:n}^{\theta,\lambda}\big|\lambda \big)\right).
\end{align}
The expected (average) log-likelihood $l_{n}(\theta,\lambda)$ of $Y_{1:n}$
given the model $\left\{ (X_{n}^{\theta,\lambda}, Y_{n}^{\theta,\lambda} ) \right\}_{n\geq 0}$ is specified as
\begin{align}\label{1.1'}
	l_{n}(\theta,\lambda)
	=
	E\left(\frac{1}{n} \log q_{\theta}^{n}\left(Y_{1:n}|\lambda \right)\right).
\end{align}
The entropy rate of model
$\left\{ (X_{n}^{\theta,\lambda}, Y_{n}^{\theta,\lambda} ) \right\}_{n\geq 0}$
(i.e., the entropy rate of stochastic process
$\left\{ Y_{n}^{\theta,\lambda} \right\}_{n\geq 0}$)
can then be defined as the limit
\begin{align*}
	\lim_{n\rightarrow\infty} h_{n}(\theta,\lambda).
\end{align*}
Similarly, the relative entropy rate between models
$\left\{ (X_{n}^{\theta,\lambda}, Y_{n}^{\theta,\lambda} ) \right\}_{n\geq 0}$
and
$\{(X_{n},Y_{n})\}_{n\geq 0}$
(i.e., the relative entropy rate between stochastic processes
$\left\{  Y_{n}^{\theta,\lambda} \right\}_{n\geq 0}$
and $\left\{  Y_{n} \right\}_{n\geq 0}$)
can be defined as the limit
\begin{align*}
	-\lim_{n\rightarrow\infty} (l_{n}(\theta,\lambda) + h ),
\end{align*}
where $h$ is the entropy rate of $\{Y_{n} \}_{n\geq 0}$
(provided $h$ exists).
In this context, the limit
\begin{align*}
	\lim_{n\rightarrow\infty } l_{n}(\theta,\lambda)
\end{align*}
can be viewed/referred to as the log-likelihood rate of
$\{Y_{n} \}_{n\geq 0}$ given the model $\left\{ (X_{n}^{\theta,\lambda}, Y_{n}^{\theta,\lambda} ) \right\}_{n\geq 0}$.
Entropy rate $\lim_{n\rightarrow\infty} h_{n}(\theta,\lambda)$ can be considered as a measure of the information
revealed by the model $\left\{ (X_{n}^{\theta,\lambda}, Y_{n}^{\theta,\lambda} ) \right\}_{n\geq 0}$
through its state-observations $\left\{ Y_{n}^{\theta,\lambda} \right\}_{n\geq 0}$.
Relative entropy rate $-\lim_{n\rightarrow\infty} (l_{n}(\theta,\lambda) + h )$ can be interpreted as a measure
of discrepancy between the models $\left\{ (X_{n}^{\theta,\lambda}, Y_{n}^{\theta,\lambda} ) \right\}_{n\geq 0}$
and $\{(X_{n},Y_{n})\}_{n\geq 0}$.
The entropy rates of hidden Markov models are closely related to a number of important problems
arising in engineering and statistics
such as system identification, parameter estimation, model reduction and data compression.
For example, in the recursive maximum likelihood approach to the identification of stochastic dynamical systems
and parameter estimation in time-series models,
the candidate model $\left\{ (X_{n}^{\theta,\lambda}, Y_{n}^{\theta,\lambda} ) \right\}_{n\geq 0}$
providing the best approximation to the true model $\{(X_{n},Y_{n})\}_{n\geq 0}$
is selected through the minimization of $-\lim_{n\rightarrow\infty} (l_{n}(\theta,\lambda) + h )$
(i.e., through the maximization of $\lim_{n\rightarrow\infty} l_{n}(\theta,\lambda)$).
For more details on the entropy rates and their applications, see \cite{cover&thomas}, \cite{gray}
and references therein.

We study here the rates
$\lim_{n\rightarrow\infty}h_{n}(\theta,\lambda)$,
$\lim_{n\rightarrow\infty}l_{n}(\theta,\lambda)$
and their analytical properties.
To formulate the assumptions under which these rates are analyzed,
we rely on the following notations.
For $\eta\in\mathbb{C}^{d}$, $\|\eta\|$ denotes the Euclidean norm of $\eta$.
For $\gamma\in(0,1)$, $V_{\gamma}(\Theta)$ is the open $\gamma$-vicinity of $\Theta$
in $\mathbb{C}^{d}$, i.e.,
\begin{align*}
	V_{\gamma}(\Theta)
	=
	\{\eta\in\mathbb{C}^{d}: \exists\theta\in\Theta, \|\eta-\theta\|<\gamma\}.
\end{align*}
Our analysis is based on the following assumptions.

\begin{assumption}\label{a11}
There exists a real number $\varepsilon\in(0,1)$
and for each $\theta\in\Theta$, $y\in{\cal Y}$,
there exists a finite measure $\lambda_{\theta}({\rm d}x|y)$ on ${\cal X}$
such that
\begin{align*}
	\varepsilon\lambda_{\theta}(B|y)
	\leq
	\int_{B} r_{\theta}(y,x'|x) \mu({\rm d}x')
	\leq
	\frac{\lambda_{\theta}(B|y)}{\varepsilon}
\end{align*}
for all $x\in{\cal X}$ and any Borel set $B\subseteq{\cal X}$.
\end{assumption}

\begin{assumption}\label{a1}
$r_{\theta}(y,x'|x)$ is real-analytic in $\theta$ for each $\theta\in\Theta$,
$x,x'\in{\cal X}$, $y\in{\cal Y}$.
Moreover, $r_{\theta}(y,x'|x)$ has a complex-valued continuation $\hat{r}_{\eta}(y,x'|x)$
with the following properties:

(i)
$\hat{r}_{\eta}(y,x'|x)$ maps
$\eta\in\mathbb{C}^{d}$, $x,x'\in{\cal X}$, $y\in{\cal Y}$
to $\mathbb{C}$.

(ii)
$\hat{r}_{\theta}(y,x'|x)=r_{\theta}(y,x'|x)$ for all $\theta\in\Theta$,
$x,x'\in{\cal X}$, $y\in{\cal Y}$.

(iii)
There exists a real number $\delta\in(0,1)$ such that
$\hat{r}_{\eta}(y,x'|x)$ is analytic in $\eta$ for
each $\eta\in V_{\delta}(\Theta)$, $x,x'\in{\cal X}$, $y\in{\cal Y}$.

(iv)
There exists a function $\varphi_{\eta}(y)$
which maps $\eta\in\mathbb{C}^{d}$, $y\in{\cal Y}$ to $\mathbb{C}$,
is analytic in $\eta$ for each $\eta\in V_{\delta}(\Theta)$, $y\in{\cal Y}$
and satisfies
\begin{align*}
	\varphi_{\eta}(y)\neq 0,
	\;\;\;\;\;\;\;
	|\hat{r}_{\eta}(y,x'|x)|\leq|\varphi_{\eta}(y)|
\end{align*}
for all $\eta\in V_{\delta}(\Theta)$, $x,x'\in{\cal X}$, $y\in{\cal Y}$.

(v)
There exist functions
$\phi, \psi:{\cal Y}\rightarrow(0,\infty)$
such that
$\int \phi(y) \nu({\rm d}y) < \infty$
and
\begin{align*}
	|\varphi_{\eta}(y)|\leq\phi(y),
	\;\;\;\;\;\;\;
	|\log|\varphi_{\eta}(y)||
	\leq
	\psi(y)
\end{align*}
for all $\eta\in V_{\delta}(\Theta)$, $y\in{\cal Y}$.
\end{assumption}

\begin{assumption}\label{a12}
There exists a real number $\gamma\in(0,1)$ such that
\begin{align*}
	\int r_{\theta}(y,x'|x)\mu({\rm d}x')\geq\gamma|\varphi_{\theta}(y)|
\end{align*}
for all $\theta\in\Theta$, $x\in{\cal X}$, $y\in{\cal Y}$.
\end{assumption}

\begin{assumption}\label{a3}
$\int\psi(y)\phi(y)\nu({\rm d}y)<\infty$.
\end{assumption}

\begin{assumption}\label{a2}
There exists a real number $K\in[1,\infty)$ such that
\begin{align*}
	\int \psi(y) Q(x,{\rm d}y) \leq K
\end{align*}
for all $x\in{\cal X}$.
Moreover, there exist a probability measure $\pi({\rm d}x)$ on ${\cal X}$ and a real number
$\rho\in(0,1)$ such that
\begin{align}\label{a.a2.1*}
	|P^{n}(x,B)-\pi(B)|\leq K\rho^{n}
\end{align}
for all $x\in{\cal X}$, $n\geq 0$ and any Borel-set $B\subseteq{\cal X}$.
\end{assumption}

Assumption \ref{a11} is related to the stability of the hidden Markov model
$\left\{(X_{n}^{\theta,\lambda},Y_{n}^{\theta,\lambda})\right\}_{n\geq 0}$
and its optimal filter.
This assumption ensures that the Markov chain
$\left\{(X_{n}^{\theta,\lambda},Y_{n}^{\theta,\lambda})\right\}_{n\geq 0}$
is geometrically ergodic
(see Lemma \ref{lemma2.1}) and that the optimal filter for the model
$\left\{(X_{n}^{\theta,\lambda},Y_{n}^{\theta,\lambda})\right\}_{n\geq 0}$
forgets initial conditions at an exponential rate
(see Lemma \ref{lemma1.1}).
In this or similar form, Assumption \ref{a11} is an ingredient of a number of asymptotic results on optimal filtering and maximum likelihood estimation in
hidden Markov models
(see \cite{delmoral&guionnet}, \cite{douc&moulines&ryden}, \cite{legland&mevel}, \cite{legland&oudjane}).

Assumption \ref{a1} is a condition on the parameterization of the model
$\left\{(X_{n}^{\theta,\lambda},Y_{n}^{\theta,\lambda})\right\}_{n\geq 0}$.
It requires the transition kernel and density of the chain
$\left\{(X_{n}^{\theta,\lambda},Y_{n}^{\theta,\lambda})\right\}_{n\geq 0}$
to be real-analytic in parameter $\theta$.
Together with Assumption \ref{a11}, Assumption \ref{a1} ensures
that an analytic continuation of this kernel exists and is geometrically ergodic
(see Lemma \ref{lemma2.1}).

Assumption \ref{a12} is also related to the parameterization
of the model $\left\{(X_{n}^{\theta,\lambda},Y_{n}^{\theta,\lambda})\right\}_{n\geq 0}$.
This assumption ensures that the ratio
\begin{align*}
	\frac{r_{\theta}(y,x'|x) }{\int r_{\theta}(y,x''|x)\mu({\rm d}x'') }
\end{align*}
is uniformly bounded in $\theta,x,x'$.
Together with Assumptions \ref{a11} and \ref{a1}, Assumption \ref{a12} ensures
that an analytic continuation of the optimal filter for the model
$\left\{(X_{n}^{\theta,\lambda},Y_{n}^{\theta,\lambda})\right\}_{n\geq 0}$
exists and forgets initial conditions at an exponential rate
(see Lemma \ref{lemma1.6}).

Assumption \ref{a3} requires the product of the bounding functions $\phi(y)$, $\psi(y)$ 
to be integrable with respect to
the measure $\nu({\rm d}y)$.
Together with Assumption \ref{a1}, Assumption \ref{a3} ensures that the entropy
$h_{n}(\theta,\lambda)$  
exists and has an analytic continuation
in $\theta$ (see Lemma \ref{lemma3.2}).

Assumption \ref{a2} ensures that the Markov chain $\left\{(X_{n},Y_{n})\right\}_{n\geq 0}$
is geometrically ergodic
(see Lemma \ref{lemma2.3}).
Together with Assumption \ref{a2}, Assumption \ref{a1} also ensures that
the log-likelihood $l_{n}(\theta,\lambda)$ defined in (\ref{1.1'}) exists
and admits an analytic continuation.

The following two theorems are the main results of the paper.

\begin{theorem}\label{theorem1}
Let Assumptions \ref{a11} -- \ref{a12} and \ref{a2} hold.
Then, there exists a function $l:\Theta\rightarrow\mathbb{R}$
such that $l(\theta)$ is real-analytic for each $\theta\in\Theta$
and $l(\theta)=\lim_{n\rightarrow\infty}l_{n}(\theta,\lambda)$
for all $\theta\in\Theta$, $\lambda\in{\cal P}({\cal X} )$.
\end{theorem}

\begin{theorem}\label{theorem2}
Let Assumptions \ref{a11} -- \ref{a3} hold.
Then, there exists a function $h:\Theta\rightarrow\mathbb{R}$
such that $h(\theta)$ is real-analytic for each $\theta\in\Theta$
and $h(\theta)=\lim_{n\rightarrow\infty}h_{n}(\theta,\lambda)$
for all $\theta\in\Theta$, $\lambda\in{\cal P}({\cal X} )$.
\end{theorem}

\begin{remark}
As $\Theta$ can be represented as a union of open balls,
it is sufficient to show Theorems \ref{theorem1} and \ref{theorem2}
for the case where $\Theta$ is convex and bounded.
Therefore, throughout the analysis carried out in Sections \ref{section2*} -- \ref{section4*},
we assume that $\Theta$ is a bounded open convex set.
\end{remark}

Theorems \ref{theorem1} and \ref{theorem2} are proved in Section \ref{section3*}.
According to these theorems, for all $\theta\in\Theta$, $\lambda\in{\cal P}({\cal X})$,
rates
$\lim_{n\rightarrow\infty}h_{n}(\theta,\lambda)$ and
$\lim_{n\rightarrow\infty}l_{n}(\theta,\lambda)$ are well-defined.
Moreover, for each $\theta\in\Theta$, the rates
$\lim_{n\rightarrow\infty}h_{n}(\theta,\lambda)$ and
$\lim_{n\rightarrow\infty}l_{n}(\theta,\lambda)$ are independent of $\lambda$
and real-analytic in $\theta$.

The analytical properties of the entropy rates
of hidden Markov models have already been extensively studied in several papers
\cite{han&marcus1} -- 
\cite{holliday&goldsmith&glynn},
\cite{ordentlitch&weissman}, \cite{peres}, \cite{schoenhuht}, \cite{tadic2}.
However, the results presented therein apply exclusively to models with finite state-spaces. 
To the best of our knowledge, Theorems \ref{theorem1} and \ref{theorem2} are the first results on
the analyticity of the entropy rates of continuous-state hidden Markov models.
These theorems also generalize
the existing results on the analyticity of the entropy rates of finite-state hidden Markov models.
More specifically, \cite{han&marcus3} can be considered as the strongest existing result
of this kind.
Theorem \ref{theorem2} includes, as a particular case, the results of \cite{han&marcus3}
and simplifies the conditions under which these results hold
(see Appendix \ref{appendix2} for details).
Theorems \ref{theorem1} and \ref{theorem2} are relevant
for several theoretically and practically important problems arising in statistical inference and system identification. 
In \cite{tadic&doucet2}, we rely on these theorems to analyze recursive maximum likelihood estimation in non-linear state-space models.
The same theorems can also be used to study the higher-order statistical asymptotics for maximum likelihood estimation in time-series models
(for details on such asymptotics, see \cite{taniguchi&kakizawa}).

\section{Example: Mixture of Densities}\label{section2}

In this section, the main results are applied to the case when
$p_{\theta}(x'|x)$ and $q_{\theta}(y|x)$ are mixtures of probability densities, i.e.,
\begin{align}
	&\label{2.1'}
	p_{\theta}(x'|x)
	=
	\sum_{i=1}^{N_{x}} a_{\theta}^{i}(x) v_{i}(x'),
	\\
	&\label{2.1''}
	q_{\theta}(y|x)
	=
	\sum_{j=1}^{N_{y}} b_{\theta}^{j}(x) w_{j}(y)
\end{align}
for $\theta\in\Theta$, $x,x'\in{\cal X}$, $y\in{\cal Y}$ and integers $N_{x}>1$ and $N_{y}>1$. 
Here $\Theta$, ${\cal X}$, ${\cal Y}$ have the same meaning as in the previous section. $\{v_{i}(x)\}_{1\leq i\leq N_{x}}$ and
$\{w_{j}(y)\}_{1\leq j\leq N_{y}}$ are functions which map
$x\in{\cal X}$, $y\in{\cal Y}$ to $[0,\infty)$
and satisfy
\begin{align*}
	\int v_{i}(x) \mu({\rm d}x)
	=
	\int w_{j}(y) \nu({\rm d}y)
	=
	1
\end{align*}
for each $1\leq i\leq N_{x}$, $1\leq j\leq N_{y}$
($\mu({\rm d}x)$, $\nu({\rm d}y)$ have the same meaning as in the previous section).
$\{a_{\theta}^{i}(x)\}_{1\leq i\leq N_{x}}$ and
$\{b_{\theta}^{j}(x)\}_{1\leq j\leq N_{y}}$ are functions which map
$\theta\in\Theta$, $x\in{\cal X}$ to $[0,\infty)$
and satisfy
\begin{align*}
	\sum_{i=1}^{N_{x}} a_{\theta}^{i}(x)
	=
	\sum_{j=1}^{N_{y}} b_{\theta}^{j}(x)
	=
	1.
\end{align*}
Under these conditions,
$v_{i}(x)$ and $w_{j}(y)$ are probability densities on ${\cal X}$ and ${\cal Y}$ (respectively),
while $a_{\theta}^{i}(x)$ and $b_{\theta}^{j}(x)$ are probability masses
in $i$ and $j$ (respectively).
Hence, in $x'$, $y$,
$p_{\theta}(x'|x)$ and $q_{\theta}(y|x)$
are mixtures of probability densities.
$v_{i}(x)$ and $w_{j}(y)$ are the components of these mixtures,
while $a_{\theta}^{i}(x)$ and $b_{\theta}^{j}(x)$ are the corresponding weights.

The entropy rates of hidden Markov model specified in (\ref{2.1'}), (\ref{2.1''}) are studied under
the following assumptions.

\begin{assumption}\label{b1}
${\cal X}$ is a compact set.
\end{assumption}

\begin{assumption}\label{b2}
$a_{\theta}^{i}(x)>0$ and $b_{\theta}^{j}(x)>0$
for all $\theta\in\Theta$, $x\in{\cal X}$, $1\leq i\leq N_{x}$, $1\leq j\leq N_{y}$.
Moreover, $a_{\theta}^{i}(x)$ and $b_{\theta}^{j}(x)$ are real-analytic in $(\theta,x)$
for each $\theta\in\Theta$, $x\in{\cal X}$, $1\leq i\leq N_{x}$, $1\leq j\leq N_{y}$.
\end{assumption}

\begin{assumption}\label{b3}
There exists a real number $\varepsilon\in(0,1)$ such that
$\varepsilon\leq v_{i}(x)\leq 1/\varepsilon$ for all $x\in{\cal X}$,
$1\leq i\leq N_{x}$.
\end{assumption}

\begin{assumption}\label{b4}
$\int |\log w_{k}(y)| w_{j}(y)\nu({\rm d}y) < \infty$ for each $1\leq j,k\leq N_{y}$.
\end{assumption}

\begin{assumption}\label{b5}
There exists a real number $K\in[1,\infty )$ such that
\begin{align*}
	\int|\log w_{k}(y)| Q(x,{\rm d}y) \leq K
\end{align*}
for all $x\in{\cal X}$, $1\leq k\leq N_{y}$.
Moreover, there exist a probability measure $\pi({\rm d}x)$ on ${\cal X}$
and a real number $\rho\in(0,1)$ such that (\ref{a.a2.1*}) holds for all $x\in{\cal X}$, $n\geq 0$
and any Borel-measurable set $B\subseteq{\cal X}$.
\end{assumption}

Assumptions \ref{b1} -- \ref{b5} cover several classes of hidden Markov models met
in practice. These assumptions indeed hold if $q_{\theta}(y|x)$ is a mixture of Gamma, Gaussian, Pareto and
logistic distributions, and if $p_{\theta}(x'|x)$ is a mixture of the same distributions truncated to a compact domain.

Using Theorem \ref{theorem1} and Theorem \ref{theorem2}, we obtain the following results.

\begin{corollary}\label{corollary1b}
Let Assumptions \ref{b1} -- \ref{b3} and \ref{b5} hold.
Then, all conclusions of Theorem \ref{theorem1} are true.
\end{corollary}

\begin{corollary}\label{corollary2b}
Let Assumptions  \ref{b1} -- \ref{b4} hold.
Then, all conclusions of Theorem \ref{theorem2} are true.
\end{corollary}

Corollaries \ref{corollary1b} and \ref{corollary2b} are proved in Section \ref{section4*}.

\section{Example: Non-Linear State-Space Models}\label{section3}

In this section, the main results are used to study the entropy rates of
non-linear state-space models.
We consider the following parameterized state-space model:
\begin{align}
	\label{3.1'}
	X_{n+1}^{\theta,\lambda}
	=&
	A_{\theta}(X_{n}^{\theta,\lambda} ) + B_{\theta}(X_{n}^{\theta,\lambda} ) V_{n},
	\\
	\label{3.1''}
	Y_{n}^{\theta,\lambda}
	=&
	C_{\theta}(X_{n}^{\theta,\lambda} ) + D_{\theta}(X_{n}^{\theta,\lambda} ) W_{n},
	\;\;\;\;\;
	n\geq 0.
\end{align}
Here $\theta\in\Theta$, $\lambda\in{\cal P}({\cal X})$ are the parameters indexing the state-space model
(\ref{3.1'}), (\ref{3.1''})
($\Theta$, ${\cal P}({\cal X})$ have the same meaning as in Section \ref{section1}).
$A_{\theta}(x)$ and $B_{\theta}(x)$
are functions which map $\theta\in\Theta$, $x\in\mathbb{R}^{d_{x}}$ (respectively) to
$\mathbb{R}^{d_{x}}$ and $\mathbb{R}^{d_{x}\times d_{x}}$
($d_{x}$ has the same meaning as in Section \ref{section1}).
$C_{\theta}(x)$ and $D_{\theta}(x)$
are functions which map $\theta\in\Theta$, $x\in\mathbb{R}^{d_{x}}$ (respectively) to
$\mathbb{R}^{d_{y}}$ and $\mathbb{R}^{d_{y}\times d_{y}}$
($d_{y}$ has the same meaning as in Section \ref{section1}).
$X_{0}^{\theta,\lambda}$ is an $\mathbb{R}^{d_{x}}$-valued random variable
defined on a probability space $(\Omega,{\cal F},P)$
and distributed according to $\lambda$.
$\{V_{n}\}_{n\geq 0}$ are $\mathbb{R}^{d_{x}}$-valued i.i.d. random variables
which are defined on $(\Omega, {\cal F}, P)$ and have (marginal) probability
density $v(x)$ with respect to the Lebesgue measure.
$\{W_{n}\}_{n\geq 0}$ are $\mathbb{R}^{d_{y}}$-valued i.i.d. random variables
which are defined on $(\Omega, {\cal F}, P)$ and have (marginal) probability
density $w(y)$ with respect to the Lebesgue measure.
We also assume that $X_{0}^{\theta,\lambda}$,
$\{V_{n} \}_{n\geq 0}$ and $\{W_{n} \}_{n\geq 0}$ are (jointly) independent.

We use here the following notations. For $\theta\in\Theta$, $x,x'\in\mathbb{R}^{d_{x}}$, $y\in\mathbb{R}^{d_{y}}$,
$\tilde{p}_{\theta}(x'|x)$ and $\tilde{q}_{\theta}(y|x)$
are the functions defined by
\begin{align*}
	\tilde{p}_{\theta}(x'|x)
	=&
	\frac{v\left(B_{\theta}^{-1}(x)(x'-A_{\theta}(x) ) \right) }
	{|\text{det}B_{\theta}(x)|},
	\\
	\tilde{q}_{\theta}(y|x)
	=&
	\frac{w\left(D_{\theta}^{-1}(x)(y-C_{\theta}(x) ) \right) }
	{|\text{det}D_{\theta}(x)|}
\end{align*}
(provided that $B_{\theta}(x)$, $D_{\theta}(x)$ are invertible), while
$p_{\theta}(x'|x)$ and $q_{\theta}(y|x)$
are defined by
\begin{align}
	\label{3.3'}
	p_{\theta}(x'|x)
	=&
	\frac{v\left(B_{\theta}^{-1}(x)(x'-A_{\theta}(x) ) \right) 1_{\cal X}(x') }
	{\int_{\cal X} v\left(B_{\theta}^{-1}(x)(x''-A_{\theta}(x) ) \right) {\rm d}x'' },
	\\
	\label{3.3''}
	q_{\theta}(y|x)
	=&
	\frac{w\left(D_{\theta}^{-1}(x)(y-C_{\theta}(x) ) \right) 1_{\cal Y}(y) }
	{\int_{\cal Y} w\left(D_{\theta}^{-1}(x)(y'-C_{\theta}(x) ) \right) {\rm d}y' }.
\end{align}
It is straightforward to show that
$\tilde{p}_{\theta}(x'|x)$  and $\tilde{q}_{\theta}(y|x)$ are the conditional densities
of $X_{n+1}^{\theta,\lambda}$ and $Y_{n}^{\theta,\lambda}$ (respectively)
given $X_{n}^{\theta,\lambda}=x$.
$p_{\theta}(x'|x)$ and $q_{\theta}(y|x)$ can be interpreted as
truncations of $\tilde{p}_{\theta}(x'|x)$ and $\tilde{q}_{\theta}(y|x)$
to domains ${\cal X}$ and ${\cal Y}$
(i.e., the hidden Markov model specified in (\ref{3.3'}), (\ref{3.3''}) can be viewed as a truncated version
of the original model (\ref{3.1'}), (\ref{3.1''})).
$p_{\theta}(x'|x)$ and $q_{\theta}(y|x)$
accurately approximate $\tilde{p}_{\theta}(x'|x)$ and $\tilde{q}_{\theta}(y|x)$
when domains ${\cal X}$ and ${\cal Y}$ are sufficiently large
(i.e., when ${\cal X}$, ${\cal Y}$ contain balls of sufficiently large radius).
This kind of approximation is involved (implicitly or explicitly) in any
numerical implementation of the optimal filter
for state-space model (\ref{3.1'}), (\ref{3.1''})
(for details see e.g., \cite{cappe&moulines&ryden}, \cite{crisan&rozovskii}, \cite{douc&moulines&stoffer}).

The entropy rates of the hidden Markov model (\ref{3.3'}), (\ref{3.3''})
are studied under the following assumptions.

\begin{assumption}\label{c1}
${\cal X}$ and ${\cal Y}$ are compact sets with non-empty interiors.
\end{assumption}

\begin{assumption}\label{c2}
$v(x)>0$ and $w(y)>0$ for all $x\in\mathbb{R}^{d_{x}}$,
$y\in\mathbb{R}^{d_{y}}$.
Moreover, $v(x)$ and $w(y)$ are real-analytic for each $x\in\mathbb{R}^{d_{x}}$,
$y\in\mathbb{R}^{d_{y}}$.
\end{assumption}

\begin{assumption}\label{c3}
$B_{\theta}(x)$ and $D_{\theta}(x)$ are invertible for all $\theta\in\Theta$,
$x\in\mathbb{R}^{d_{x}}$.
Moreover, $A_{\theta}(x)$, $B_{\theta}(x)$, $C_{\theta}(x)$ and $D_{\theta}(x)$
are real-analytic in $(\theta,x)$ for each $\theta\in\Theta$,
$x\in\mathbb{R}^{d_{x}}$.
\end{assumption}

\begin{assumption}\label{c4}
There exist a probability measure $\pi({\rm d}x)$ on ${\cal X}$
and real numbers $\rho\in(0,1)$, $K\in[1,\infty )$ such that (\ref{a.a2.1*}) holds for all $x\in{\cal X}$, $n\geq 0$
and any Borel-measurable set $B\subseteq{\cal X}$.
\end{assumption}

Assumptions \ref{c1} -- \ref{c3} are relevant for several practically important
classes of non-linear state-space models.
E.g., these assumptions cover stochastic volatility and dynamic probit models
and their truncated versions.
For other models satisfying (\ref{3.1'}), (\ref{3.1''}) and Assumptions \ref{c1} -- \ref{c3},
see \cite{cappe&moulines&ryden}, \cite{crisan&rozovskii}, \cite{douc&moulines&stoffer}
and references cited therein.

Using Theorems \ref{theorem1} and \ref{theorem2}, we get the following results.

\begin{corollary}\label{corollary1c}
Let Assumptions \ref{c1} -- \ref{c4} hold.
Then, all conclusions of Theorem \ref{theorem1} are true.
\end{corollary}

\begin{corollary}\label{corollary2c}
Let Assumptions  \ref{c1} -- \ref{c3} hold.
Then, all conclusions of Theorem \ref{theorem2} are true.
\end{corollary}

Corollaries \ref{corollary1c} and \ref{corollary2c} are proved in Section \ref{section4*}.

\section{Results Related to Kernels of
$\{ (X_{n},Y_{n}) \}_{n\geq 0}$ and
$\left\{ (X_{n}^{\theta,\lambda},Y_{n}^{\theta,\lambda} ) \right\}_{n\geq 0}$}\label{section2*}

In this section, an analytical (complex-valued) continuation of the transition kernel
of $\left\{ (X_{n}^{\theta,\lambda},Y_{n}^{\theta,\lambda} ) \right\}_{n\geq 0}$
is constructed, and its asymptotic properties (geometric ergodicity) are studied.
The same properties of the transition kernel of $\{ (X_{n},Y_{n}) \}_{n\geq 0}$
are studied, too.
Throughout this and later sections, the following notations is used.
Let ${\cal W}$ be any Borel set in $\mathbb{R}^{d_{w} }$,
where $d_{w}$ is any positive integer.
Then, ${\cal B}({\cal W} )$ denotes the collection of Borel sets in ${\cal W}$.
${\cal P}({\cal W} )$ is the collection of probability measures on ${\cal W}$,
while ${\cal M}_{p}({\cal W})$ is the set of positive measures on ${\cal W}$.
${\cal M}_{c}({\cal W})$ is the collection of complex measures on ${\cal W}$,
while ${\cal P}_{c}({\cal W})$ is the set defined by
\begin{align*}
	{\cal P}_{c}({\cal W})
	=
	\{\zeta\in{\cal M}_{c}({\cal W}): \zeta({\cal W})=1 \}.
\end{align*}
For $\zeta\in{\cal M}_{c}({\cal W})$,
$\|\zeta\|$ denotes the total variation norm of $\zeta$,
while $|\zeta|(dw)$ is the total variation of $\zeta(dw)$.
For $w\in{\cal W}$, $\delta_{w}(dw')$ is the Dirac measure centered at $w$
(i.e., $\delta_{w}(B)=I_{B}(w)$ for $B\in{\cal B}({\cal W})$).

We rely here on the following notations, too. ${\cal Z}$ is the set defined by ${\cal Z}={\cal Y}\times{\cal X}$.
$\hat{s}_{\eta}(x)$ and
$\tilde{r}_{\eta}(y,x'|x)$ are the functions defined by
\begin{align}
	&
	\hat{s}_{\eta}(x)
	=
	\int\int \hat{r}_{\eta}(y',x''|x)\nu({\rm d}y')\mu({\rm d}x''),
	\\
	&\label{5.51}
	\tilde{r}_{\eta}(y,x'|x)
	=
	\begin{cases}
	\hat{r}_{\eta}(y,x'|x)/\hat{s}_{\eta}(x),
	&\text{ if } \hat{s}_{\eta}(x)\neq 0
	\\
	0, &\text{ otherwise }
	\end{cases}
\end{align}
for $\eta\in\mathbb{C}^{d}$, $x,x'\in{\cal X}$, $y\in{\cal Y}$.
$\tilde{\psi}(z)$ is the function defined by
\begin{align}\label{5.5}
	\tilde{\psi}(z)
	=
	1+\psi(y),
\end{align}
where $z=(y,x)$.
$u_{\eta}^{n}(x_{0:n},y_{1:n})$ is the function defined by
\begin{align}\label{5.3}
	u_{\eta}^{n}(x_{0:n},y_{1:n})
	=
	\prod_{k=1}^{n} \tilde{r}_{\eta}(y_{k},x_{k}|x_{k-1}),
\end{align}
where $x_{0},\dots,x_{n}\in{\cal X}$,
$y_{1},\dots,y_{n}\in{\cal Y}$, $n\geq 1$.
$\sigma({\rm d}z)$ is the measure defined by
\begin{align*}
	\sigma(B)
	=
	\int\int I_{B}(y,x) Q(x,{\rm d}y) \pi({\rm d}x)
\end{align*}
for $B\in{\cal B}({\cal Z})$.
$S(z,{\rm d}z')$, $S_{\eta}(z,{\rm d}z')$ are the kernels defined by
\begin{align}
	\label{5.1'}
	S(z,B)
	=&
	\int\int I_{B}(y',x') Q(x',{\rm d}y') P(x,{\rm d}x'),
	\\
	\label{5.1''}
	S_{\eta}(z,B)
	=&
	\int\int I_{B}(y',x') \tilde{r}_{\eta}(y',x'|x) \nu({\rm d}y')\mu({\rm d}x')
\end{align}
(as in (\ref{5.5}), $z$ denotes $(y,z)$). 
$\{S^{n}(z,{\rm d}z') \}_{n\geq 0}$, $\{S_{\eta}^{n}(z,{\rm d}z') \}_{n\geq 0}$ are the
kernels recursively defined by
$S^{0}(z,B)=S_{\eta}^{0}(z,B)=\delta_{z}(B)$ and
\begin{align*}
	S^{n+1}(z,B)
	=&
	\int S^{n}(z',B) S(z,{\rm d}z'),
	\\
	S_{\eta}^{n+1}(z,B)
	=&
	\int S_{\eta}^{n}(z',B) S_{\eta}(z,{\rm d}z').
\end{align*}
$\{(S_{\eta}^{n}\zeta)({\rm d}z) \}_{n\geq 0}$ are the
measures defined by
\begin{align*}
	(S_{\eta}^{n}\zeta)(B)
	=
	\int S_{\eta}^{n}(z,B) \zeta({\rm d}z),
\end{align*}
where $\zeta\in{\cal M}_{c}({\cal Z})$.

\begin{remark}
$S(z,{\rm d}z')$ and $\sigma({\rm d}z)$ are the transition kernel and the invariant distribution
of $\{ (X_{n},Y_{n} ) \}_{n\geq 0}$.
When $\theta\in\Theta$, $S_{\theta}(z,{\rm d}z')$ boils down to the transition kernel of
$\left\{ (X_{n}^{\theta,\lambda}, Y_{n}^{\theta,\lambda} ) \right\}_{n\geq 0}$.
Hence, for $\eta\in\mathbb{C}^{d}$, $S_{\eta}(z,{\rm d}z')$ can be considered as
a complex-valued continuation of the transition kernel of
$\left\{ (X_{n}^{\theta,\lambda}, Y_{n}^{\theta,\lambda} ) \right\}_{n\geq 0}$.
The kernel $S_{\eta}^{n}(z,{\rm d}z')$ admits the representation
\begin{align}\label{5.71}
	(S_{\eta}^{n}\zeta)(B)	
	=&
	\int\cdots\int\int
	I_{B}(y_{n},x_{n}) u_{\eta}^{n}(x_{0:n},y_{1:n})
	\nonumber\\
	&\cdot
	(\nu\times\mu)({\rm d}y_{n},{\rm d}x_{n}) \cdots (\nu\times\mu)({\rm d}y_{1},{\rm d}x_{1})
	\nonumber\\
	&\cdot
	\zeta({\rm d}y_{0},{\rm d}x_{0}).
\end{align}
This representation is used to show that
$S_{\eta}(z,{\rm d}z')$ is geometrically ergodic
(see Lemma \ref{lemma2.1} and its proof).
It is also used to show the analyticity of integral (\ref{l3.2.3*})
(see Lemma \ref{lemma3.2} and its proof).
\end{remark}

\begin{remark}
Throughout this section and later sections, the following convention is applied.
Diacritic $\tilde{}$ is used to denote a locally defined quantity,
i.e., a quantity whose definition holds only within the proof where
the quantity appears.
\end{remark}

\begin{lemma}\label{lemma2.3}
Let Assumption \ref{a2} hold.
Then, there exists a real number $C_{1}\in[1,\infty )$ such that
\begin{align*}
	&
	\int \tilde{\psi}(z') S(z,{\rm d}z') \leq C_{1},
	\\
	&
	|S^{n}-\sigma|(z,B)
	\leq
	C_{1}\rho^{n}
\end{align*}
for all $z\in{\cal Z}$, $B\in{\cal B}({\cal Z})$, $n\geq 0$
(here, $|S^{n}-\sigma|(z,{\rm d}z')$ denotes the total variation of
$S^{n}(z,{\rm d}z')-\sigma({\rm d}z')$,
while $\rho$ is specified in Assumption \ref{a2}).
\end{lemma}

\begin{proof}
Let $C_{1}=2K$ ($K$ is specified in Assumption \ref{a2}).
Moreover, let $x$, $y$ be any elements of ${\cal X}$, ${\cal Y}$ (respectively),
while $z=(y,x)$ (notice that $z$ can be any element of ${\cal Z}$).
Then, we have
\begin{align*}
	\int\tilde{\psi}(z')S(z,{\rm d}z')
	=&
	\int\int \left(1+\psi(y')\right) Q(x',{\rm d}y') P(x,{\rm d}x')
	\\
	\leq &
	1+K
	\leq
	C_{1}.
\end{align*}
We also have
\begin{align*}
	&
	|S^{n}(z,B)-\sigma(B)|
	\\
	&
	=
	\left|
	\int\int I_{B}(y',x') Q(x',{\rm d}y')
	(P^{n}-\pi)(x,{\rm d}x')
	\right|
	\\
	&\leq
	\int\int I_{B}(y',x') Q(x',{\rm d}y')
	|P^{n}-\pi|(x,{\rm d}x')
	\\
	&\leq
	2K\rho^{n}
	\leq 
	C_{1}\rho^{n}
\end{align*}
for $B\in{\cal B}({\cal Z})$, $n\geq 0$.
\end{proof}

\begin{lemma}\label{lemma2.4}
Let Assumption \ref{a1} hold. Then, the following is true:

(i) $\tilde{r}_{\theta}(y,x'|x)=r_{\theta}(y,x'|x)$
for all $\theta\in\Theta$, $x,x'\in{\cal X}$, $y\in{\cal Y}$.

(ii) There exists a real number $\delta_{1}\in(0,\delta]$
such that $\tilde{r}_{\eta}(y,x'|x)$ is analytic in $\eta$
and satisfies
\begin{align*}
	&
	\left|\tilde{r}_{\eta}(y,x'|x) \right|
	\leq
	2|\varphi_{\eta}(y)|,
	\\
	&
	\int\int \tilde{r}_{\eta}(y',x''|x)\nu({\rm d}y')\mu({\rm d}x'') = 1
\end{align*}
for all $\eta\in V_{\delta_{1}}(\Theta)$, $x,x'\in{\cal X}$,
$y\in{\cal Y}$
($\delta$ is specified in Assumption \ref{a1}).
\end{lemma}

\begin{remark}
As a direct consequence of Lemma \ref{lemma2.4} (Part (ii)),
we have $S_{\eta}^{n}\zeta\in{\cal P}_{c}({\cal Z})$
(i.e., $(S_{\eta}^{n}\zeta)({\cal Z})=1$)
for $\eta\in V_{\delta_{1}}(\Theta)$,
$\zeta\in{\cal P}_{c}({\cal Z})$, $n\geq 1$.
\end{remark}

\begin{proof}
Due to Assumption \ref{a1}, we have
\begin{align*}
	\int\int\phi(y)\nu({\rm d}y)\mu({\rm d}x)
	=
	\|\mu\|\int\phi(y)\nu({\rm d}y)
	<\infty.
\end{align*}
Then, using Assumption \ref{a1} and Lemma \ref{lemmaa1} (see Appendix \ref{appendix1}),
we conclude that $\hat{s}_{\eta}(x)$
is analytic in $\eta$
for each $\eta\in V_{\delta}(\Theta)$, $x\in{\cal X}$.
Relying on the same arguments, we deduce
\begin{align}\label{l2.4.1}
	\left|
	\hat{r}_{\eta'}(y,x'|x) - \hat{r}_{\eta''}(y,x'|x)
	\right|
	\leq &
	\frac{d\:\phi(y)\|\eta'-\eta''\|}{\delta}
\end{align}
for $\eta',\eta''\in V_{\delta}(\Theta)$, $x,x'\in{\cal X}$, $y\in{\cal Y}$
(here, $d$ denotes the dimension of vectors in $\Theta$, $V_{\delta}(\Theta)$).

Throughout the rest of the proof,
the following notations is used.
$\tilde{C}$, $\delta_{1}$ are the real numbers defined by
\begin{align*}
	&
	\tilde{C}
	=
	\frac{d\|\mu\|}{\delta}
	\int \phi(y)\nu({\rm d}y),
	\\
	&
	\delta_{1}
	=
	\min\left\{\delta, \frac{1}{2\tilde{C} } \right\}.
\end{align*}
$\eta$, $\eta'$, $\eta''$ are any elements in $V_{\delta_{1} }(\Theta)$,
while $\theta$ is any element of $\Theta$ satisfying
$\|\eta-\theta\|<\delta_{1}$.
$x$, $x'$ are any elements of ${\cal X}$,
while $y$ is any element in ${\cal Y}$.

Using (\ref{l2.4.1}), we conclude
\begin{align*}
	&
	\left|
	\int\int
	\left(\hat{r}_{\eta'}(y,x'|x) - \hat{r}_{\eta''}(y,x'|x) \right)
	\nu({\rm d}y)\mu({\rm d}x')
	\right|
	\\
	&\leq
	\int\int
	\left|\hat{r}_{\eta'}(y,x'|x) - \hat{r}_{\eta''}(y,x'|x) \right|
	\nu({\rm d}y)\mu({\rm d}x')
	\\
	&\leq
	\frac{d\|\mu\| \|\eta'-\eta''\|}{\delta}
	\int\phi(y)\nu({\rm d}y)
	=
	\tilde{C}\|\eta'-\eta''\|.
\end{align*}
Consequently, we have
\begin{align*}
	\left|\hat{s}_{\eta}(x)\right|
	=&
	\left|
	\int\!\int \hat{r}_{\eta}(y,x'|x) \nu({\rm d}y)\mu({\rm d}x')
	\right|
	\\
	\geq &
	\int\!\int \hat{r}_{\theta}(y,x'|x) \nu({\rm d}y)\mu({\rm d}x')
	\\
	&
	-
	\left|
	\int\!\int
	\left(\hat{r}_{\eta}(y,x'|x) - \hat{r}_{\theta}(y,x'|x) \right)
	\nu({\rm d}y)\mu({\rm d}x')
	\right|
	\\
	\geq &
	1-\tilde{C}\|\eta-\theta\|
	\geq 
	\frac{1}{2}.
\end{align*}
Hence, we get
\begin{align}\label{l2.4.5}
	\left|\hat{s}_{\eta}(x)\right|
	=
	\left|
	\int\!\int \hat{r}_{\eta}(y,x'|x) \nu({\rm d}y)\mu({\rm d}x')
	\right|
	\geq
	\frac{1}{2}.
\end{align}
Therefore, we have
\begin{align}\label{l2.4.7}
	\tilde{r}_{\eta}(y,x'|x)
	=
	\frac{\hat{r}_{\eta}(y,x'|x) }
	{\hat{s}_{\eta}(x) }.
\end{align}
As $\hat{s}_{\eta}(x)$
is analytic in $\eta$
for each $\eta\in V_{\delta_{1}}(\Theta)$,
we conclude from Assumption \ref{a1} and (\ref{l2.4.5}), (\ref{l2.4.7})
that (i), (ii) are true.
\end{proof}

\begin{lemma}\label{lemma2.2}
Let Assumption \ref{a1} hold. Then, the following is true:

(i) $u_{\eta}^{n}(x_{0:n},y_{1:n} )$ is analytic in $\eta$ for all $\eta\in V_{\delta_{1}}(\Theta)$,
$x_{0},\dots,x_{n}\in{\cal X}$, $y_{1},\dots,y_{n}\in{\cal Y}$, $n\geq 1$
($\delta_{1}$ is specified in Lemma \ref{lemma2.4}).

(ii) There exists a non-decreasing sequence
$\{K_{n}\}_{n\geq 1}$ in $[1,\infty)$ such that
\begin{align*}
	&
	\left|
	u_{\eta}^{n}(x_{0:n},y_{1:n} )
	\right|
	\leq
	K_{n}
	\left( \prod_{k=1}^{n} \phi(y_{k} ) \right),
	\nonumber\\
	&
	\left|
	u_{\eta'}^{n}(x_{0:n},y_{1:n} )
	-
	u_{\eta''}^{n}(x_{0:n},y_{1:n} )
	\right|
	\\
	&
	\leq
	K_{n} \|\eta'-\eta''\|
	\left( \prod_{k=1}^{n} \phi(y_{k} ) \right)
	\nonumber
\end{align*}
for all $\eta,\eta',\eta''\in V_{\delta_{1}}(\Theta)$,
$x_{0},\dots,x_{n}\in{\cal X}$,
$y_{1},\dots,y_{n}\in{\cal Y}$,
$n\geq 1$.
\end{lemma}
	
\begin{proof}
Throughout the proof, the following notations is used.
$\{ K_{n} \}_{n\geq 1}$ are the real numbers defined by
$K_{n}=2^{n}d/\delta_{1} $ for $n\geq 1$
(here, $d$ denotes the dimension of vectors in $\Theta$, $V_{\delta}(\Theta)$).
$\eta$, $\eta'$, $\eta''$ are any elements of $V_{\delta_{1} }(\Theta)$.
$\{ x_{n} \}_{n\geq 0}$, $\{y_{n} \}_{n\geq 1}$ are any sequences in ${\cal X}$,
${\cal Y}$ (respectively).

Owing to Lemma \ref{lemma2.4},
$u_{\eta}^{n}(x_{0:n},y_{1:n} )$ is analytic in $\eta$
for each $\eta\in V_{\delta_{1} }(\Theta)$.
Due to Assumption \ref{a1} and the same lemma, we have
\begin{align*}
	&
	\left|
	u_{\eta}^{n}(x_{0:n},y_{1:n} )
	\right|
	\leq
	2^{n} \left( \prod_{k=1}^{n} |\varphi_{\eta}|(y_{k} ) \right)
	\leq
	K_{n}
	\left( \prod_{k=1}^{n} \phi(y_{k} ) \right)
\end{align*}
for $n\geq 1$.
Consequently, Lemma \ref{lemmaa1} (see Appendix \ref{appendix1}) yields
\begin{align*}
	&
	\left|
	u_{\eta'}^{n}(x_{0:n},y_{1:n} )
	-
	u_{\eta''}^{n}(x_{0:n},y_{1:n} )
	\right|
	\\
	&
	\leq
	\frac{2^{n}d \|\eta'-\eta''\| }{\delta_{1} }
	\left( \prod_{k=1}^{n} \phi(y_{k} ) \right)
	\!=\! 
	K_{n} \|\eta'-\eta''\|
	\left( \prod_{k=1}^{n} \phi(y_{k} ) \right)
\end{align*}
for $n\geq 1$.
\end{proof}

\begin{lemma}\label{lemma2.1}
Let Assumptions \ref{a11}, \ref{a1} and \ref{a3} hold. Then, the following is true:

(i) There exist real numbers $\delta_{2}\in(0,\delta_{1}]$, $C_{2}\in[1,\infty)$ such that
\begin{align*}
	&
	\left|
	S_{\eta'} - S_{\eta''}
	\right|(z,B)
	\leq
	C_{2} \|\eta'-\eta''\|,
	\\
	&
	\int \tilde{\psi}(z') \left|S_{\eta}\right|(z,{\rm d}z')
	\leq
	C_{2}
\end{align*}
for all $\eta,\eta',\eta''\in V_{\delta_{2}}(\Theta)$,
$z\in{\cal Z}$, $B\in{\cal B}({\cal Z})$
(here, $\left|S_{\eta'}-S_{\eta''}\right|(z,{\rm d}z')$
denotes the total variation of $S_{\eta'}(z,{\rm d}z')-S_{\eta''}(z,{\rm d}z')$,
while $\delta_{1}$ is specified in Lemma \ref{lemma2.4}).

(ii) For each $\eta\!\in\! V_{\delta_{2}}(\Theta)$,
there exists a complex measure $\sigma_{\eta}({\rm d}z)$ on ${\cal Z}$ such that
$\sigma_{\eta}(B)\!=\lim_{n\rightarrow\infty} S_{\eta}^{n}(z,B)$
for all $z\in{\cal Z}$, $B\in{\cal B}({\cal Z})$.

(iii) There exists a real number $\gamma_{1} \in (0,1)$,
such that
\begin{align*}
	\left|
	S_{\eta}^{n}-\sigma_{\eta}
	\right|(z,B)
	\leq
	C_{2}\gamma_{1}^{n}
\end{align*}
for all $\eta\in V_{\delta_{2}}(\Theta)$, $z\in{\cal Z}$, $B\in{\cal B}({\cal Z})$, $n\geq 0$
(here, $\left|S_{\eta}^{n}-\sigma_{\eta}\right|(z,{\rm d}z')$
stands for the total variation of
$S_{\eta}^{n}(z,{\rm d}z')-\sigma_{\eta}({\rm d}z')$).
\end{lemma}

\begin{proof}
Throughout the proof, the following notations is used.
$\tilde{C}_{1}$, $\tilde{C}_{2}$ are the real numbers defined by
\begin{align*}
	&
	\tilde{C}_{1}
	=
	2\|\mu\|\int\psi(y)\phi(y)\nu({\rm d}y),
	\\
	&
	\tilde{C}_{2}
	=
	2\|\mu\|\int\phi(y)\nu({\rm d}y),
\end{align*}
while $n_{0}$ is the integer defined as
\begin{align*}
	n_{0}
	=
	\left\lceil \frac{\log4}{|\log(1-\varepsilon^{2} ) |} \right\rceil
\end{align*}
($\varepsilon$, $\phi(y)$, $\psi(y)$,  are specified in Assumptions \ref{a11}, \ref{a12}).
$\{\tilde{K}_{n} \}_{n\geq 1}$ are the real numbers defined by
$\tilde{K}_{n}=(1+\tilde{C}_{2})^{n}K_{n}$ for $n\geq 1$,
while $\delta_{2}$, $\gamma_{1}$, $\tilde{C}_{3}$, $C_{2}$ are the real numbers
defined as $\delta_{2}=\delta_{1}/(4\tilde{K}_{n_{0}} )$, $\gamma_{1}=2^{-1/n_{0}}$,
$\tilde{C}_{3}=\tilde{K}_{1}+\tilde{C}_{1}+\tilde{C}_{2}$,
$C_{2}=16\tilde{C}_{3}\gamma_{1}^{-n_{0}}$
($\delta_{1}$, $K_{n}$ are specified in
Lemmas \ref{lemma2.4}, \ref{lemma2.2}).
$\eta$, $\eta'$, $\eta''$ are any elements in $V_{\delta_{2} }(\Theta)$,
while
$\theta$ is any element of $\Theta$ satisfying $\|\eta-\theta\|<\delta_{2}$.
$x$, $y$ are any elements of ${\cal X}$, ${\cal Y}$ (respectively),
while $z=(y,x)$.
$\zeta$, $\zeta'$, $\zeta''$ are any elements of ${\cal P}_{c}({\cal Z} )$,
while $B$ is any element of ${\cal B}({\cal Z})$.
$n\geq 1$, $k\geq 0$ are any integers.

Relying on Assumption \ref{a3} and Lemma \ref{lemma2.4}, we deduce
\begin{align*}
	\int\tilde{\psi}(z')\left|S_{\eta}\right|(z,{\rm d}z')
	\leq &
	\int\int
	\left(1+\psi(y') \right)
	\left|\tilde{r}_{\eta}(y',x'|x) \right|
	\\
	&\cdot
	\nu({\rm d}y') \mu({\rm d}x')
	\\
	\leq &
	2\|\mu\|
	\int\left(1+\psi(y') \right) \varphi(y') \nu({\rm d}y')
	\\
	= &
	\tilde{C}_{1}+\tilde{C}_{2}
	\leq 
	C_{2}
\end{align*}
as $\tilde{C}_{1}+\tilde{C}_{2} \leq \tilde{C}_{3} \leq C_{2}$.
Moreover, using Lemma \ref{lemma2.2}, we conclude
\begin{align*}
	&
	\left|
	(S_{\eta'}^{n}\zeta )(B)
	-
	(S_{\eta''}^{n}\zeta )(B)
	\right|
	\\
	&
	\begin{aligned}
	\leq
	&
	\int\cdots\int\int
	I_{B}(y_{n},x_{n})
	\left|
	u_{\eta'}^{n}(x_{0:n},y_{1:n})
	-
	u_{\eta''}^{n}(x_{0:n},y_{1:n})
	\right|
	\\
	&\cdot
	(\nu\times\mu)({\rm d}y_{n},{\rm d}x_{n}) \cdots (\nu\times\mu)({\rm d}y_{1},{\rm d}x_{1})
	|\zeta|({\rm d}y_{0},{\rm d}x_{0})
	\end{aligned}
	\\
	&\leq
	K_{n} \|\mu\|^{n} \|\zeta\| \|\eta'-\eta''\|
	\left(
	\prod_{k=1}^{n}
	\int\phi(y_{k})\nu({\rm d}y_{k})
	\right)
	\\
	&\leq
	\tilde{K}_{n} \|\zeta\| \|\eta'-\eta''\|.
\end{align*}
Therefore, we get
\begin{align}\label{l2.1.1}
	\left\|
	S_{\eta'}^{n}\zeta - S_{\eta''}^{n}\zeta
	\right\|
	\leq
	\tilde{K}_{n}\|\zeta\|\|\eta'-\eta''\|.
\end{align}
Hence, we have
\begin{align*}
	\left|
	S_{\eta'}
	-
	S_{\eta''}
	\right|(z,B)
	=&
	\left|
	S_{\eta'}\delta_{z}
	-
	S_{\eta''}\delta_{z}
	\right|(B)
	\\
	\leq &
	\tilde{K}_{1}\|\delta_{z}\|\|\eta'-\eta''\|
	\\
	\leq &
	C_{2}\|\eta'-\eta''\| 
\end{align*}
as $\tilde{K}_{1}\leq\tilde{C}_{3}\leq C_{2}$.

Let $\tau_{\theta}({\rm d}z)$ be the measure defined by
\begin{align*}
	\tau_{\theta}(B)
	=
	\int\int I_{B}(y,x) \lambda_{\theta}({\rm d}x|y)\nu({\rm d}y).
\end{align*}
Owing to Assumption \ref{a11}, we have
\begin{align*}
	1
	=
	\int\int r_{\theta}(y,x'|x)\nu({\rm d}y)\mu({\rm d}x')
	\leq
	\frac{1}{\varepsilon}
	\int\lambda_{\theta}({\cal X}|y)\nu({\rm d}y).
\end{align*}
Hence, we get
\begin{align*}
	\tau_{\theta}({\cal Z} )
	=
	\int\lambda_{\theta}({\cal X}|y)\nu({\rm d}y)
	\geq
	\varepsilon.
\end{align*}
Moreover, due to Assumption \ref{a11} and Lemma \ref{lemma2.4}, we have
\begin{align*}
	S_{\theta}(z,B)
	=&
	\int\int I_{B}(y',x') r_{\theta}(y',x'|x) \nu({\rm d}y') \mu({\rm d}x')
	\\
	\geq &
	\varepsilon \int\int I_{B}(y',x') \lambda_{\theta}({\rm d}x'|y') \nu({\rm d}y')
	=
	\varepsilon \tau_{\theta}(B). 
\end{align*}
Then, standard results in Markov chain theory
(see e.g., \cite[Theorem 16.0.2]{meyn&tweedie})
imply that there exists
a probability measure $\sigma_{\theta}({\rm d}z)$ on ${\cal Z}$
such that
\begin{align*}
	\left|
	S_{\theta}^{n}(z,B)-\sigma_{\theta}(B)
	\right|
	\leq
	\left(
	1 - \varepsilon \tau_{\theta}({\cal Z})
	\right)^{n}
	\leq
	(1-\varepsilon^{2} )^{n}.
\end{align*}
As $\sigma_{\theta}(B)(\zeta'({\cal Z})-\zeta''({\cal Z}) )=0$, we get
\begin{align*}
	&
	\left|
	(S_{\theta}^{n}\zeta')(B) - (S_{\theta}^{n}\zeta'')(B)
	\right|
	\\
	&=
	\left|
	\int
	\left(S_{\theta}^{n} - \sigma_{\theta} \right)(z,B)
	(\zeta'-\zeta'')({\rm d}z)
	\right|
	\\
	&\leq
	\int
	\left|S_{\theta}^{n}-\sigma_{\theta} \right|(z,B)
	|\zeta'-\zeta''|({\rm d}z)
	\\
	&\leq
	(1-\varepsilon^{2} )^{n} \|\zeta'-\zeta''\|.
\end{align*}
Hence, we have
\begin{align}\label{l2.1.3}
	\left\|
	S_{\theta}^{n}\zeta' - S_{\theta}^{n}\zeta''
	\right\|
	\leq
	(1-\varepsilon^{2} )^{n} \|\zeta'-\zeta''\|.
\end{align}

Since $S_{\theta}^{n}(z,{\rm d}z')$ is an element of ${\cal P}({\cal Z} )$, we conclude
$\|S_{\theta}^{n}\zeta\|\leq\|\zeta\|$.
Then, owing to (\ref{l2.1.1}), we have
\begin{align}\label{l2.1.5}
	\left\|S_{\eta}^{n}\zeta \right\|
	\leq &
	\left\|S_{\theta}^{n}\zeta\right\|
	+
	\left\|\left(S_{\eta}^{n}-S_{\theta}^{n} \right) \zeta \right\|
	\nonumber\\
	\leq &
	\left(1 + \tilde{K}_{n}\|\eta-\theta\| \right) \|\zeta\|
	\nonumber\\
	\leq &
	(1 + \tilde{K}_{n}\delta_{2} ) \|\zeta\|
	\leq
	2\|\zeta\|
\end{align}
when $n\leq n_{0}$, 
as $\|\eta-\theta\|<\delta_{2}$,
$\tilde{K}_{n}\delta_{2}\leq\tilde{K}_{n_{0}}\delta_{2}=\delta_{1}/4\leq 1/4$.
Moreover, due to (\ref{l2.1.1}), (\ref{l2.1.3}), we have
\begin{align*}
	\left\|
	S_{\eta}^{n}\zeta' - S_{\eta}^{n}\zeta''
	\right\|
	\leq &
	\left\|
	S_{\theta}^{n}\zeta' - S_{\theta}^{n}\zeta''
	\right\|
	+
	\left\|
	(S_{\eta}^{n}-S_{\theta}^{n} ) (\zeta'-\zeta'')
	\right\|
	\\
	\leq &
	\left((1-\varepsilon^{2} )^{n} + \tilde{K}_{n}\|\eta-\theta\| \right)
	\|\zeta'-\zeta''\|
	\\
	\leq &
	\left(
	(1-\varepsilon^{2} )^{n} + \frac{1}{4}
	\right)
	\|\zeta'-\zeta''\|
\end{align*}
when $n\leq n_{0}$.
Setting $n=n_{0}$, we conclude
\begin{align*}
	\left\|
	S_{\eta}^{n_{0}}\zeta' - S_{\eta}^{n_{0}}\zeta''
	\right\|
	\leq
	\frac{\|\zeta'-\zeta''\|}{2}
\end{align*}
as $(1-\varepsilon^{2} )^{n_{0}}\leq 1/4$.
Since $S_{\eta}^{n}\zeta\in {\cal P}_{c}({\cal Z})$
(see Lemma \ref{lemma2.4} and the remark immediately after its statement),
we have
\begin{align}\label{l2.1.7}
	\left\|
	S_{\eta}^{(k+1)n_{0}} (\zeta'-\zeta'')
	\right\|
	=&
	\left\|
	S_{\eta}^{n_{0}}
	\left(
	S_{\eta}^{kn_{0}}\zeta' - S_{\eta}^{kn_{0}}\zeta''
	\right)
	\right\|
	\nonumber\\
	\leq &
	\frac{1}{2}
	\left\|
	S_{\eta}^{kn_{0}} (\zeta'-\zeta'')
	\right\|.
\end{align}
Iterating (\ref{l2.1.7}), we get
\begin{align}\label{l2.1.9}
	\left\|
	S_{\eta}^{kn_{0}} (\zeta'-\zeta'')
	\right\|
	\leq
	\frac{1}{2^{k}} \|\zeta'-\zeta''\|.
\end{align}
Using (\ref{l2.1.5}), (\ref{l2.1.9}), we conclude
\begin{align}\label{l2.1.21}
	\left\|
	S_{\eta}^{(k+1)n_{0}}\zeta - S_{\eta}^{kn_{0}}\zeta
	\right\|
	=&
	\left\|
	S_{\eta}^{kn_{0}}
	\left(S_{\eta}^{n_{0}}\zeta - \zeta
	\right)
	\right\|
	\nonumber\\
	\leq &
	\frac{1}{2^{k}}
	\left\|
	S_{\eta}^{n_{0}}\zeta - \zeta
	\right\|
	\nonumber\\
	\leq &
	\frac{1}{2^{k}}
	\left(
	\left\|
	S_{\eta}^{n_{0}}\zeta
	\right\|
	+
	\|\zeta\|
	\right)
	\nonumber\\
	\leq &
	\frac{\|\zeta\|}{2^{k-2}}.
\end{align}
Hence, we get
\begin{align}\label{l2.1.23}
	\sum_{k=0}^{\infty}
	\left\|
	S_{\eta}^{(k+1)n_{0}}\zeta - S_{\eta}^{kn_{0}}\zeta
	\right\|
	\leq
	\sum_{k=0}^{\infty} \frac{\|\zeta\|}{2^{k-2} }
	=
	8\|\zeta\|
	<
	\infty.
\end{align}

Let $(S_{\eta}^{\infty}\zeta )({\rm d}z)$ be the measure defined by
\begin{align*}
	(S_{\eta}^{\infty}\zeta )(B)
	=
	\zeta(B)
	+
	\sum_{k=0}^{\infty}
	\left(
	(S_{\eta}^{(k+1)n_{0}}\zeta)(B) - (S_{\eta}^{kn_{0}}\zeta)(B)
	\right).
\end{align*}
Then, due to (\ref{l2.1.23}),
$(S_{\eta}^{\infty}\zeta )({\rm d}z)$ is well-defined and satisfies
$S_{\eta}^{\infty}\zeta\in{\cal P}_{c}({\cal Z})$.
Moreover, owing to (\ref{l2.1.21}), (\ref{l2.1.23}), we have
\begin{align}\label{l2.1.25}
	\left\|
	S_{\eta}^{kn_{0}}\zeta - S_{\eta}^{\infty}\zeta
	\right\|
	=&
	\left\|
	\sum_{j=k}^{\infty}
	\left(
	S_{\eta}^{(j+1)n_{0}}\zeta - S_{\eta}^{jn_{0}}\zeta
	\right)
	\right\|
	\nonumber\\
	\leq &
	\sum_{j=k}^{\infty}
	\left\|
	S_{\eta}^{(j+1)n_{0}}\zeta - S_{\eta}^{jn_{0}}\zeta
	\right\|
	\nonumber\\
	\leq &
	\sum_{j=k}^{\infty} \frac{\|\zeta\|}{2^{j-2} }
	=
	\frac{\|\zeta\|}{2^{k-3}}.
\end{align}
Combining this with (\ref{l2.1.9}), we get
\begin{align*}
	\left\|
	S_{\eta}^{\infty}\zeta'
	-
	S_{\eta}^{\infty}\zeta''
	\right\|
	\leq &
	\left\|
	S_{\eta}^{kn_{0}}\zeta'
	-
	S_{\eta}^{\infty}\zeta'
	\right\|
	+
	\left\|
	S_{\eta}^{kn_{0}}\zeta''
	-
	S_{\eta}^{\infty}\zeta''
	\right\|
	\\
	&+
	\left\|
	S_{\eta}^{kn_{0}}\zeta'
	-
	S_{\eta}^{kn_{0}}\zeta''
	\right\|
	\\
	\leq &
	\frac{\|\zeta'\|+\|\zeta''\|+\|\zeta'-\zeta''\|}{2^{k-3}}.
\end{align*}
Therefore,
$S_{\eta}^{\infty}\zeta'=S_{\eta}^{\infty}\zeta''$
for any $\zeta',\zeta''\in{\cal P}_{c}({\cal Z})$.
Consequently, there exists $\sigma_{\eta}\in{\cal P}_{c}({\cal Z})$ such that
$S_{\eta}^{\infty}\zeta=\sigma_{\eta}$ for any $\zeta\in{\cal P}_{c}({\cal Z})$.
Hence, $S_{\eta}^{\infty}(S_{\eta}^{n}\zeta)=\sigma_{\eta}$,
as $S_{\eta}^{n}\zeta\in{\cal P}_{c}({\cal Z})$.
Then, (\ref{l2.1.5}), (\ref{l2.1.25}) imply
\begin{align}\label{l2.1.27}
	\left\|
	S_{\eta}^{n}\zeta
	-
	\sigma_{\eta}
	\right\|
	=&
	\left\|
	S_{\eta}^{kn_{0}}(S_{\eta}^{n-kn_{0}}\zeta )
	-
	S_{\eta}^{\infty}(S_{\eta}^{n-kn_{0}}\zeta )
	\right\|
	\nonumber\\
	\leq &
	\frac{1}{2^{k-3}}
	\left\|
	S_{\eta}^{n-kn_{0}}\zeta
	\right\|
	\nonumber\\
	\leq &
	\frac{\|\zeta\|}{2^{k-4}}
	\leq
	C_{2}\gamma_{1}^{n} \|\zeta\|
\end{align}
when $(k+1)n_{0} \geq n>kn_{0}$, 
as $2^{-(k-4)}=16\gamma_{1}^{kn_{0}}
\leq (16\gamma_{1}^{-n_{0}} ) \gamma_{1}^{n}
\leq C_{2}\gamma_{1}^{n}$.
Thus, we get
\begin{align*}
	\left|
	S_{\eta}^{n}
	-
	\sigma_{\eta}
	\right|(z,B)
	=&
	\left|
	S_{\eta}^{n}\delta_{z}
	-
	\sigma_{\eta}
	\right|(B)
	\\
	\leq &
	\left\|
	S_{\eta}^{n}\delta_{z}
	-
	\sigma_{\eta}
	\right\|
	\\
	\leq &
	C_{2}\gamma_{1}^{n} \|\delta_{z}\|
	=
	C_{2}\gamma_{1}^{n} 
\end{align*}
by setting $k=\lfloor (n-m)/n_{0} \rfloor$ in (\ref{l2.1.27}). 
\end{proof}

\section{Results Related to Optimal Filter}\label{section1*}

In this section, an analytic (complex-valued) continuation of the optimal filter
is constructed, and its asymptotic properties (exponential forgetting) are studied.
Here, we rely on the following notations.
${\cal B}({\cal X})$, ${\cal P}({\cal X})$, ${\cal M}_{p}({\cal X})$ and
${\cal M}_{c}({\cal X})$ have been defined at the beginning of Section \ref{section2*}.
For $x\in{\cal X}$, $\xi\in{\cal M}_{c}({\cal X})$,
$\|\xi\|$, $|\xi|({\rm d}x')$ and $\delta_{x}({\rm d}x')$
are the norm and measures specified at the beginning of Section \ref{section2*}.
For $\gamma\in(0,1)$, $V_{\gamma}({\cal P}({\cal X}) )$ is
the open $\gamma$-vicinity of ${\cal P}({\cal X})$, i.e.,
\begin{align*}
	V_{\gamma}({\cal P}({\cal X}) )
	=
	\{\xi\in{\cal M}_{c}({\cal X}): \exists\lambda\in{\cal P}({\cal X}), \|\xi-\lambda\|<\gamma \}.
\end{align*}
$R_{\eta,y}({\rm d}x|\xi)$ is the measure defined by
\begin{align}\label{2.51*}
	R_{\eta,y}(B|\xi)
	=
	\int\int I_{B}(x') \tilde{r}_{\eta}(y,x'|x)\mu({\rm d}x')\xi({\rm d}x)
\end{align}
for $\eta\in\mathbb{C}^{d}$, $\xi\in{\cal M}_{c}({\cal X})$, $B\in{\cal B}({\cal X})$,
$y\in{\cal Y}$ ($\tilde{r}_{\eta}(y,x'|x)$ is specified in Lemma \ref{lemma2.4}).
$\Phi_{\eta,y}(\xi)$ is the function defined by
\begin{align}\label{2.7*}
	\Phi_{\eta,y}(\xi)
	=
	\begin{cases}
	\log R_{\eta,y}({\cal X}|\xi),
	&\text{ if } R_{\eta,y}({\cal X}|\xi)\neq 0
	\\
	0, &\text{ otherwise }
	\end{cases}.
\end{align}
$v_{\eta,\boldsymbol y}^{m:n}(x_{m:n})$ and $\varphi_{\eta,\boldsymbol y}^{m:n}$ are the functions
defined by
\begin{align}
	\label{2.53'*}
	&
	v_{\eta,\boldsymbol y}^{m:n}(x_{m:n})
	=
	\prod_{k=m+1}^{n} \tilde{r}_{\eta}(y_{k},x_{k}|x_{k-1} ),
	\\
	&
	\varphi_{\eta,\boldsymbol y}^{m:n}
	=
	\prod_{k=m+1}^{n} \varphi_{\eta}(y_{k} ),
\end{align}
where $x_{m},\dots,x_{n}\in{\cal X}$, $n>m\geq 0$
and $\boldsymbol y = \{y_{n} \}_{n\geq 1}$ is any sequence in ${\cal Y}$.
$r_{\eta,\boldsymbol y}^{m:n}(x'|x)$ is the function defined by
\begin{align}\label{2.53*}
	r_{\eta,\boldsymbol y}^{m:n}(x'|x)
	=&
	\int\int\cdots\int\int
	v_{\eta,\boldsymbol y}^{m:n}(x_{m:n} )
	\nonumber\\
	&\cdot
	\delta_{x'}({\rm d}x_{n})\mu({\rm d}x_{n-1})\cdots\mu({\rm d}x_{m+1})\delta_{x}({\rm d}x_{m}),
\end{align}
where $x,x'\in{\cal X}$.
$R_{\eta,\boldsymbol y}^{m:m}({\rm d}x|\xi)$ and $R_{\eta,\boldsymbol y}^{m:n}({\rm d}x|\xi)$ are the measures defined by
$R_{\eta,\boldsymbol y}^{m:m}(B|\xi)=\xi(B)$
and
\begin{align}\label{2.55*}
	R_{\eta,\boldsymbol y}^{m:n}(B|\xi)
	=
	\int\int I_{B}(x') r_{\eta,\boldsymbol y}^{m:n}(x'|x)
	\mu({\rm d}x')\xi({\rm d}x).
\end{align}
$f_{\eta,\boldsymbol y}^{m:n}(x|\xi)$,
$g_{\eta,\boldsymbol y}^{m:n}(x'|x,\xi)$,
$h_{\eta,\boldsymbol y}^{m:n}(x|x',\xi)$
are the functions defined by
\begin{align}
	&\label{2.57*}
	g_{\eta,\boldsymbol y}^{m:n}(x'|x,\xi)
	\!=\!\!
	\begin{cases}
	r_{\eta,\boldsymbol y}^{m:n}(x'|x) \big/
	R_{\eta,\boldsymbol y}^{m:n}({\cal X}|\xi),
	&\!\!\!\!\text{if } R_{\eta,\boldsymbol y}^{m:n}({\cal X}|\xi)\!\neq\! 0
	\\
	0, &\!\!\!\!\text{otherwise }
	\end{cases}\!,
	\\	
	&\label{2.59*}
	f_{\eta,\boldsymbol y}^{m:n}(x|\xi)
	=
	\int g_{\eta,\boldsymbol y}^{m:n}(x|x'',\xi)\xi({\rm d}x''),
	\\
	&\label{2.61*}
\begin{aligned}[b]	
	h_{\eta,\boldsymbol y}^{m:n}(x'|x,\xi)
	=&
	-
	f_{\eta,\boldsymbol y}^{m:n}(x'|\xi)
	\int g_{\eta,\boldsymbol y}^{m:n}(x''|x,\xi)\mu({\rm d}x'')
	\\
	&+
	g_{\eta,\boldsymbol y}^{m:n}(x'|x,\xi).
\end{aligned}	
\end{align}
$F_{\eta,\boldsymbol y}^{m:m}({\rm d}x|\xi)$ and $F_{\eta,\boldsymbol y}^{m:n}({\rm d}x|\xi)$ are the measures
defined by $F_{\eta,\boldsymbol y}^{m:m}(B|\xi)=\xi(B)$ and
\begin{align}\label{2.5*}
	F_{\eta,\boldsymbol y}^{m:n}(B|\xi)
	=
	\int I_{B}(x) f_{\eta,\boldsymbol y}^{m:n}(x|\xi)\mu({\rm d}x).
\end{align}
Throughout this and later sections,
measures $R_{\eta,\boldsymbol y}^{m:n}({\rm d}x|\xi)$,
$F_{\eta,\boldsymbol y}^{m:n}({\rm d}x|\xi)$
are also denoted by $R_{\eta,\boldsymbol y}^{m:n}(\xi)$,
$F_{\eta,\boldsymbol y}^{m:n}(\xi)$ (short-hand notations),
while $\left\langle R_{\eta,\boldsymbol y}^{m:n}(\xi) \right\rangle$,
$\left\langle F_{\eta,\boldsymbol y}^{m:n}(\xi) \right\rangle$
are defined by
\begin{align}
	&\label{2.5''*}
	\left\langle R_{\eta,\boldsymbol y}^{m:n}(\xi) \right\rangle
	=
	R_{\eta,\boldsymbol y}^{m:n}({\cal X}|\xi),
	\\
	&
	\left\langle F_{\eta,\boldsymbol y}^{m:n}(\xi) \right\rangle
	=
	F_{\eta,\boldsymbol y}^{m:n}({\cal X}|\xi).
	\nonumber
\end{align}

\begin{remark}
When $\theta\!\in\!\Theta$, $\lambda\!\in\!{\cal P}({\cal X})$,
$F_{\theta,\boldsymbol y}^{m:n}(\lambda)$
is the optimal filter for the model
$\left\{ (X_{n}^{\theta,\lambda}, \! Y_{n}^{\theta,\lambda} ) \right\}_{n\geq 0}$, i.e.,
\begin{align*}
	F_{\theta,\boldsymbol y}^{1:n}(B|\lambda)
	=
	P\left(\left. X_{n}^{\theta,\lambda}\in B \right|Y_{1:n}^{\theta,\lambda}=y_{1:n} \right).
\end{align*}
Hence, for $\eta\in\mathbb{C}^{d}$, $\xi\in{\cal M}_{c}({\cal X})$,
$F_{\eta,\boldsymbol y}^{m:n}(\xi)$
can be considered as a complex-valued continuation of the optimal filter.
Consequently, $f_{\theta,\boldsymbol y}^{m:n}(x|\xi)$
can be viewed as a complex-valued continuation of the optimal filtering density.
$h_{\theta,\boldsymbol y}^{m:n}(x'|x,\xi)$ can be described as
the Gateaux derivative of $f_{\theta,\boldsymbol y}^{m:n}(x|\xi)$
with respect to $\xi$
(see (\ref{l1.4.701'}) -- (\ref{l1.4.701'''})).
$h_{\theta,\boldsymbol y}^{m:n}(x'|x,\xi)$ is used to show that
$F_{\eta,\boldsymbol y}^{m:n}(\xi)$ forgets initial condition $\xi$
at an exponential rate (see Lemmas \ref{lemma1.4}, \ref{lemma1.6} and their proofs).
\end{remark}

\begin{lemma}\label{lemma1.21}
Let $\eta$, $\xi$ be any elements of $\mathbb{C}^{d}$, ${\cal M}_{c}({\cal X} )$ (respectively),
while $\boldsymbol y = \{y_{n} \}_{n\geq 1}$ is any sequence in ${\cal Y}$.
Moreover, let $n$, $m$, $k$ be any integers satisfying $n\geq k\geq m$.
Then, the following is true:

(i) $R_{\eta,\boldsymbol y}^{m:n}(\xi) =
R_{\eta,\boldsymbol y}^{k:n}\left(R_{\eta,\boldsymbol y}^{m:k}(\xi) \right)$.
	
(ii) $\left\langle R_{\eta,\boldsymbol y}^{m:n}(\xi) \right\rangle =
\left\langle R_{\eta,\boldsymbol y}^{k:n}\left( F_{\eta,\boldsymbol y}^{m:k}(\xi) \right) \right\rangle
\left\langle R_{\eta,\boldsymbol y}^{m:k}(\xi) \right\rangle$
when $\left\langle R_{\eta,\boldsymbol y}^{m:k}(\xi) \right\rangle \neq 0$.

(iii) $F_{\eta,\boldsymbol y}^{m:n}(\xi) =
F_{\eta,\boldsymbol y}^{k:n}\left( F_{\eta,\boldsymbol y}^{m:k}(\xi) \right)$
when $\left\langle R_{\eta,\boldsymbol y}^{m:k}(\xi) \right\rangle \neq 0$ and
$\left\langle R_{\eta,\boldsymbol y}^{m:n}(\xi) \right\rangle \neq 0$.
\end{lemma}

\begin{proof}
(i) When $k=m$ or $k=n$, (i) is trivially satisfied.
In what follows in this part of the proof, we assume $n>k>m$.

Owing to (\ref{2.51*}), we have
\begin{align*}
	v_{\eta,\boldsymbol y}^{m:n}(x_{m:n})
	=
	v_{\eta,\boldsymbol y}^{k:n}(x_{k:n})
	v_{\eta,\boldsymbol y}^{m:k}(x_{m:k})
\end{align*}
for $x_{m},\dots,x_{n}\in{\cal X}$.
Combining this with (\ref{2.53*}), it is easy to show
\begin{align*}
	r_{\eta,\boldsymbol y}^{m:n}(x'|x)
	=
	\int r_{\eta,\boldsymbol y}^{k:n}(x'|x'')
	r_{\eta,\boldsymbol y}^{m:k}(x''|x) \mu({\rm d}x'')
\end{align*}
for $x,x'\in{\cal X}$.
Then, using (\ref{2.55*}), we conclude
\begin{align*}
	R_{\eta, \boldsymbol y}^{m:n}(B|\xi)
	=&
	\int\int\int I_{B}(x') r_{\eta,\boldsymbol y}^{k:n}(x'|x'')
	r_{\eta,\boldsymbol y}^{m:k}(x''|x)
	\\
	&\cdot
	\mu({\rm d}x'')\mu({\rm d}x')\xi({\rm d}x)
	\\
	=&
	\int\int I_{B}(x') r_{\eta,\boldsymbol y}^{k:n}(x'|x'')
	R_{\eta,\boldsymbol y}^{m:k}({\rm d}x''|\xi) \mu({\rm d}x')
	\\
	=&
	R_{\eta,\boldsymbol y}^{k:n}\left( B| R_{\eta,\boldsymbol y}^{m:k}(\xi) \right)
\end{align*}
for $B\in{\cal B}({\cal X} )$.
Hence, (i) holds when $n>k>m$.

(ii) We assume $\left\langle R_{\eta,\boldsymbol y}^{m:k}(\xi) \right\rangle \neq 0$
(i.e., $R_{\eta,\boldsymbol y}^{m:k}({\cal X}|\xi) \neq 0$).
Then, using (\ref{2.57*}), (\ref{2.59*}), (\ref{2.5*}), we conclude
\begin{align}\label{l1.21.1}
	F_{\eta,\boldsymbol y}^{m:k}(\xi)
	=
	\frac{R_{\eta,\boldsymbol y}^{m:k}(\xi)}{\left\langle R_{\eta,\boldsymbol y}^{m:k}(\xi) \right\rangle }.
\end{align}
Since $\left\langle R_{\eta,\boldsymbol y}^{k:n}(\xi) \right\rangle$ is linear in $\xi$,
we deduce
\begin{align*}
	\left\langle R_{\eta,\boldsymbol y}^{k:n}\left( F_{\eta,\boldsymbol y}^{m:k}(\xi) \right) \right\rangle
	=
	\frac{\left\langle R_{\eta,\boldsymbol y}^{k:n}\left( R_{\eta,\boldsymbol y}^{m:k}(\xi) \right) \right\rangle }
	{\left\langle R_{\eta,\boldsymbol y}^{m:k}(\xi) \right\rangle }.
\end{align*}
Combining this with (i), we get
\begin{align}\label{l1.21.3}
	\left\langle R_{\eta,\boldsymbol y}^{m:n}(\xi) \right\rangle
	=&
	\left\langle R_{\eta,\boldsymbol y}^{k:n}\left( R_{\eta,\boldsymbol y}^{m:k}(\xi) \right) \right\rangle
	\nonumber\\
	=&
	\left\langle R_{\eta,\boldsymbol y}^{k:n}\left( F_{\eta,\boldsymbol y}^{m:k}(\xi) \right) \right\rangle
	\left\langle R_{\eta,\boldsymbol y}^{m:k}(\xi) \right\rangle.
\end{align}
Thus, (ii) is true.

(iii) We assume $\left\langle R_{\eta,\boldsymbol y}^{m:k}(\xi) \right\rangle \neq 0$,
$\left\langle R_{\eta,\boldsymbol y}^{m:n}(\xi) \right\rangle \neq 0$.
Therefore, (ii) implies
$\left\langle R_{\eta,\boldsymbol y}^{k:n}\left( F_{\eta,\boldsymbol y}^{m:k}(\xi) \right) \right\rangle \neq 0$.
Then, using the same arguments as in (ii), we deduce
\begin{align*}
	&
	F_{\eta,\boldsymbol y}^{m:n}(\xi)
	=
	\frac{R_{\eta,\boldsymbol y}^{m:n}(\xi)}{\left\langle R_{\eta,\boldsymbol y}^{m:n}(\xi) \right\rangle },
	\\
	&
	F_{\eta,\boldsymbol y}^{k:n}\left( F_{\eta,\boldsymbol y}^{m:k}(\xi) \right)
	=
	\frac{R_{\eta,\boldsymbol y}^{k:n}\left( F_{\eta,\boldsymbol y}^{m:k}(\xi) \right) }
	{\left\langle R_{\eta,\boldsymbol y}^{k:n}\left( F_{\eta,\boldsymbol y}^{m:k}(\xi) \right) \right\rangle }.
\end{align*}
Combining this with (i) and (\ref{l1.21.1}), (\ref{l1.21.3}), we get
\begin{align*}
	F_{\eta,\boldsymbol y}^{m:n}(\xi)
	=&
	\frac{R_{\eta,\boldsymbol y}^{k:n}\left(R_{\eta,\boldsymbol y}^{m:k}(\xi) \right)}
	{\left\langle R_{\eta,\boldsymbol y}^{k:n}\left( F_{\eta,\boldsymbol y}^{m:k}(\xi) \right) \right\rangle
	\left\langle R_{\eta,\boldsymbol y}^{m:k}(\xi) \right\rangle }
	\\
	=&
	\frac{R_{\eta,\boldsymbol y}^{k:n}\left(F_{\eta,\boldsymbol y}^{m:k}(\xi) \right)}
	{\left\langle R_{\eta,\boldsymbol y}^{k:n}\left( F_{\eta,\boldsymbol y}^{m:k}(\xi) \right) \right\rangle }
	=
	F_{\eta,\boldsymbol y}^{k:n}\left( F_{\eta,\boldsymbol y}^{m:k}(\xi) \right)
\end{align*}
by using again the fact that $R_{\eta,\boldsymbol y}^{k:n}(\xi)$ is linear in $\xi$.
Hence, (iii) holds.
\end{proof}

\begin{lemma}\label{lemma1.1}
Let Assumption \ref{a11} hold.
Then, there exist real numbers $\delta_{3}\in(0,\delta_{1}]$,
$\gamma_{2}\in(0,1)$, $C_{3}\in[1,\infty)$ such that
\begin{align*}
	\left\|
	F_{\theta,\boldsymbol y}^{m:n}(\lambda' )
	-
	F_{\theta,\boldsymbol y}^{m:n}(\lambda'' )
	\right\|
	\leq
	C_{3}\gamma_{2}^{n-m}
	\left\|\lambda'-\lambda''\right\|
\end{align*}
for all $\theta\in\Theta$,
$\lambda',\lambda''\in V_{\delta_{3} }({\cal P}({\cal X}) )\cap {\cal M}_{p}({\cal X})$,
$n\geq m\geq 0$
and any sequence $\boldsymbol y = \{y_{n} \}_{n\geq 1}$ in ${\cal Y}$
($\delta_{1}$ is specified in Lemma \ref{lemma2.4}).
\end{lemma}

\begin{proof}
Due to \cite[Proposition 4.1, Corollary 4.2]{legland&oudjane}
(or \cite[Theorem 3.1]{tadic&doucet1}) and Lemma \ref{lemma2.4},
there exist real numbers
$\gamma_{2}\in(0,1)$, $C_{3}\in[1,\infty)$ such that
\begin{align}\label{l1.1.1}
	\left\|
	F_{\theta,\boldsymbol y}^{m:n}(\lambda')
	-
	F_{\theta,\boldsymbol y}^{m:n}(\lambda'')
	\right\|
	\leq
	\frac{C_{3}\gamma_{2}^{n-m} }{4}
	\left\|
	\frac{\lambda'}{\|\lambda'\|}
	-
	\frac{\lambda''}{\|\lambda''\|}
	\right\|
\end{align}
for all $\theta\in\Theta$, $\lambda',\lambda''\in{\cal M}_{p}({\cal X})$,
$n\geq m\geq 0$ and any sequence $\boldsymbol y = \{y_{n} \}_{n\geq 1}$ in ${\cal Y}$. 
We have used here the identity $F_{\theta,\boldsymbol y}^{m:n}(\lambda)=
F_{\theta,\boldsymbol y}^{m:n}(\lambda/\|\lambda\|)$ for
$\lambda\in{\cal M}_{p}({\cal X})$.

Let $\delta_{3}=\min\{1/2,\delta_{1}\}$,
while $\boldsymbol y = \{y_{n} \}_{n\geq 1}$ is any sequence in ${\cal Y}$.
Then, we have
$\|\lambda\|\geq 1-\delta_{3}\geq 1/2$
for $\lambda\in V_{\delta_{3} }({\cal P}({\cal X}) ) \cap {\cal M}_{p}({\cal X})$.
Consequently, (\ref{l1.1.1}) implies
\begin{align*}
	&
	\left\|
	F_{\theta,\boldsymbol y}^{m:n}(\lambda')
	-
	F_{\theta,\boldsymbol y}^{m:n}(\lambda'')
	\right\|
	\\
	&\leq
	\frac{C_{3}\gamma_{2}^{n-m} }{4}
	\left\|
	\frac{\lambda'-\lambda''}{\|\lambda'\|}
	-
	\frac{\lambda'' (\|\lambda'\|-\|\lambda''\| ) }{\|\lambda'\| \|\lambda''\|}
	\right\|
	\\
	&\leq
	\frac{C_{3}\gamma_{2}^{n-m} \|\lambda'-\lambda''\| }{2\|\lambda'\| }
	\\
	&\leq
	C_{3}\gamma_{2}^{n-m} \|\lambda'-\lambda''\|
\end{align*}
for $\theta\in\Theta$,
$\lambda',\lambda''\in V_{\delta_{3} }({\cal P}({\cal X}) ) \cap {\cal M}_{p}({\cal X})$,
$n\geq m\geq 0$, 
as $2\|\lambda'\|\geq 1$, $\|\lambda''/\|\lambda''\|\|=1$,
$|\|\lambda'\|-\|\lambda''\||\leq\|\lambda'-\lambda''\|$.
\end{proof}

\begin{lemma}\label{lemma1.2}
Let Assumptions \ref{a1} and \ref{a12} hold.
Then, the following is true:
\\
(i) $\left\langle R_{\eta,\boldsymbol y}^{m:n}(\xi) \right\rangle$ is analytic
in $\eta$
for all $\eta\in V_{\delta_{1}}(\Theta)$, $\xi\in V_{\delta_{1}}({\cal P}({\cal X} ) )$,
$n\geq m\geq 0$ and any sequence
$\boldsymbol y = \{y_{n} \}_{n\geq 1}$ in ${\cal Y}$
($\delta_{1}$ is specified in Lemma \ref{lemma2.4}).
\\
(ii) There exists a non-decreasing sequence
$\{L_{n} \}_{n\geq 1}$ in $[1,\infty)$ such that
\begin{align*}
	&
	\left|
	\frac{\left\langle R_{\eta,\boldsymbol y}^{m:n}(\xi) \right\rangle }
	{\varphi_{\eta,\boldsymbol y}^{m:n} }
	\right|
	\leq
	L_{n-m},
	\\
	&
	\left|
	\frac{\left\langle R_{\eta',\boldsymbol y}^{m:n}(\xi') \right\rangle }
	{\varphi_{\eta',\boldsymbol y}^{m:n} }
	-
	\frac{\left\langle R_{\eta'',\boldsymbol y}^{m:n}(\xi'') \right\rangle }
	{\varphi_{\eta'',\boldsymbol y}^{m:n} }
	\right|
	\\
	&\leq
	L_{n-m}
	\left( \|\eta'-\eta''\| + \|\xi'-\xi''\| \right)
\end{align*}
for all $\eta,\eta',\eta''\in V_{\delta_{1}}(\Theta)$,
$\xi,\xi',\xi''\in V_{\delta_{1}}({\cal P}({\cal X} ) )$,
$n>m\geq 0$ and any sequence $\boldsymbol y = \{y_{n} \}_{n\geq 1}$
in ${\cal Y}$.
\\
(iii) There exists a non-increasing sequence
$\{\alpha_{n} \}_{n\geq 1}$ in $(0,\delta_{1}]$
such that
\begin{align*}
	\frac{\text{\rm Re}\left\{
	\left\langle R_{\eta,\boldsymbol y}^{m:n}(\xi) \right\rangle
	\right\} }
	{|\varphi_{\theta,\boldsymbol y}^{m:n} | }
	\geq
	\frac{1}{L_{n-m} }
\end{align*}
for all $\eta\in V_{\alpha_{n-m} }(\Theta)$,
$\xi\in V_{\alpha_{n-m} }({\cal P}({\cal X}) )$,
$n>m\geq 0$
and any sequence $\boldsymbol y = \{y_{n} \}_{n\geq 1}$
in ${\cal Y}$.
\\
(iv) There exists a non-decreasing sequence $\{M_{n} \}_{n\geq 1}$ in $[1,\infty)$
such that
\begin{align*}
	&
	\max\left\{
	\left|f_{\eta,\boldsymbol y}^{m:n}(x|\xi) \right|,
	\left|h_{\eta,\boldsymbol y}^{m:n}(x'|x,\xi) \right|
	\right\}
	\leq
	M_{n-m},
	\\
	&
	\left|
	f_{\eta',\boldsymbol y}^{m:n}(x|\xi')
	-
	f_{\eta'',\boldsymbol y}^{m:n}(x|\xi'')
	\right|
	\\
	&\leq
	M_{n-m}
	(\|\eta'-\eta''\|+\|\xi'-\xi''\| ), 
	\\
	&
	\left|
	h_{\eta',\boldsymbol y}^{m:n}(x'|x,\xi')
	-
	h_{\eta'',\boldsymbol y}^{m:n}(x'|x,\xi'')
	\right|
	\\
	&\leq
	M_{n-m}
	(\|\eta'-\eta''\|+\|\xi'-\xi''\| )
\end{align*}
for all $\eta,\eta',\eta''\in V_{\alpha_{n-m} }(\Theta)$,
$\xi,\xi',\xi''\in V_{\alpha_{n-m} }({\cal P}({\cal X}) )$,
$x,x'\in{\cal X}$,
$n>m\geq 0$
and any sequence $\boldsymbol y = \{y_{n} \}_{n\geq 1}$
in ${\cal Y}$.
\end{lemma}

\begin{proof}
(i) and (ii)
Throughout these parts of the proof, the following notations is used.
$\{\tilde{L}_{l} \}_{l\geq 1}$, $\{L_{l} \}_{l\geq 1}$ are the real numbers
defined by
\begin{align}\label{l1.2.77}
	\tilde{L}_{l}
	=
	\frac{2^{l+1}d }{\delta_{1} }
	\left(\|\mu\| + \frac{1}{\gamma} \right)^{l},
	\;\;\;\;\;\;\;
	L_{l}
	=
	2\tilde{L}_{l}^{2}
\end{align}
for $l\geq 1$, where $\gamma$, $K_{l}$ are specified in Assumption \ref{a12} and Lemma \ref{lemma2.2}.
$m$, $n$ are any integers satisfying $n>m\geq 0$.
In what follows in the proof of (i), (ii), both $m$, $n$ are kept fixed.
$\eta$, $\eta'$, $\eta''$ are any elements in $V_{\delta_{1} }(\Theta)$.
$\xi$, $\xi'$, $\xi''$ are any elements of $V_{\delta_{1} }({\cal P}({\cal X} ) )$.
$x$, $x'$ are any elements of ${\cal X}$,
while $\boldsymbol y = \{y_{n} \}_{n\geq 0}$ is any sequence in ${\cal Y}$.

Using (\ref{2.53*}), (\ref{2.55*}), (\ref{2.5''*}), it is straightforward to verify
\begin{align}
	\label{l1.2.7}
	\frac{\left\langle R_{\eta,\boldsymbol y}^{m:n}(\xi) \right\rangle }
	{\varphi_{\eta,\boldsymbol y}^{m:n} }
	=&
	\int\cdots\int\int
	\frac{v_{\eta,\boldsymbol y}^{m:n}(x_{m:n}) }
	{\varphi_{\eta,\boldsymbol y}^{m:n} }
	\nonumber\\
	&\cdot
	\mu({\rm d}x_{n})\cdots\mu({\rm d}x_{m+1} )\xi({\rm d}x_{m}).
\end{align}
Moreover, Lemma \ref{lemma2.4} yields
\begin{align}\label{l1.2.7'}
	\left|
	\frac{v_{\eta,\boldsymbol y}^{m:n}(x_{m:n}) }
	{\varphi_{\eta,\boldsymbol y}^{m:n} }
	\right|
	=
	\prod_{k=m+1}^{n}
	\left|
	\frac{\tilde{r}_{\eta}(y_{k},x_{k}|x_{k-1} ) }{\varphi_{\eta}(y_{k} ) }
	\right|
	\leq
	2^{n-m}.
\end{align}
Since $v_{\eta,\boldsymbol y}^{m:n}(x_{m:n}) / \varphi_{\eta,\boldsymbol y}^{m:n}$ is analytic in $\eta$
for each $\eta\in V_{\delta_{1}}(\Theta)$, $x_{m},\dots,x_{n}\in{\cal X}$
(due to Assumption \ref{a1}),
Lemma \ref{lemmaa1} (see Appendix \ref{appendix1}) and (\ref{l1.2.7}), (\ref{l1.2.7'})
imply that
$\left\langle R_{\eta,\boldsymbol y}^{m:n}(\xi) \right\rangle / \varphi_{\eta,\boldsymbol y}^{m:n}$
is analytic in $\eta$ for all $\eta\in V_{\delta_{1}}(\Theta)$.
Consequently, $\left\langle R_{\eta,\boldsymbol y}^{m:n}(\xi) \right\rangle$
is analytic in $\eta$ for each $\eta\in V_{\delta_{1}}(\Theta)$.
Hence, (i) holds.

Owing to (\ref{l1.2.7}), (\ref{l1.2.7'}), we have
\begin{align}\label{l1.2.75}
	\left|
	\frac{\left\langle R_{\eta,\boldsymbol y}^{m:n}(\xi) \right\rangle }
	{\varphi_{\eta,\boldsymbol y}^{m:n} }
	\right|
	\leq &
	\int\cdots\int\int
	\left|
	\frac{v_{\eta,\boldsymbol y}^{m:n}(x_{m:n}) }
	{\varphi_{\eta,\boldsymbol y}^{m:n} }
	\right|
	\nonumber\\
	&\cdot
	\mu({\rm d}x_{n})\cdots\mu({\rm d}x_{m+1} )|\xi|({\rm d}x_{m})
	\nonumber\\
	\leq &
	2^{n-m} \|\mu\|^{n-m} \|\xi\|
	\nonumber\\
	\leq &
	\tilde{L}_{n-m}
	\leq 
	L_{n-m}
\end{align}
as $\xi\in V_{\delta_{1}}({\cal P}({\cal X} ) )$ results in $\|\xi\|\leq 1+\delta_{1} \leq 2$.
Using similar arguments, we get
\begin{align}\label{l1.2.75'}
	\left|
	\frac{\left\langle R_{\eta,\boldsymbol y}^{m:n}(\xi') \right\rangle }
	{\varphi_{\eta,\boldsymbol y}^{m:n} }
	-
	\frac{\left\langle R_{\eta,\boldsymbol y}^{m:n}(\xi'') \right\rangle }
	{\varphi_{\eta,\boldsymbol y}^{m:n} }
	\right|
	= &
	\left|
	\frac{\left\langle R_{\eta,\boldsymbol y}^{m:n}(\xi'-\xi'') \right\rangle }
	{\varphi_{\eta,\boldsymbol y}^{m:n} }
	\right|
	\nonumber\\
	\leq &
	2^{n-m} \|\mu\|^{n-m} \|\xi'-\xi''\|
	\nonumber\\
	\leq &
	\tilde{L}_{n-m} \|\xi'-\xi''\|.
\end{align}
Since $\left\langle R_{\eta,\boldsymbol y}^{m:n}(\xi) \right\rangle / \varphi_{\eta,\boldsymbol y}^{m:n}$ is analytic in $\eta$
for each $\eta\in V_{\delta_{1}}(\Theta)$,
Lemma \ref{lemmaa1} and (\ref{l1.2.75}) imply
\begin{align*}
	\left|
	\frac{\left\langle R_{\eta',\boldsymbol y}^{m:n}(\xi) \right\rangle }
	{\varphi_{\eta'',\boldsymbol y}^{m:n} }
	-
	\frac{\left\langle R_{\eta'',\boldsymbol y}^{m:n}(\xi) \right\rangle }
	{\varphi_{\eta'',\boldsymbol y}^{m:n} }
	\right|
	\leq &
	\frac{2^{n-m} d\|\mu\|^{n-m} \|\xi\| \|\eta'-\eta''\| }{\delta_{1} }
	\nonumber\\
	\leq &
	\tilde{L}_{n-m} \|\eta'-\eta''\|.
\end{align*}
Then, we have
\begin{align}\label{l1.2.25}
	&
	\left|
	\frac{\left\langle R_{\eta',\boldsymbol y}^{m:n}(\xi') \right\rangle }
	{\varphi_{\eta'',\boldsymbol y}^{m:n} }
	-
	\frac{\left\langle R_{\eta'',\boldsymbol y}^{m:n}(\xi'') \right\rangle }
	{\varphi_{\eta'',\boldsymbol y}^{m:n} }
	\right|
	\nonumber\\
	&
	\begin{aligned}
	\leq &
	\left|
	\frac{\left\langle R_{\eta',\boldsymbol y}^{m:n}(\xi') \right\rangle }
	{\varphi_{\eta'',\boldsymbol y}^{m:n} }
	-
	\frac{\left\langle R_{\eta'',\boldsymbol y}^{m:n}(\xi') \right\rangle }
	{\varphi_{\eta'',\boldsymbol y}^{m:n} }
	\right|
	\\
	&+
	\left|
	\frac{\left\langle R_{\eta'',\boldsymbol y}^{m:n}(\xi') \right\rangle }
	{\varphi_{\eta'',\boldsymbol y}^{m:n} }
	-
	\frac{\left\langle R_{\eta'',\boldsymbol y}^{m:n}(\xi'') \right\rangle }
	{\varphi_{\eta'',\boldsymbol y}^{m:n} }
	\right|
	\end{aligned}
	\nonumber\\
	\nonumber\\
	&\leq
	\tilde{L}_{n-m}\left( \|\eta'-\eta''\| + \|\xi'-\xi''\| \right).
\end{align}
Using (\ref{l1.2.75}), (\ref{l1.2.25}), we conclude that (ii) is true.

(iii) and (iv)
Throughout these parts of the proof, we use the following notations.
$\{\tilde{L}_{l} \}_{l\geq 1}$ has the same meaning as in (\ref{l1.2.77}),
while $\{\alpha_{l} \}_{l\geq 1}$,
$\{\tilde{M}_{l} \}_{l\geq 1}$, $\{M_{l} \}_{l\geq 1}$ are the numbers defined by
\begin{align*}
	\alpha_{l}
	=
	\frac{\delta_{1}}{4\tilde{L}_{l}^{2} },
	\;\;\;\;\;\;\;
	\tilde{M}_{l}
	=
	10\tilde{L}_{l}^{4},
	\;\;\;\;\;\;\;
	M_{l}=5\tilde{M}_{l}^{2}(\|\mu\|+1).
\end{align*}
$m$, $n$ are any integers satisfying $n>m\geq 0$.
In what follows in the proof of (iii), (iv), both $m$, $n$ are kept fixed.
$\eta$, $\eta'$, $\eta''$ are any elements of $V_{\alpha_{n-m} }(\Theta)$,
while $\theta$ is any element of $\Theta$ satisfying
$\|\eta-\theta\|<\alpha_{n-m}$.
$\xi$, $\xi'$, $\xi''$ are any elements of $V_{\alpha_{n-m} }({\cal P}({\cal X} ) )$,
while $\lambda$ is any element of ${\cal P}({\cal X} )$ satisfying
$\|\xi-\lambda\|<\alpha_{n-m}$.
$x$, $x'$ are any elements of ${\cal X}$,
while $\boldsymbol y = \{y_{n} \}_{n\geq 1}$ is any sequence in ${\cal Y}$.

Using Lemma \ref{lemma2.4} and (\ref{2.53*}), it is straightforward to verify
\begin{align*}
	\big\langle
	R_{\theta,\boldsymbol y}^{m:k+1}(\lambda)
	\big\rangle
	\!=\!\!
	\begin{aligned}[t]
	&
	\int\!\cdots\!\int\!\int\!
	\left(
	\int\! r_{\theta}(y_{k+1},x_{k+1}|x_{k}) \mu({\rm d}x_{k+1})
	\right)
	\\
	&\cdot
	v_{\theta,\boldsymbol y}^{m:k}(x_{m:k}) \mu({\rm d}x_{k})\cdots\mu({\rm d}x_{m+1})\lambda({\rm d}x_{m})
	\end{aligned}
\end{align*}
for $k>m$.
Consequently, Assumption \ref{a12} yields
\begin{align}\label{l1.2.901}
	\big\langle
	R_{\theta,\boldsymbol y}^{m:k+1}(\lambda)
	\big\rangle
	\geq &
	\begin{aligned}[t]
	&
	\gamma|\varphi_{\theta}(y_{k+1}) |
	\int\cdots\int\int
	v_{\theta,\boldsymbol y}^{m:k}(x_{m:k})
	\\
	&\;\;\;\cdot
	\mu({\rm d}x_{k})\cdots\mu({\rm d}x_{m+1})\lambda({\rm d}x_{m})
	\end{aligned}
	\nonumber\\
	=&
	\gamma|\varphi_{\theta}(y_{k+1}) |
	\left\langle
	R_{\theta,\boldsymbol y}^{m:k}(\lambda)
	\right\rangle.
\end{align}
The same arguments also imply
\begin{align*}
	\big\langle
	R_{\theta,\boldsymbol y}^{m:m+1}(\lambda)
	\big\rangle
	\!=\!\! &
	\int\!\!\left(\!\int\! r_{\theta}(y_{m+1},x_{m+1}|x_{m})\mu({\rm d}x_{m+1})\!\right)
	\!\!\lambda({\rm d}x_{m})
	\\
	\geq &
	\gamma|\varphi_{\theta}(y_{m+1}) |
	\|\lambda\|
	=
	\gamma|\varphi_{\theta}(y_{m+1}) |.
\end{align*}
Then, iterating (\ref{l1.2.901}), we get
\begin{align*}
	\big\langle
	R_{\theta,\boldsymbol y}^{m:k+1}(\lambda)
	\big\rangle
	\geq &
	\gamma^{k-m-1}
	\left(\prod_{l=m+2}^{k+1}|\varphi_{\theta}(y_{l}) | \right)
	\big\langle
	R_{\theta,\boldsymbol y}^{m:m+1}(\lambda)
	\big\rangle
	\\
	\geq &
	\gamma^{k-m}
	\left(\prod_{l=m+1}^{k+1}|\varphi_{\theta}(y_{l}) | \right)
	=
	\gamma^{k-m}|\varphi_{\theta,\boldsymbol y}^{m:n}|.
\end{align*}
Hence, we have
\begin{align*}
	\frac{\big\langle
	R_{\theta,\boldsymbol y}^{m:n}(\lambda)
	\big\rangle }
	{|\varphi_{\theta,\boldsymbol y}^{m:n}|}
	\geq
	\frac{1}{\tilde{L}_{n-m}}
\end{align*}
as $\tilde{L}_{n-m}\geq\gamma^{-(n-m)}$.
Combining this with (\ref{l1.2.25}), we get
\begin{align}\label{l1.2.27}
	\frac{\text{\rm Re}\left\{
	\left\langle R_{\eta,\boldsymbol y}^{m:n}(\xi) \right\rangle
	\right\} }
	{|\varphi_{\eta,\boldsymbol y}^{m:n}|}
	\!\geq\! &
	\frac{\big\langle
	R_{\theta,\boldsymbol y}^{m:n}(\lambda)
	\big\rangle }
	{|\varphi_{\theta,\boldsymbol y}^{m:n}|}
	\!-\!
	\left|
	\frac{\left\langle R_{\eta,\boldsymbol y}^{m:n}(\xi) \right\rangle }
	{\varphi_{\eta,\boldsymbol y}^{m:n}}
	\!-\!
	\frac{\big\langle R_{\theta,\boldsymbol y}^{m:n}(\lambda) \big\rangle }
	{\varphi_{\theta,\boldsymbol y}^{m:n} }
	\right|
	\nonumber\\
	\geq &
	\frac{1}{\tilde{L}_{n-m} }
	-
	\tilde{L}_{n-m} \left(\|\eta-\theta\| + \|\xi-\lambda \| \right)
	\nonumber\\
	\geq &
	\frac{1}{\tilde{L}_{n-m} }
	-
	2\tilde{L}_{n-m} \alpha_{n-m}
	\nonumber\\
	\geq &
	\frac{1}{2\tilde{L}_{n-m} }
	\geq
	\frac{1}{2\tilde{L}_{n-m} }
\end{align}
as $\|\eta-\theta\|<\alpha_{n-m}$, $\|\xi-\lambda\|<\alpha_{n-m}$.

Using (\ref{2.53*}), it is straightforward to verify
\begin{align}
	\label{l1.2.501'}
	\frac{r_{\eta,\boldsymbol y}^{m:n}(x'|x) }
	{\varphi_{\eta,\boldsymbol y}^{m:n} }
	\!=\!&
	\int\!\int\!\cdots\!\int\!\int
	\frac{v_{\eta,\boldsymbol y}^{m:n}(x_{m:n}) }
	{\varphi_{\eta,\boldsymbol y}^{m:n} }
	\nonumber\\
	&\cdot
	\delta_{x'}({\rm d}x_{n})\mu({\rm d}x_{n-1})\cdots\mu({\rm d}x_{m+1} )\delta_{x}({\rm d}x_{m}).
\end{align}
Consequently, Assumption \ref{a1}, Lemma \ref{lemmaa1} and (\ref{l1.2.7'})
imply that $r_{\eta,\boldsymbol y}^{m:n}(x'|x) / \varphi_{\eta,\boldsymbol y}^{m:n}$
is analytic in $\eta$ for each $\eta\in V_{\delta_{1}}(\Theta)$.
Therefore, $r_{\eta,\boldsymbol y}^{m:n}(x'|x)$
is analytic in $\eta$ for all $\eta\in V_{\delta_{1}}(\Theta)$.
Since $\left\langle R_{\eta,\boldsymbol y}^{m:n}(\xi) \right\rangle$
is non-zero and analytic in $\eta$ for all $\eta\in V_{\delta_{1}}(\Theta)$,
we then conclude from (\ref{2.57*}), (\ref{l1.2.27}) that $g_{\eta,\boldsymbol y}^{m:n}(x'|x,\xi)$
is analytic in $\eta$ for all $\eta\in V_{\delta_{1}}(\Theta)$.

Owing to (\ref{l1.2.7'}), (\ref{l1.2.501'}), we have
\begin{align}\label{l1.2.501}
	\left|
	\frac{r_{\eta,\boldsymbol y}^{m:n}(x'|x) }
	{\varphi_{\eta,\boldsymbol y}^{m:n} }
	\right|
	\!\leq\! &
	\int\!\int\!\cdots\!\int\!\int
	\left|
	\frac{v_{\eta,\boldsymbol y}^{m:n}(x_{m:n}) }
	{\varphi_{\eta,\boldsymbol y}^{m:n} }
	\right|
	\nonumber\\
	&\cdot
	\delta_{x'}({\rm d}x_{n})\mu({\rm d}x_{n-1})\cdots\mu({\rm d}x_{m+1} )\delta_{x}({\rm d}x_{m})
	\nonumber\\
	\leq &
	2^{n-m}\|\delta_{x}\|\|\delta_{x'}\|\|\mu\|^{n-m-1}
	\leq 
	\tilde{L}_{n-m}.
\end{align}
Then, (\ref{2.57*}), (\ref{l1.2.27}) imply
\begin{align}\label{l1.2.41}
	\left|
	g_{\eta,\boldsymbol y}^{m:n}(x'|x,\xi)
	\right|
	=
	\left|
	\frac{r_{\eta,\boldsymbol y}^{m:n}(x'|x) }
	{\left\langle R_{\eta,\boldsymbol y}^{m:n}(\xi) \right\rangle }
	\right|
	\leq
	2 \tilde{L}_{n-m}^{2}
	\leq
	\tilde{M}_{n-m}
\end{align}
as $\left|\left\langle R_{\eta,\boldsymbol y}^{m:n}(\xi) \right\rangle \right| \geq
\text{\rm Re}\left(\left\langle R_{\eta,\boldsymbol y}^{m:n}(\xi) \right\rangle \right)
>0$.
Consequently, (\ref{2.59*}) yields
\begin{align}\label{l1.2.43}
	\left|
	f_{\eta,\boldsymbol y}^{m:n}(x|\xi)
	\right|
	\leq &
	\int
	\left|
	g_{\eta,\boldsymbol y}^{m:n}(x|x',\xi)
	\right|
	|\xi|({\rm d}x')
	\nonumber\\
	\leq &
	\tilde{M}_{n-m} \|\xi\|
	\leq
	2 \tilde{M}_{n-m}
	\leq
	M_{n-m}
\end{align}
as $\xi\in V_{\alpha_{n-m} }({\cal P}({\cal X}) )$ results in $\|\xi\|\leq 1+\alpha_{n-m}\leq 2$.
Similarly, we have
\begin{align}\label{l1.2.45}
	\int
	\left|
	g_{\eta,\boldsymbol y}^{m:n}(x'|x,\xi)
	\right|
	\mu({\rm d}x')
	\leq
	\tilde{M}_{n-m} \|\mu\|.
\end{align}
Combining this with (\ref{2.61*}), (\ref{l1.2.41}), (\ref{l1.2.43}), we get
\begin{align}\label{l1.2.71}
	\left|
	h_{\eta,\boldsymbol y}^{m:n}(x'|x,\xi)
	\right|
	\leq &
	\left|
	f_{\eta,\boldsymbol y}^{m:n}(x'|\xi)
	\right|
	\int
	\left|
	g_{\eta,\boldsymbol y}^{m:n}(x''|x,\xi)
	\right|
	\mu({\rm d}x'')
	\nonumber\\
	&+
	\left|
	g_{\eta,\boldsymbol y}^{m:n}(x'|x,\xi)
	\right|
	\nonumber\\
	\leq &
	\tilde{M}_{n-m} + 2\tilde{M}_{n-m}^{2}\|\mu\|
	\leq 
	M_{n-m}.
\end{align}

Since $g_{\eta,\boldsymbol y}^{m:n}(x'|x,\xi)$ is analytic in $\eta$
for each $\eta\in V_{\alpha_{n-m}}(\Theta)$,
Lemma \ref{lemmaa1} and (\ref{l1.2.41}) imply
\begin{align}\label{l1.2.525'}
	\left|
	g_{\eta',\boldsymbol y}^{m:n}(x'|x,\xi)
	-
	g_{\eta'',\boldsymbol y}^{m:n}(x'|x,\xi)
	\right|
	\leq &
	\frac{2d\tilde{L}_{n-m}^{2}\|\eta'-\eta''\|}{\alpha_{n-m}}
	\nonumber\\
	\leq &
	\tilde{M}_{n-m}\|\eta'-\eta''\|.
\end{align}
Moreover, (\ref{l1.2.75'}), (\ref{l1.2.27}), (\ref{l1.2.41}) yield
\begin{align}\label{l1.2.525''}
	&
	\left|
	g_{\eta,\boldsymbol y}^{m:n}(x'|x,\xi')
	-
	g_{\eta,\boldsymbol y}^{m:n}(x'|x,\xi'')
	\right|
	\nonumber\\
	&=
	\left|
	g_{\eta,\boldsymbol y}^{m:n}(x'|x,\xi')
	\right|
	\left|
	\frac{\left\langle R_{\eta,\boldsymbol y}^{m:n}(\xi') \right\rangle
	-
	\left\langle R_{\eta,\boldsymbol y}^{m:n}(\xi'') \right\rangle }
	{\left\langle R_{\eta,\boldsymbol y}^{m:n}(\xi'') \right\rangle }
	\right|
	\nonumber\\
	&\leq
	2\tilde{L}_{n-m}^{4}\|\xi'-\xi''\|
	\leq 
	\tilde{M}_{n-m}\|\xi'-\xi''\|.
\end{align}
Combining (\ref{l1.2.525'}), (\ref{l1.2.525''}), we get
\begin{align}\label{l1.2.525}
	&
	\left|
	g_{\eta',\boldsymbol y}^{m:n}(x'|x,\xi')
	-
	g_{\eta'',\boldsymbol y}^{m:n}(x'|x,\xi'')
	\right|
	\nonumber\\
	&
	\begin{aligned}
	\leq &
	\left|
	g_{\eta',\boldsymbol y}^{m:n}(x'|x,\xi')
	-
	g_{\eta'',\boldsymbol y}^{m:n}(x'|x,\xi')
	\right|
	\\
	&+
	\left|
	g_{\eta'',\boldsymbol y}^{m:n}(x'|x,\xi')
	-
	g_{\eta'',\boldsymbol y}^{m:n}(x'|x,\xi'')
	\right|
	\end{aligned}
	\nonumber\\
	&\leq
	\tilde{M}_{n-m} (\|\eta'-\eta''\| + \|\xi'-\xi''\| ).
\end{align}
Consequently, (\ref{2.59*}), (\ref{l1.2.41}) imply
\begin{align}\label{l1.2.73}
	&
	\left|
	f_{\eta',\boldsymbol y}^{m:n}(x|\xi')
	-
	f_{\eta'',\boldsymbol y}^{m:n}(x|\xi'')
	\right|
	\nonumber\\
	&
	\begin{aligned}
	\leq &
	\int
	\left|
	g_{\eta',\boldsymbol y}^{m:n}(x|x',\xi')
	-
	g_{\eta'',\boldsymbol y}^{m:n}(x|x',\xi'')
	\right|
	|\xi'|({\rm d}x')
	\\
	&+
	\int
	\left|
	g_{\eta'',\boldsymbol y}^{m:n}(x|x',\xi'')
	\right|
	|\xi'-\xi''|({\rm d}x')
	\end{aligned}
	\nonumber\\
	&\leq
	\tilde{M}_{n-m} \|\xi'\| (\|\eta'-\eta''\|+\|\xi'-\xi''\|)
	+
	\tilde{M}_{n-m} \|\xi'-\xi''\|
	\nonumber\\
	&\leq
	3\tilde{M}_{n-m} (\|\eta'-\eta''\|+\|\xi'-\xi''\|)
	\nonumber\\
	&\leq
	M_{n-m} (\|\eta'-\eta''\|+\|\xi'-\xi''\|)
\end{align}
as $\|\xi'\|\leq 1+\alpha_{n-m}\leq 2$.
Similarly, we get
\begin{align*}
	&
	\int
	\left|
	g_{\eta',\boldsymbol y}^{m:n}(x'|x,\xi')
	-
	g_{\eta'',\boldsymbol y}^{m:n}(x'|x,\xi'')
	\right|
	\mu({\rm d}x')
	\\
	&\leq
	\tilde{M}_{n-m} \|\mu\| (\|\eta'-\eta''\|+\|\xi'-\xi''\|).
\end{align*}
Combining this with (\ref{2.61*}), (\ref{l1.2.43}), (\ref{l1.2.45}), (\ref{l1.2.525}), (\ref{l1.2.73}),
we get
\begin{align}\label{l1.2.91}
	&
	\left|
	h_{\eta',\boldsymbol y}^{m:n}(x'|x,\xi')
	-
	h_{\eta'',\boldsymbol y}^{m:n}(x'|x,\xi'')
	\right|
	\nonumber\\
	&
	\begin{aligned}
	\leq &
	\left|
	g_{\eta',\boldsymbol y}^{m:n}(x'|x,\xi')
	-
	g_{\eta'',\boldsymbol y}^{m:n}(x'|x,\xi'')
	\right|
	\\
	&\!+\!\!
	\left|
	f_{\eta',\boldsymbol y}^{m:n}(x'|\xi')
	-
	f_{\eta'',\boldsymbol y}^{m:n}(x'|\xi'')
	\right|
	\int
	\left|
	g_{\eta',\boldsymbol y}^{m:n}(x''|x,\xi')
	\right|
	\mu({\rm d}x'')
	\\
	&\!+\!\!
	\left|
	f_{\eta'',\boldsymbol y}^{m:n}(x'|\xi'')
	\right|
	\!\int\!
	\left|
	g_{\eta',\boldsymbol y}^{m:n}(x''|x,\xi')
	\!-\!
	g_{\eta'',\boldsymbol y}^{m:n}(x''|x,\xi'')
	\right|
	\mu({\rm d}x'')
	\end{aligned}
	\nonumber\\
	&\leq
	(\tilde{M}_{n-m} + 5\tilde{M}_{n-m}^{2}\|\mu\| )
	(\|\eta'-\eta''\|+\|\xi'-\xi''\| )
	\nonumber\\
	&\leq
	M_{n-m} (\|\eta'-\eta''\|+\|\xi'-\xi''\| ).
\end{align}
Using (\ref{l1.2.27}), (\ref{l1.2.43}), (\ref{l1.2.71}) -- (\ref{l1.2.91}),
we conclude that (iii), (iv) hold.
\end{proof}

\begin{lemma}\label{lemma1.3}
Let Assumptions \ref{a1} and \ref{a12} hold. Then, the following is true:

(i) There exists a real number $\delta_{4}\in(0,\delta_{1}]$ such that
$\text{\rm Re}\left\{ R_{\eta,y}({\cal X}|\xi) \right\}> 0$
for all $\eta\in V_{\delta_{4}}(\Theta)$,
$\xi\in V_{\delta_{4}}({\cal P}({\cal X}) )$,
$y\in{\cal Y}$ ($\delta_{1}$ is specified in Lemma \ref{lemma2.4}).

(ii) There exists a real number $C_{4}\in[1,\infty)$ such that
\begin{align*}
	&
	\left|\Phi_{\eta,y}(\xi) \right|
	\leq
	C_{4}\left(1+\psi(y) \right),
	\\
	&
	\left|
	\Phi_{\eta',y}(\xi')
	-
	\Phi_{\eta'',y}(\xi'')
	\right|
	\\
	&\leq
	C_{4} \left(1+\psi(y) \right) \left(\|\eta'-\eta''\| + \|\xi'-\xi''\| \right)
\end{align*}
for all $\eta,\eta',\eta''\in V_{\delta_{4}}(\Theta)$,
$\xi,\xi',\xi''\in V_{\delta_{4}}({\cal P}({\cal X}) )$,
$y\in{\cal Y}$.
\end{lemma}

\begin{proof}
Throughout the proof, the following notations is used.
$\delta_{4}$, $C_{4}$ are the real numbers defined by
$\delta_{4}=\alpha_{1}$, $C_{4}=4L_{1}^{2}$
($\alpha_{1}$, $L_{1}$ are specified in Lemma \ref{lemma1.2}).
$\eta$, $\eta'$, $\eta''$ are any elements of $V_{\delta_{4}}(\Theta)$,
while $\xi$, $\xi'$, $\xi''$ are any elements in $V_{\delta_{4}}({\cal P}({\cal X}) )$.
$y$ is any element of ${\cal Y}$.

Since $R_{\eta,y}({\cal X}|y) = \left\langle R_{\eta,\boldsymbol y}^{0:1}(\xi) \right\rangle$
for any sequence $\boldsymbol y = \{y_{n} \}_{n\geq 1}$ in ${\cal Y}$
satisfying $y=y_{1}$,
Lemma \ref{lemma1.2} yields
\begin{align}
	&\label{l1.3.1}
	\text{\rm Re}\left\{
	R_{\eta,y}({\cal X}|\xi)
	\right\}
	\geq
	\frac{|\varphi_{\eta}(y)|}{L_{1} },
	\\
	&\label{l1.3.3}
	\left|
	R_{\eta,y}({\cal X}|\xi)
	\right|
	\leq
	L_{1} |\varphi_{\eta}(y)|,
	\\
	&\label{l1.3.5}
	\left|
	R_{\eta,y}({\cal X}|\xi')
	-
	R_{\eta,y}({\cal X}|\xi'')
	\right|
	\leq
	L_{1} |\varphi_{\eta}(y)|
	\|\xi'-\xi''\|.
\end{align}
Due to the same arguments, $R_{\eta,y}({\cal X}|\xi)$ is analytic in $\eta$
for each $\eta\in V_{\delta_{4}}(\Theta)$.
As $\varphi_{\eta}(y)\neq 0$ (owing to Assumption \ref{a1}), (\ref{l1.3.1}) implies that (i) holds.
Consequently, (\ref{2.7*}) yields that
$\Phi_{\eta,y}(\xi)$ is analytic in $\eta$ for all $\eta\in V_{\delta_{4}}(\Theta)$.
Moreover, due to (\ref{l1.3.1}), (\ref{l1.3.3}), we have
\begin{align*}
	&
	\log\left|
	R_{\eta,y}({\cal X}|\xi)
	\right|
	\leq
	\log L_{1} + \log|\varphi_{\eta}(y)|
	\leq
	L_{1}\left( 1+\psi(y) \right),
	\\
	&
	\log\left|
	R_{\eta,y}({\cal X}|\xi)
	\right|
	\geq
	-\log L_{1} + \log|\varphi_{\eta}(y)|
	\geq
	-L_{1}\left( 1+\psi(y) \right).
\end{align*}
Therefore, we get
\begin{align}\label{l1.3.51}
	\left|
	\Phi_{\eta,y}(\xi)
	\right|
	=
	\left|
	\log R_{\eta,y}({\cal X}|\xi)
	\right|
	\leq &
	\left|
	\log\left|
	R_{\eta,y}({\cal X}|\xi)
	\right|
	\right|
	+
	\pi
	\nonumber\\
	\leq &
	4L_{1}\left( 1+\psi(y) \right)
	\nonumber\\
	\leq &
	C_{4}\left( 1+\psi(y) \right).
\end{align}
Then, Lemma \ref{lemmaa1} implies
\begin{align}\label{l1.3.51'}
	\left|
	\Phi_{\eta',y}(\xi) - \Phi_{\eta'',y}(\xi)
	\right|
	\leq &
	\frac{4dL_{1}(1+\psi(y) )\|\eta'-\eta''\| }{\delta_{4} }
	\nonumber\\
	\leq &
	C_{4}(1+\psi(y) )\|\eta'-\eta''\|.
\end{align}

Let $\phi_{\eta,y}(t|\xi',\xi'')$ be the function defined by
\begin{align*}
	\phi_{\eta,y}(t|\xi',\xi'')
	=
	\log\left(
	t R_{\eta,y}({\cal X}|\xi')
	+
	(1-t) R_{\eta,y}({\cal X}|\xi'')
	\right)
\end{align*}
for $t\in[0,1]$.
Due to Assumption \ref{a1} and (\ref{l1.3.1}), we have
\begin{align}\label{l1.3.7}
	&
	\left|
	t R_{\eta,y}({\cal X}|\xi')
	+
	(1-t) R_{\eta,y}({\cal X}|\xi'')
	\right|
	\nonumber\\
	&\geq
	t \text{\rm Re}\left\{ R_{\eta,y}({\cal X}|\xi') \right\}
	+
	(1-t) \text{\rm Re}\left\{ R_{\eta,y}({\cal X}|\xi'') \right\}
	\nonumber\\
	&\geq
	\frac{|\varphi_{\eta}(y)| }{L_{1} }
	>0
\end{align}
for $t\in[0,1]$.
Hence,
$\phi_{\eta,y}(t|\xi',\xi'')$ is well-defined
and differentiable in $t$ for each $t\in[0,1]$.
We also have
\begin{align*}
	\phi'_{\eta,y}(t|\xi',\xi'')
	=&
	\frac{\partial}{\partial t}
	\phi_{\eta,y}(t|\xi',\xi'')
	\\
	=&
	\frac{\text{\rm Re}\left\{ R_{\eta,y}({\cal X}|\xi') \right\}
	-
	\text{\rm Re}\left\{ R_{\eta,y}({\cal X}|\xi'') \right\} }
	{t \text{\rm Re}\left\{ R_{\eta,y}({\cal X}|\xi') \right\}
	+
	(1-t) \text{\rm Re}\left\{ R_{\eta,y}({\cal X}|\xi'') \right\} }. 
\end{align*}
Consequently, (\ref{l1.3.5}), (\ref{l1.3.7}) yield
\begin{align*}
	\left|
	\phi'_{\eta,y}(t|\xi',\xi'')
	\right|
	\leq
	L_{1}^{2} \|\xi'-\xi''\|.
\end{align*}
Thus, we get
\begin{align*}
	\left|
	\Phi_{\eta,y}(\xi')
	-
	\Phi_{\eta,y}(\xi'')
	\right|
	=&
	\left|
	\phi_{\eta,y}(1|\xi',\xi'')
	-
	\phi_{\eta,y}(0|\xi',\xi'')
	\right|
	\\
	=&
	\left|
	\int_{0}^{1}
	\phi'_{\eta,y}(t|\xi',\xi'') {\rm d}t
	\right|
	\nonumber\\
	\leq&
	L_{1}^{2}\|\xi'-\xi''\|
	\leq 
	C_{4}\|\xi'-\xi''\|.
\end{align*}
Consequently, (\ref{l1.3.51'}) implies
\begin{align}\label{l1.3.53}
	\left|
	\Phi_{\eta',y}(\xi')
	\!-\!
	\Phi_{\eta'',y}(\xi'')
	\right|
	\!\leq\! &
	\left|
	\Phi_{\eta',y}(\xi')
	\!-\!
	\Phi_{\eta'',y}(\xi')
	\right|
	\nonumber\\
	&+
	\left|
	\Phi_{\eta'',y}(\xi')
	-
	\Phi_{\eta'',y}(\xi'')
	\right|
	\nonumber\\
	\!\leq &
	C_{4}(1\!+\!\psi(y) )\!\left(\|\eta'\!-\!\eta''\|\!+\!\|\xi'\!-\!\xi''\| \right).
\end{align}
Using (\ref{l1.3.51}), (\ref{l1.3.53}), we deduce that (ii) is true.
\end{proof}

\begin{lemma}\label{lemma1.4}
Let Assumptions \ref{a11} -- \ref{a12} hold. Then, the following is true:

(i) There exist real numbers
$\delta_{5}, \delta_{6} \in (0,\delta_{4}]$,
$C_{5}\in[1,\infty)$ and an integer $n_{0}\geq 1$ such that
$\text{\rm Re}\left\{ \left\langle R_{\eta,\boldsymbol y}^{m:n}(\xi) \right\rangle
\right\}>0$,
$F_{\eta,\boldsymbol y}^{m:n}(\xi)\in	V_{\delta_{4} }({\cal P}({\cal X}) )$,
$F_{\eta,\boldsymbol y}^{m:m+n_{0}}(\xi)\in	V_{\delta_{6} }({\cal P}({\cal X}) )$ and
\begin{align}
	&\label{l1.4.1*}
	\left\|
	F_{\eta,\boldsymbol y}^{m:n}(\xi')
	-
	F_{\eta,\boldsymbol y}^{m:n}(\xi'')
	\right\|
	\leq
	C_{5} \|\xi'-\xi''\|,
	\\
	&\label{l1.4.3*}
	\left\|
	F_{\eta,\boldsymbol y}^{m:m+n_{0}}(\xi')
	-
	F_{\eta,\boldsymbol y}^{m:m+n_{0}}(\xi'')
	\right\|
	\leq
	\frac{\|\xi'-\xi''\| }{2}
\end{align}
for all $\eta\in V_{\delta_{5} }(\Theta)$,
$\xi,\xi',\xi''\in V_{\delta_{6}}({\cal P}({\cal X}) )$,
$m+n_{0}\geq n\geq m\geq 0$ and any sequence
$\boldsymbol y = \{y_{n} \}_{n\geq 1}$ in ${\cal Y}$
($\delta_{4}$ is specified in Lemma \ref{lemma1.3}).

(ii) There exist real numbers $\delta_{7}\in(0,\delta_{5}]$,
$\delta_{8}\in(0,\delta_{6}]$ such that
$F_{\eta,\boldsymbol y}^{m:n}(\xi)\in V_{\delta_{6}}({\cal P}({\cal X}) )$
for all $\eta\in V_{\delta_{7} }(\Theta)$,
$\xi\in V_{\delta_{8}}({\cal P}({\cal X}) )$,
$m+n_{0}\geq n\geq m\geq 0$ and any sequence
$\boldsymbol y = \{y_{n} \}_{n\geq 1}$ in ${\cal Y}$.
\end{lemma}

\begin{proof}
(i) Throughout this part of the proof, the following notations is used.
$n_{0}$ is the integer defined by
\begin{align*}
	n_{0}
	=
	\left\lceil \frac{\log(4C_{3} ) }{|\log\gamma_{2} |} \right\rceil,
\end{align*}
while $C_{5}$, $\delta_{5}$, $\delta_{6}$ are the real numbers defined by
$C_{5}=M_{n_{0} } (1+\|\mu\|)$ and
\begin{align}\label{l1.4.101}
	\delta_{5}
	=
	\frac{\min\{\alpha_{n_{0} },\delta_{4} \} }{16C_{5}^{2} },
	\;\;\;\;\;\;\;
	\delta_{6}=2C_{5}\delta_{5}
\end{align}
($\gamma_{2}$, $C_{3}$, $\alpha_{n}$, $M_{n}$ are specified
in Lemmas \ref{lemma1.1}, \ref{lemma1.2}).
$\eta$, $\eta'$, $\eta''$ are any elements in $V_{\delta_{5}}(\Theta)$,
while $\theta$ is any element of $\Theta$ satisfying
$\|\eta-\theta\|<\delta_{5}$.
$\xi$, $\xi'$, $\xi''$ are any elements in $V_{\delta_{6}}({\cal P}({\cal X}) )$,
while $\lambda$ is any element of ${\cal P}({\cal X} )$ satisfying
$\|\xi-\lambda\|<\delta_{6}$.
$x$ is any element of ${\cal X}$,
while $\boldsymbol y = \{y_{n} \}_{n\geq 1}$ is any sequence in ${\cal Y}$.
$B$ is any element of ${\cal B}({\cal X})$.
$m$, $n$ are any integers satisfying $m+n_{0}\geq n>m>0$.

Owing to Lemma \ref{lemma1.2}, we have
$\text{\rm Re}\left\{
\left\langle R_{\eta,\boldsymbol y}^{m:n} (\xi) \right\rangle \right\} > 0$, 
as $\delta_{5}\leq\delta_{6}\leq\alpha_{n_{0}}\leq\alpha_{n-m}$
results from $n-m\leq n_{0}$.
Hence, we get
\begin{align}\label{l1.4.701}
	&
	\text{\rm Re}\left\{
	\left\langle R_{\eta,\boldsymbol y}^{m:n} (t\xi'+(1-t)\xi'') \right\rangle
	\right\}
	\nonumber\\
	&=
	t\text{\rm Re}\left\{
	\left\langle R_{\eta,\boldsymbol y}^{m:n} (\xi') \right\rangle
	\right\}
	+
	(1-t)\text{\rm Re}\left\{
	\left\langle R_{\eta,\boldsymbol y}^{m:n} (\xi'') \right\rangle
	\right\}
	>
	0
\end{align}
for $t\in[0,1]$.
Moreover,
using Lemma \ref{lemma1.2}, we conclude
\begin{align}\label{l1.4.1}
	&
	\left|
	\int I_{B}(x')
	\left(
	h_{\eta',\boldsymbol y}^{m:n}(x'|x,\xi')
	-
	h_{\eta'',\boldsymbol y}^{m:n}(x'|x,\xi'')
	\right)
	\mu({\rm d}x')
	\right|
	\nonumber\\
	&\leq
	\int I_{B}(x')
	\left|
	h_{\eta',\boldsymbol y}^{m:n}(x'|x,\xi')
	-
	h_{\eta'',\boldsymbol y}^{m:n}(x'|x,\xi'')
	\right|
	\mu({\rm d}x')
	\nonumber\\
	&\leq
	M_{n-m} \|\mu\|
	\left(\|\eta'-\eta''\| + \|\xi'-\xi''\| \right)
	\nonumber\\
	&\leq
	C_{5} \left(\|\eta'-\eta''\| + \|\xi'-\xi''\| \right)
\end{align}
as $C_{5}\geq M_{n_{0}}\geq M_{n-m}$ results from $n-m\leq n_{0}$.
Relying on the same lemma, we deduce
\begin{align*}
	&
	\left|
	F_{\eta',\boldsymbol y}^{m:n}(B|\xi')
	-
	F_{\eta'',\boldsymbol y}^{m:n}(B|\xi'')
	\right|
	\\
	&\leq
	\int I_{B}(x)
	\left|
	f_{\eta',\boldsymbol y}^{m:n}(x|\xi')
	-
	f_{\eta'',\boldsymbol y}^{m:n}(x|\xi'')
	\right|
	\mu({\rm d}x)
	\\
	&\leq
	M_{n-m}\|\mu\|
	\left(\|\eta'-\eta''\|+\|\xi'-\xi''\| \right)
	\\
	&\leq
	C_{5} \left(\|\eta'-\eta''\|+\|\xi'-\xi''\| \right).
\end{align*}
Hence, we have
\begin{align}\label{l1.4.3}
	\left\|
	F_{\eta',\boldsymbol y}^{m:n}(\xi')
	-
	F_{\eta'',\boldsymbol y}^{m:n}(\xi'')
	\right\|
	\leq
	C_{5} \left(\|\eta'-\eta''\|+\|\xi'-\xi''\| \right)
\end{align}
as $\|\eta-\theta\|<\delta_{5}$, $\|\xi-\lambda\|<\delta_{6}$.
Setting $\eta'=\eta$, $\eta''=\eta$ in (\ref{l1.4.3}), 
we get (\ref{l1.4.1*}).
We also get
\begin{align*}
	\left\|
	F_{\eta,\boldsymbol y}^{m:n}(\xi)
	-
	F_{\theta,\boldsymbol y}^{m:n}(\lambda)
	\right\|
	\leq &
	C_{5} \left(\|\eta-\theta\|+\|\xi-\lambda\| \right)
	\\
	<&
	C_{5} (\delta_{5}+\delta_{6} )
	\leq
	\delta_{4}.
\end{align*}
Therefore,
$F_{\eta,\boldsymbol y}^{m:n}(\xi)\in V_{\delta_{4}}({\cal P}({\cal X}) )$
for $m+n_{0}\geq n>m\geq 0$, 
as $F_{\theta,\boldsymbol y}^{m:n}(\lambda)\in{\cal P}({\cal X})$.

Let $\phi_{\eta,\boldsymbol y}^{m:n}(t,x|\xi',\xi'')$ be the function
defined by
\begin{align}\label{l1.4.701'}
	\phi_{\eta,\boldsymbol y}^{m:n}(t,x|\xi',\xi'')
	=
	f_{\eta,\boldsymbol y}^{m:n}(x|t\xi'+(1-t)\xi'' )
\end{align}
for $t\in[0,1]$, $m+n_{0}\geq n>m\geq 0$.
Then, due to (\ref{l1.4.701}), we have
\begin{align*}
	&
	\phi_{\eta,\boldsymbol y}^{m:n}(t,x|\xi',\xi'')
	\\
	&=
	\frac{
	\int r_{\eta,\boldsymbol y}^{m:n}(x|x') (t\xi'+(1-t)\xi'')({\rm d}x') }
	{\left\langle R_{\eta,\boldsymbol y}^{m:n}(t\xi'+(1-t)\xi'') \right\rangle }
	\\
	&=
	\frac{
	t\int r_{\eta,\boldsymbol y}^{m:n}(x|x') \xi'({\rm d}x')
	+
	(1-t)\int r_{\eta,\boldsymbol y}^{m:n}(x|x') \xi''({\rm d}x') }
	{t \left\langle R_{\eta,\boldsymbol y}^{m:n}(\xi') \right\rangle
	+
	(1-t) \left\langle R_{\eta,\boldsymbol y}^{m:n}(\xi'') \right\rangle }.
\end{align*}
Thus, we get
\begin{align*}
	\frac{\partial}{\partial t}
	\phi_{\eta,\boldsymbol y}^{m:n}(t,x|\xi',\xi'')
	=&
	\frac{
	\int r_{\eta,\boldsymbol y}^{m:n}(x|x') (\xi'-\xi'')({\rm d}x') }
	{\left\langle R_{\eta,\boldsymbol y}^{m:n}(t\xi'+(1-t)\xi'' ) \right\rangle }
	\nonumber\\
	&-
	f_{\eta,\boldsymbol y}^{m:n}(x|t\xi'+(1-t)\xi'' )
	\nonumber\\
	&\cdot
	\frac{
	\left\langle R_{\eta,\boldsymbol y}^{m:n}(\xi') \right\rangle
	-
	\left\langle R_{\eta,\boldsymbol y}^{m:n}(\xi'') \right\rangle  }
	{\left\langle R_{\eta,\boldsymbol y}^{m:n}(t\xi'+(1-t)\xi'' ) \right\rangle }
	\nonumber\\
	=&
	\frac{
	\int r_{\eta,\boldsymbol y}^{m:n}(x|x') (\xi'-\xi'')({\rm d}x') }
	{\left\langle R_{\eta,\boldsymbol y}^{m:n}(t\xi'+(1-t)\xi'' ) \right\rangle }
	\nonumber\\
	&-
	f_{\eta,\boldsymbol y}^{m:n}(x|t\xi'+(1-t)\xi'' )
	\nonumber\\
	&\cdot
	\frac{\int\int r_{\eta,\boldsymbol y}^{m:n}(x''|x')\mu({\rm d}x'') (\xi'\!-\!\xi'')({\rm d}x') }
	{\left\langle R_{\eta,\boldsymbol y}^{m:n}(t\xi'+(1-t)\xi'' ) \right\rangle }.
\end{align*}
Consequently, (\ref{2.57*}) -- (\ref{2.61*}) imply
\begin{align}\label{l1.4.701''}
	&
	\frac{\partial}{\partial t}
	\phi_{\eta,\boldsymbol y}^{m:n}(t,x|\xi',\xi'')
	\nonumber\\
	&
	\begin{aligned}
	=&
	\int g_{\eta,\boldsymbol y}^{m:n}(x|x',t\xi'+(1-t)\xi'')
	(\xi'-\xi'')({\rm d}x')
	\\
	&-
	f_{\eta,\boldsymbol y}^{m:n}(x|t\xi'+(1-t)\xi'' )
	\\
	&\cdot
	\int\int g_{\eta,\boldsymbol y}^{m:n}(x''|x',t\xi'+(1-t)\xi'')
	\mu({\rm d}x'') (\xi'-\xi'')({\rm d}x')
	\end{aligned}
	\nonumber\\
	&=
	\int h_{\eta,\boldsymbol y}^{m:n}(x|x',t\xi'+(1-t)\xi'')
	(\xi'-\xi'')({\rm d}x').
\end{align}
Hence, we have
\begin{align}\label{l1.4.701'''}
	&
	f_{\eta,\boldsymbol y}^{m:n}(x|\xi')
	-
	f_{\eta,\boldsymbol y}^{m:n}(x|\xi'')
	\nonumber\\
	&=
	\phi_{\eta,\boldsymbol y}^{m:n}(1,x|\xi',\xi'')
	-
	\phi_{\eta,\boldsymbol y}^{m:n}(0,x|\xi',\xi'')
	\nonumber\\
	&=
	\int\int_{0}^{1} h_{\eta,\boldsymbol y}^{m:n}(x|x',t\xi'+(1-t)\xi'')
	(\xi'-\xi'')({\rm d}x') {\rm d}t.
\end{align}
Therefore, (\ref{2.5*}) yields
\begin{align}\label{l1.4.7}
	&
	F_{\eta,\boldsymbol y}^{m:n}(B|\xi')
	-
	F_{\eta,\boldsymbol y}^{m:n}(B|\xi'')
	\nonumber\\
	&=
	\int I_{B}(x)
	\left(
	f_{\eta,\boldsymbol y}^{m:n}(x|\xi')
	-
	f_{\eta,\boldsymbol y}^{m:n}(x|\xi'')
	\right)
	\mu({\rm d}x)
	\nonumber\\
	&
	\begin{aligned}[b]
	=&
	\int\int\int_{0}^{1} I_{B}(x)
	h_{\eta,\boldsymbol y}^{m:n}(x|x',t\xi'+(1-t)\xi'')
	\\
	&\cdot
	\mu({\rm d}x) (\xi'-\xi'')({\rm d}x') {\rm d}t.
	\end{aligned}
\end{align}
Since $\lambda+\alpha t\delta_{x}\in V_{\delta_{6}}({\cal P}({\cal X}) )$
for $\alpha\in(0,\delta_{6})$, $t\in[0,1]$,
we then get
\begin{align}\label{l1.4.9}
	&
	F_{\theta,\boldsymbol y}^{m:n}(B|\lambda+\alpha\delta_{x} )
	-
	F_{\theta,\boldsymbol y}^{m:n}(B|\lambda)
	\nonumber\\
	&
	=
	\alpha\int\int_{0}^{1}  I_{B}(x')
	h_{\theta,\boldsymbol y}^{m:n}(x'|x,\lambda+\alpha t\delta_{x} )
	\mu({\rm d}x') {\rm d}t
\end{align}
for the same $\alpha$.
Moreover, Lemma \ref{lemma1.1} yields
\begin{align}\label{l1.4.11}
	&
	\left|
	F_{\theta,\boldsymbol y}^{m:n}(B|\lambda+\alpha\delta_{x} )
	-
	F_{\theta,\boldsymbol y}^{m:n}(B|\lambda)
	\right|
	\nonumber\\
	&\leq
	C_{3}\gamma_{2}^{n-m} \|\alpha\delta_{x}\|
	=
	\alpha C_{3}\gamma_{2}^{n-m}
\end{align}
for $\alpha\in(0,\delta_{6})$.
Combining (\ref{l1.4.9}), (\ref{l1.4.11}), we get
\begin{align}\label{l1.4.21}
	\left|
	\int\int_{0}^{1}  I_{B}(x')
	h_{\theta,\boldsymbol y}^{m:n}(x'|x,\lambda+\alpha t\delta_{x} )
	\mu({\rm d}x') {\rm d}t
	\right|
	\leq
	C_{3}\gamma_{2}^{n-m}.
\end{align}
Using (\ref{l1.4.1}), (\ref{l1.4.21}), we conclude
\begin{align*}
	&
	\left|
	\int I_{B}(x')
	h_{\theta,\boldsymbol y}^{m:n}(x'|x,\lambda) \mu({\rm d}x')
	\right|
	\\
	&
	\begin{aligned}
	\leq&
	\left|
	\int\int_{0}^{1}  I_{B}(x')
	h_{\theta,\boldsymbol y}^{m:n}(x'|x,\lambda+\alpha t\delta_{x} )
	\mu({\rm d}x') {\rm d}t
	\right|
	\\
	&\!\!\!+\!\!\!
	\int_{0}^{1}\!\!
	\left|
	\int\!\!\! I_{B}(x'\!)\!
	\left(
	h_{\theta,\boldsymbol y}^{m:n}(x'|x,\!\lambda\!+\!\alpha t\delta_{x} )
	\!-\!
	h_{\theta,\boldsymbol y}^{m:n}(x'|x,\!\lambda)
	\right) \!\mu({\rm d}x'\!)
	\right|\! {\rm d}t
	\end{aligned}
	\\
	&\leq
	C_{3}\gamma_{2}^{n-m} + C_{5}\alpha
\end{align*}
for $\alpha\in(0,\delta_{6})$, 
as $\|\alpha t\delta_{x}\|\leq\alpha$.
Letting $\alpha\rightarrow 0$, we deduce
\begin{align*}
	\left|
	\int I_{B}(x')
	h_{\theta,\boldsymbol y}^{m:n}(x'|x,\lambda) \mu({\rm d}x')
	\right|
	\leq&
	C_{3}\gamma_{2}^{n-m}.
\end{align*}
Consequently, (\ref{l1.4.1}) yields
\begin{align*}
	&
	\left|
	\int I_{B}(x')
	h_{\eta,\boldsymbol y}^{m:n}(x'|x,\xi) \mu({\rm d}x')
	\right|
	\\
	&
	\begin{aligned}
	\leq &
	\left|
	\int I_{B}(x')
	h_{\theta,\boldsymbol y}^{m:n}(x'|x,\lambda) \mu({\rm d}x')
	\right|
	\\
	&
	+
	\left|
	\int I_{B}(x')
	\left(
	h_{\eta,\boldsymbol y}^{m:n}(x'|x,\xi)
	-
	h_{\theta,\boldsymbol y}^{m:n}(x'|x,\lambda)
	\right) \mu({\rm d}x')
	\right|
	\end{aligned}
	\\
	&\leq
	C_{3}\gamma_{2}^{n-m}
	+
	C_{5}\left(\|\eta-\theta \| + \|\xi-\lambda \| \right)
	\\
	&\leq
	C_{3}\gamma_{2}^{n-m} + C_{5}(\delta_{5}+\delta_{6} )
	\leq 
	C_{3}\gamma_{2}^{n-m} + \frac{1}{4}
\end{align*}
as $\|\eta-\theta\|<\delta_{5}$, $\|\xi-\lambda\|<\delta_{6}$,
$C_{5}\delta_{5}\leq C_{5}\delta_{6}\leq 1/8$.
Combining this with (\ref{l1.4.7}), we get
\begin{align*}
	&
	\left|
	F_{\eta,\boldsymbol y}^{m:n}(B|\xi')
	-
	F_{\eta,\boldsymbol y}^{m:n}(B|\xi'')
	\right|
	\\
	&
	\begin{aligned}
	\leq &
	\int\!\int_{0}^{1}
	\left|
	\int I_{B}(x')
	h_{\eta,\boldsymbol y}^{m:n}(x'|x,t\xi'+(1-t)\xi'')
	\mu({\rm d}x')
	\right|
	\\
	&\cdot
	|\xi'-\xi''|({\rm d}x) {\rm d}t
	\end{aligned}
	\\
	&\leq
	\left(C_{3}\gamma_{2}^{n-m} + \frac{1}{4} \right)
	\|\xi'-\xi''\|
\end{align*}
as $t\xi'+(1-t)\xi''\in V_{\delta_{6}}({\cal P}({\cal X}) )$ results
from $t\in[0,1]$
and the convexity of $V_{\delta_{6}}({\cal P}({\cal X}) )$.
Therefore, we have
\begin{align*}
	\left\|
	F_{\eta,\boldsymbol y}^{m:n}(\xi')
	-
	F_{\eta,\boldsymbol y}^{m:n}(\xi'')
	\right\|
	\leq
	\left(C_{3}\gamma_{2}^{n-m} + \frac{1}{4} \right)
	\|\xi'-\xi''\|.
\end{align*}
Hence, we get
\begin{align}\label{l1.4.23}
	\left\|
	F_{\eta,\boldsymbol y}^{m:m+n_{0}}(\xi')
	-
	F_{\eta,\boldsymbol y}^{m:m+n_{0}}(\xi'')
	\right\|
	\leq &
	\left(C_{3}\gamma_{2}^{n_{0}} + \frac{1}{4} \right)
	\|\xi'-\xi''\|
	\nonumber\\
	\leq &
	\frac{\|\xi'-\xi''\|}{2}
\end{align}
as $C_{3}\gamma_{2}^{n_{0}}\leq 1/4$.
Consequently, (\ref{l1.4.3*}) holds.
Moreover, (\ref{l1.4.3}) implies
\begin{align*}
	&
	\left\|
	F_{\eta,\boldsymbol y}^{m:m+n_{0}}(\xi)
	-
	F_{\theta,\boldsymbol y}^{m:m+n_{0}}(\lambda)
	\right\|
	\\
	&\begin{aligned}
	\leq &
	\left\|
	F_{\eta,\boldsymbol y}^{m:m+n_{0}}(\xi)
	-
	F_{\eta,\boldsymbol y}^{m:m+n_{0}}(\lambda)
	\right\|
	\\
	&+
	\left\|
	F_{\eta,\boldsymbol y}^{m:m+n_{0}}(\lambda)
	-
	F_{\theta,\boldsymbol y}^{m:m+n_{0}}(\lambda)
	\right\|
	\end{aligned}
	\\
	&\leq
	\frac{\|\xi-\lambda \|}{2}
	+
	C_{5} \|\eta-\theta \|
	\\
	&<
	\frac{\delta_{6}}{2} + C_{5}\delta_{5}
	= 
	\delta_{6}
\end{align*}
as $\|\eta-\theta\|<\delta_{5}$, $\|\xi-\lambda\|<\delta_{6}$.
Thus, $F_{\eta,\boldsymbol y}^{m:m+n_{0}}(\xi)\in V_{\delta_{6}}({\cal P}({\cal X}) )$
for $m\geq 0$, as $F_{\theta,\boldsymbol y}^{m:m+n_{0}}(\lambda)\in {\cal P}({\cal X})$.

(ii) Let $\delta_{7}$, $\delta_{8}$ be the real numbers defined by $\delta_{7}=\delta_{5}$, $\delta_{8}=\delta_{5}$
($\delta_{5}$ is specified in (\ref{l1.4.101})).
Moreover, let $\theta$, $\lambda$, $\boldsymbol y$ have the same meaning as in (i),
while $\eta$, $\xi$ are any elements of $V_{\delta_{6}}(\Theta)$, $V_{\delta_{7}}({\cal P}({\cal X} ) )$
(respectively).
Consequently, when $\|\eta-\theta\|<\delta_{7}$,
$\|\xi-\lambda\|<\delta_{8}$, (\ref{l1.4.3}) yields
\begin{align*}
	\left\|
	F_{\eta,\boldsymbol y}^{m:n}(\xi)
	-
	F_{\theta,\boldsymbol y}^{m:n}(\lambda)
	\right\|
	\leq &
	C_{5} \left(\|\eta-\theta\|+\|\xi-\lambda\| \right)
	\\
	< &
	C_{5} (\delta_{7}+\delta_{8} )
	\leq
	\delta_{6}
\end{align*}
for $m+n_{0}\geq n>m\geq 0$.
Therefore,
$F_{\eta,\boldsymbol y}^{m:n}(\xi)\in V_{\delta_{6}}({\cal P}({\cal X}) )$
for $m+n_{0}\geq n>m\geq 0$.
\end{proof}

\begin{lemma}\label{lemma1.6}
Let Assumptions \ref{a11} -- \ref{a12} hold. Then, the following is true:

(i) $\text{\rm Re}\left\{ \left\langle R_{\eta,\boldsymbol y}^{m:n}(\xi) \right\rangle
\right\}\neq 0$,
$F_{\eta,\boldsymbol y}^{m:n}(\xi)\in	V_{\delta_{4} }({\cal P}({\cal X}) )$
for all $\eta\in V_{\delta_{5} }(\Theta)$,
$\xi\in V_{\delta_{6}}({\cal P}({\cal X}) )$,
$n\geq m\geq 0$ and any sequence
$\boldsymbol y = \{y_{n} \}_{n\geq 1}$ in ${\cal Y}$
($\delta_{4}$, $\delta_{5}$, $\delta_{6}$ are specified in Lemmas
\ref{lemma1.3}, \ref{lemma1.4}).

(ii) There exist real numbers $\gamma_{3}\in(0,1)$, $C_{6}\in[1,\infty)$
such that
\begin{align}\label{l1.6.1*}
	\left\|
	F_{\eta,\boldsymbol y}^{m:n}(\xi')
	-
	F_{\eta,\boldsymbol y}^{m:n}(\xi'')
	\right\|
	\leq
	C_{6}\gamma_{3}^{n-m}
	\left\|\xi'-\xi'' \right\|
\end{align}
for all $\eta\in V_{\delta_{5}}(\Theta)$,
$\xi',\xi''\in V_{\delta_{6}}({\cal P}({\cal X}) )$,
$n\geq m\geq 0$ and any sequence
$\boldsymbol y = \{y_{n} \}_{n\geq 1}$ in ${\cal Y}$.
\end{lemma}

\begin{proof}
(i)
Let $n_{k}(m)$ be the integer defined by $n_{k}(m)=m+kn_{0}$ for $m,k\geq 0$
($n_{0}$ is specified in Lemma \ref{lemma1.4}).
Moreover, let $\boldsymbol y = \{y_{n} \}_{n\geq 1}$ be any sequence in ${\cal Y}$.

First, we show
\begin{align}
	&\label{l1.6.31'}
	\text{\rm Re}\left\{\left\langle R_{\eta,\boldsymbol y}^{m:n}(\xi) \right\rangle \right\}
	\neq 0,
	\\
	&\label{l1.6.31''}
	F_{\eta,\boldsymbol y}^{m:n}(\xi)
	\in V_{\delta_{4}}({\cal P}({\cal X}) ),
	\\
	&\label{l1.6.31'''}
	F_{\eta,\boldsymbol y}^{m:n_{k}(m)}(\xi)
	\in V_{\delta_{6}}({\cal P}({\cal X}) )
\end{align}
for each $\eta\in V_{\delta_{5}}(\Theta)$, $\xi\in V_{\delta_{6}}({\cal P}({\cal X}) )$,
$n_{k}(m)\geq n\geq m\geq 0$, $k\geq 0$.
We prove this by induction in $k$.

Since $n_{0}(m)=n=m$ when $n_{0}(m)\geq n\geq m\geq 0$,
we have
$R_{\eta,\boldsymbol y}^{m:n}(\xi) = \xi$,
$F_{\eta,\boldsymbol y}^{m:n}(\xi) = \xi$
for $\eta\in V_{\delta_{5}}(\Theta)$,
$\xi\in V_{\delta_{6}}({\cal P}({\cal X}) )$,
$n_{0}(m)\geq n\geq m\geq 0$.
Hence, (\ref{l1.6.31'}) -- (\ref{l1.6.31'''}) hold for $k=0$
and $\eta\in V_{\delta_{5}}(\Theta)$, $\xi\in V_{\delta_{6}}({\cal P}({\cal X}) )$,
$n_{k}(m)\geq n\geq m\geq 0$.
Now, the induction hypothesis is formulated:
Suppose that (\ref{l1.6.31'}) -- (\ref{l1.6.31'''}) are true for some $k\geq 0$
and any $\eta\in V_{\delta_{5}}(\Theta)$, $\xi\in V_{\delta_{6}}({\cal P}({\cal X}) )$,
$n_{k}(m)\geq n\geq m\geq 0$.
Then, to show (\ref{l1.6.31'}) -- (\ref{l1.6.31'''})
for $\eta\in V_{\delta_{5}}(\Theta)$, $\xi\in V_{\delta_{6}}({\cal P}({\cal X}) )$,
$n_{k+1}(m)\geq n\geq m\geq 0$,
it is sufficient to demonstrate (\ref{l1.6.31'}) -- (\ref{l1.6.31'''})
for $\eta\in V_{\delta_{5}}(\Theta)$, $\xi\in V_{\delta_{6}}({\cal P}({\cal X}) )$,
$n_{k+1}(m)\geq n\geq n_{k}(m)$, $m\geq 0$.

In the rest of the proof of (i), $\eta$, $\xi$ are any elements of
$V_{\delta_{5}}(\Theta)$, $V_{\delta_{6}}({\cal P}({\cal X}) )$ (respectively).
Owing to Lemma \ref{lemma1.4}, we have
$\text{\rm Re}\big\{ \big\langle R_{\eta,\boldsymbol y}^{n_{k}(m):n}(\xi)
\big\rangle \big\}>0$,
$F_{\eta,\boldsymbol y}^{n_{k}(m):n}(\xi)\in V_{\delta_{4}}({\cal P}({\cal X}) )$,
$F_{\eta,\boldsymbol y}^{n_{k}(m):n_{k+1}(m)}(\xi)\in V_{\delta_{6}}({\cal P}({\cal X}) )$
for $n_{k+1}(m)\geq n\geq n_{k}(m)$, $m\geq 0$.
Since $F_{\eta,\boldsymbol y}^{m:n_{k}(m)}(\xi)\in V_{\delta_{6}}({\cal P}({\cal X}) )$
(due to the induction hypothesis), we then get
\begin{align}
	&\label{l1.6.201'}
	\text{\rm Re}\big\{
	\big\langle
	R_{\eta,\boldsymbol y}^{n_{k}(m):n}
	\big( F_{\eta,\boldsymbol y}^{m:n_{k}(m)}(\xi) \big)
	\big\rangle
	\big\}
	>0,
	\\
	&\label{l1.6.201}
	F_{\eta,\boldsymbol y}^{n_{k}(m):n}
	\big( F_{\eta,\boldsymbol y}^{m:n_{k}(m)}(\xi) \big)
	\in V_{\delta_{4}}({\cal P}({\cal X}) ),
	\\
	&\label{l1.6.203}
	F_{\eta,\boldsymbol y}^{n_{k}(m):n_{k+1}(m)}
	\big( F_{\eta,\boldsymbol y}^{m:n_{k}(m)}(\xi) \big)
	\in V_{\delta_{6}}({\cal P}({\cal X}) )
\end{align}
for $n_{k+1}(m)\geq n\geq n_{k}(m)$, $m\geq 0$.
As $\big\langle R_{\eta,\boldsymbol y}^{m:n_{k}(m)} \left( \xi \right) \big\rangle \neq 0$
(due to the induction hypothesis),
Lemma \ref{lemma1.21} (Part (ii)) and (\ref{l1.6.201'}) imply
\begin{align}\label{l1.6.301}
	\left\langle
	R_{\eta,\boldsymbol y}^{m:n}
	\left( \xi \right)
	\right\rangle
	=&
	\big\langle
	R_{\eta,\boldsymbol y}^{m:n_{k}(m)}
	\left( \xi \right)
	\big\rangle
	\big\langle
	R_{\eta,\boldsymbol y}^{n_{k}(m):n}
	\big( F_{\eta,\boldsymbol y}^{m:n_{k}(m)}(\xi) \big)
	\big\rangle
	\nonumber\\
	\neq &
	0
\end{align}
for $n_{k+1}(m)\geq n\geq n_{k}(m)$, $m\geq 0$.
Since $\big\langle R_{\eta,\boldsymbol y}^{m:n_{k}(m)} \left( \xi \right) \big\rangle \neq 0$,
$\big\langle R_{\eta,\boldsymbol y}^{m:n_{k+1}(m)} \left( \xi \right) \big\rangle \neq 0$,
$F_{\eta,\boldsymbol y}^{m:n_{k}(m)}(\xi)\in V_{\delta_{6}}({\cal P}({\cal X}) )$
(due to the induction hypothesis and (\ref{l1.6.301})),
Lemma \ref{lemma1.21} (Part (iii)) and (\ref{l1.6.201}), (\ref{l1.6.203}) yield
\begin{align}
	&\label{l1.6.303}
	F_{\eta,\boldsymbol y}^{m:n}(\xi)
	=
	F_{\eta,\boldsymbol y}^{n_{k}(m):n}
	\big( F_{\eta,\boldsymbol y}^{m:n_{k}(m)}(\xi) \big)
	\in V_{\delta_{4}}({\cal P}({\cal X}) ),
	\\
	&\label{l1.6.305}
	\begin{aligned}[b]
	F_{\eta,\boldsymbol y}^{m:n_{k+1}(m)}
	\left( \xi \right)
	=&
	F_{\eta,\boldsymbol y}^{n_{k}(m):n_{k+1}(m)}
	\big( F_{\eta,\boldsymbol y}^{m:n_{k}(m)}(\xi) \big)
	\\
	\in&
	V_{\delta_{6}}({\cal P}({\cal X}) )
	\end{aligned}
\end{align}
for $n_{k+1}(m)\geq n\geq n_{k}(m)$, $m\geq 0$.
Combining (\ref{l1.6.301}) -- (\ref{l1.6.305}) with the induction hypothesis,
we deduce that (\ref{l1.6.31'}) -- (\ref{l1.6.31'''}) hold
for $n_{k+1}(m)\geq n\geq m$, $m\geq 0$.
Then, relying on the principle of mathematical induction, we conclude
that (\ref{l1.6.31'}) -- (\ref{l1.6.31'''}) are satisfied
for each $\eta\in V_{\delta_{5}}(\Theta)$, $\xi\in V_{\delta_{6}}({\cal P}({\cal X}) )$,
$n_{k}(m)\geq n\geq m\geq 0$, $k\geq 0$.
As a direct consequence of this, we have that (i) is true.

(ii) Let $\gamma_{3}$, $C_{6}$ be the real numbers defined by $\gamma_{3}=2^{-1/n_{0}}$, $C_{6}=C_{5}\gamma_{3}^{-n_{0}}$
($C_{5}$, $n_{0}$ are specified in Lemma \ref{lemma1.4}),
while $n_{k}(m)$, $\boldsymbol y$ have the same meaning as in (i).
Moreover, let $\eta$ be any element of $V_{\delta_{6}}(\Theta)$,
while $\xi'$, $\xi''$ are any elements in $V_{\delta_{6}}({\cal P}({\cal X}) )$.

Owing to Lemmas \ref{lemma1.21}, \ref{lemma1.4} and (\ref{l1.6.31'}) -- (\ref{l1.6.31'''}), we have
\begin{align*}
	&
	\left\|
	F_{\eta,\boldsymbol y}^{m:n_{k+1}(m)}(\xi')
	-
	F_{\eta,\boldsymbol y}^{m:n_{k+1}(m)}(\xi'')
	\right\|
	\\
	&
	\begin{aligned}
	=
	&
	\left\|
	F_{\eta,\boldsymbol y}^{n_{k}(m):n_{k+1}(m)}
	\!\big(F_{\eta,\boldsymbol y}^{m:n_{k}(m)}(\xi') \big)
	\right.
	\\
	&
	\left.
	\;-
	F_{\eta,\boldsymbol y}^{n_{k}(m):n_{k+1}(m)}
	\!\big(F_{\eta,\boldsymbol y}^{m:n_{k}(m)}(\xi'') \big)
	\right\|
	\end{aligned}
	\\
	&\leq
	\frac{1}{2}
	\left\|
	F_{\eta,\boldsymbol y}^{m:n_{k}(m)}(\xi')
	-
	F_{\eta,\boldsymbol y}^{m:n_{k}(m)}(\xi'')
	\right\|
\end{align*}
for $m,k\geq 0$.
Consequently, we have
\begin{align*}
	&
	\left\|
	F_{\eta,\boldsymbol y}^{m:n_{k}(m)}(\xi')
	-
	F_{\eta,\boldsymbol y}^{m:n_{k}(m)}(\xi'')
	\right\|
	\\
	&\leq
	\frac{1}{2^{k}}
	\left\|
	F_{\eta,\boldsymbol y}^{m:m}(\xi')
	-
	F_{\eta,\boldsymbol y}^{m:m}(\xi'')
	\right\|
	=
	\gamma_{3}^{n_{k}(m)-m} \|\xi'-\xi''\|. 
\end{align*}
Combining this with Lemma \ref{lemma1.4} and (\ref{l1.6.31'}) -- (\ref{l1.6.31'''}), we get
\begin{align}\label{l1.6.41}
	&
	\left\|
	F_{\eta,\boldsymbol y}^{m:n}(\xi')
	-
	F_{\eta,\boldsymbol y}^{m:n}(\xi'')
	\right\|
	\nonumber\\
	&
	=
	\left\|
	F_{\eta,\boldsymbol y}^{n_{k}(m):n}\!
	\big(F_{\eta,\boldsymbol y}^{m:n_{k}(m)}(\xi') \big)
	\!-\!
	F_{\eta,\boldsymbol y}^{n_{k}(m):n}\!
	\big(F_{\eta,\boldsymbol y}^{m:n_{k}(m)}(\xi'') \big)
	\right\|
	\nonumber\\
	&\leq
	C_{5}
	\left\|
	F_{\eta,\boldsymbol y}^{m:n_{k}(m)}(\xi')
	-
	F_{\eta,\boldsymbol y}^{m:n_{k}(m)}(\xi'')
	\right\|
	\nonumber\\
	&\leq
	C_{5} \gamma_{3}^{n_{k}(m)-m} \|\xi'-\xi''\|
	\leq 
	C_{6} \gamma_{3}^{n-m} \|\xi'-\xi''\|
\end{align}
for $n_{k+1}(m)> n \geq n_{k}(m)$, $m,k\geq 0$. 
Here, we also use relations
$C_{5}\gamma_{3}^{n_{k}(m)-m}
=(C_{5}\gamma_{3}^{n_{k}(m)-n})\gamma_{3}^{n-m}
\leq (C_{5}\gamma_{3}^{-n_{0}} )\gamma_{3}^{n-m}
=C_{6}\gamma_{3}^{n-m}$.
Setting $k=\lfloor (n-m)/n_{0}\rfloor$ in (\ref{l1.6.41}), 
we conclude that (\ref{l1.6.1*}) holds for each $n\geq m\geq 0$ by.
\end{proof}

\section{Proof of Main Results }\label{section3*}

In this section, Theorems \ref{theorem1} and \ref{theorem2} are proved.
The proofs of these theorems crucially depend on
the results related to the kernels $S(z,{\rm d}z')$, $S_{\eta}(z,{\rm d}z')$
and the optimal filter $F_{\eta,\boldsymbol y}^{m:n}(\xi)$
(i.e., on Lemmas \ref{lemma2.3}, \ref{lemma2.1}, \ref{lemma1.6}).
As the properties of $S(z,{\rm d}z')$, $S_{\eta}(z,{\rm d}z')$
are very similar, the proofs of Theorems \ref{theorem1} and \ref{theorem2} have many elements in common.
In order not to consider these elements twice
(and to prove Theorems \ref{theorem1} and \ref{theorem2} as efficiently as
possible), we introduce a new kernel $T_{\eta}(z,{\rm d}z')$,
where $\eta\in\mathbb{C}^{d}$, $z\in{\cal Z}$.\footnote
{$T_{\eta}(z,{\rm d}z')$ can be considered as a mapping with the following
properties: (i) $T_{\eta}(z,B)$ maps $\eta\in\mathbb{C}^{d}$, $z\in{\cal Z}$,
$B\in{\cal B}({\cal Z})$ to $\mathbb{C}$,
(ii) $T_{\eta}(z,B)$ is measurable in $(\eta,z)$ for each $B\in{\cal B}({\cal Z})$,
and (iii) $T_{\eta}(z,B)$ is a complex measure in $B$
for each $\eta\in\mathbb{C}^{d}$, $z\in{\cal Z}$. }
Its purpose is to capture all common features of $S(z,{\rm d}z')$, $S_{\eta}(z,{\rm d}z')$
which are relevant for the proof of Theorems \ref{theorem1} and \ref{theorem2}.
Using $T_{\eta}(z,{\rm d}z')$, we recursively define kernels
$\left\{ T_{\eta}^{n}(z,{\rm d}z') \right\}_{n\geq 0}$ by
$T_{\eta}^{0}(z,B)=\delta_{z}(B)$ and
\begin{align*}
	T_{\eta}^{n+1}(z,B)
	=
	\int T_{\eta}^{n}(z',B) T_{\eta}(z,{\rm d}z'),
\end{align*}
where
$B\in{\cal B}({\cal Z})$.

Regarding $T_{\eta}(z,{\rm d}z')$, we assume the following.

\begin{assumption}\label{a4}
For each $\theta\in\Theta$, $z\in{\cal Z}$,
$T_{\theta}(z,{\rm d}z')$ is a probability measure.
\end{assumption}

\begin{assumption}\label{a5}
(i) There exist real numbers $\alpha\in(0,\delta]$, $L\in[1,\infty)$
such that
\begin{align*}
	&
	\left|
	T_{\eta'}
	-
	T_{\eta''}
	\right|(z,B)
	\leq
	L\|\eta'-\eta''\|,
	\\
	&
	\int \tilde{\psi}(z') |T_{\eta}|(z,{\rm d}z')
	\leq
	L
\end{align*}
for all $\eta,\eta',\eta''\in V_{\alpha}(\Theta)$, $z\in{\cal Z}$, $B\in{\cal B}({\cal Z})$
(here, $|T_{\eta'}-T_{\eta''}|(z,{\rm d}z')$ denotes the total variation of
$T_{\eta'}(z,{\rm d}z')-T_{\eta''}(z,{\rm d}z')$,
while $\delta$, $\tilde{\psi}(z)$ are specified in Assumption \ref{a1} and (\ref{5.5})).

(ii) For each $\eta\in V_{\alpha}(\Theta)$,
there exists a complex measure $\tau_{\eta}({\rm d}z)$
such that
$\lim_{n\rightarrow\infty}T_{\eta}^{n}(z,B)=\tau_{\eta}(B)$
for all $z\in{\cal Z}$, $B\in{\cal B}({\cal Z})$.

(iii) There exists a real number
$\beta\in(0,1)$ such that
\begin{align*}
	\left|
	T_{\eta}^{n} - \tau_{\eta}
	\right|(z,B)
	\leq
	L\beta^{n}
\end{align*}
for all $\eta\in V_{\alpha}(\Theta)$, $z\in{\cal Z}$, $B\in{\cal B}({\cal Z})$,
$n\geq 1$
(here, $\left|T_{\eta}^{n} - \tau_{\eta} \right|(z,{\rm d}z')$
stands for the total variation of
$T_{\eta}^{n}(z,{\rm d}z') - \tau_{\eta}({\rm d}z')$).
\end{assumption}

\begin{remark}
According to Lemmas \ref{lemma2.3} and \ref{lemma2.1},
both kernels $S(z,{\rm d}z')$, $S_{\eta}(z,{\rm d}z')$ satisfy Assumptions \ref{a4} and \ref{a5}.
These assumptions capture all common properties
of $S(z,{\rm d}z')$, $S_{\eta}(z,{\rm d}z')$ relevant for the proof of Theorems \ref{theorem1} and \ref{theorem2}.
\end{remark}

Besides the notations introduced in the previous sections,
we rely here on the following notations, too.
$u_{\eta}^{n}(z_{0:n})$ and
$F_{\eta}^{n}(\xi,z_{1:n})$
are (respectively) the function and the complex measure defined by
\begin{align*}
	&
	u_{\eta}^{n}(z_{0:n})
	=
	u_{\eta}^{n}(x_{0:n},y_{1:n}),
	\\
	&
	F_{\eta}^{n}(\xi,z_{1:n})
	=
	F_{\eta,\boldsymbol y}^{0:n}(\xi)
\end{align*}
for $\eta\in\mathbb{C}^{d}$, $\xi\in{\cal M}_{c}({\cal X})$,
$x_{0},\dots,x_{n}\in{\cal X}$, $y_{0},\dots,y_{n}\in{\cal Y}$, $n\geq 0$
and $z_{0}=(y_{0},x_{0}),\dots,z_{n}=(y_{n},x_{n})$,
where $\boldsymbol y = \{y'_{n} \}_{n\geq 1}$ is any sequence in ${\cal Y}$ satisfying
$y'_{k}=y_{k}$ for $n\geq k\geq 1$.\footnote
{Symbols $y_{1:0}$, $z_{1:0}$ denote empty sequences (i.e., sequences without any element).
$u_{\eta}^{n}(x_{0:n},y_{1:n})$, $F_{\eta,\boldsymbol y}^{0:n}(\xi)$ are specified
in (\ref{5.3}),(\ref{2.5*}).
$F_{\eta,\boldsymbol y}^{0:n}(\xi)$ depends only on
$y'_{1},\dots,y'_{n}$ and is independent of other elements of $\boldsymbol y$. }
$\Phi_{\eta}(\xi,z)$ is the function defined by
\begin{align*}
	\Phi_{\eta}(\xi,z)
	=
	\Phi_{\eta,y}(\xi),
\end{align*}
where $x\in{\cal X}$, $y\in{\cal Y}$
and $z=(y,x)$
($\Phi_{\eta,y}(\xi)$ is specified in (\ref{2.7*})).\footnote
{Functions $u_{\eta}^{n}(z_{0:n})$, $F_{\eta}^{n}(\xi,z_{1:n})$,
$\Phi_{\eta}(\xi,z)$ are just another notations
for $u_{\eta}^{n}(x_{0:n},y_{1:n})$,
$F_{\eta,\boldsymbol y}^{0:n}(\xi)$, $\Phi_{\eta,y}(\xi)$.
However, notations $u_{\eta}^{n}(z_{0:n})$, $F_{\eta}^{n}(\xi,z_{1:n})$,
$\Phi_{\eta}(\xi,z)$ are more suitable (than the original one)
for measure-theoretic arguments
which the analysis carried out in this section is based on.}
$\Phi_{\eta}^{n}(\xi,z)$ is the function defined by
\begin{align}\label{3.7*}
	\Phi_{\eta}^{n}(\xi,z)
	=&
	\int\cdots\int
	\Phi_{\eta}\left(
	F_{\eta}^{n}(\xi,z_{1:n}), z_{n+1}
	\right)
	\nonumber\\
	&\cdot
	T_{\eta}(z_{n},{\rm d}z_{n+1})\cdots T_{\eta}(z,{\rm d}z_{1})
\end{align}
for $n\geq 0$.
$\bar{A}_{\eta}^{n}(\xi)$,
$A_{\eta}^{k,n}(\xi,z)$, $B_{\eta}^{n}(\xi,z)$ are the functions defined by
\begin{align*}
	&
	\begin{aligned}
	\bar{A}_{\eta}^{n}(\xi)
	=&
	\int\cdots\int\int
	\begin{aligned}[t]
	&\left(
	\Phi_{\eta}\left(
	F_{\eta}^{n}(\xi,z_{1:n}),z_{n+1}
	\right)
	\right.
	\\
	&
	\;-
	\left.
	\Phi_{\eta}\left(
	F_{\eta}^{n-1}(\xi,z_{2:n}),z_{n+1}
	\right)
	\right)
	\end{aligned}
	\\
	&\cdot
	T_{\eta}(z_{n},{\rm d}z_{n+1})\cdots T_{\eta}(z_{0},{\rm d}z_{1})\tau_{\eta}({\rm d}z_{0}),
	\end{aligned}
	\\
	&
	\begin{aligned}
	A_{\eta}^{k,n}(\xi,z)
	\!=\!
	&
	\int\cdots\int\int
	\begin{aligned}[t]
	&
	\left(
	\Phi_{\eta}\left(
	F_{\eta}^{n-k+1}(\xi,z_{k:n}),z_{n+1}
	\right)
	\right.
	\\
	&
	\;-
	\left.
	\Phi_{\eta}\left(
	F_{\eta}^{n-k}(\xi,z_{k+1:n}),z_{n+1}
	\right)
	\right)
	\end{aligned}
	\\
	&\cdot
	T_{\eta}(z_{n},\!{\rm d}z_{n+1}\!)\cdots T_{\eta}(z_{k},\!{\rm d}z_{k+1}\!)
	(T_{\eta}^{k}\!\!-\!\tau_{\eta})(z,\!{\rm d}z_{k}),
	\end{aligned}
	\\
	&
	\begin{aligned}
	B_{\eta}^{n}(\xi,z)
	=
	\int
	\Phi_{\eta}\left(\xi,z'\right)
	(T_{\eta}^{n+1}-\tau_{\eta})(z,{\rm d}z')
	\end{aligned}
\end{align*}
for $n\geq k\geq 1$.

Under the notations introduced above, we have
\begin{align}\label{3.1*}
	\log q_{\theta}^{n}(y_{1:n}|\lambda)
	=
	\sum_{k=0}^{n-1}
	\Phi_{\theta}\left(
	F_{\theta}^{k}(\lambda,z_{1:k}), z_{k+1}
	\right)
\end{align}
for $\theta\in\Theta$, $\lambda\in{\cal P}({\cal X})$,
$x_{1},\dots,x_{n}\in{\cal X}$, $y_{1},\dots,y_{n}\in{\cal Y}$, $n\geq 1$ and
$z_{1}=(y_{1},x_{1}),\dots,z_{n}=(y_{n},x_{n})$.
We also have
\begin{align}
	\label{3.3*}
	\Phi_{\eta}^{n}(\xi,z') - \Phi_{\eta}^{n}(\xi,z'')
	=&
	\sum_{k=1}^{n}\left(
	A_{\eta}^{k,n}(\xi,z') - A_{\eta}^{k,n}(\xi,z'')
	\right)
	\nonumber\\
	&
	+
	B_{\eta}^{n}(\xi,z') - B_{\eta}^{n}(\xi,z''),
	\\
	\label{3.5*}
	\Phi_{\eta}^{n+1}(\xi,z) - \Phi_{\eta}^{n}(\xi,z)
	=&
	\sum_{k=1}^{n+1}
	A_{\eta}^{k,n+1}(\xi,z)
	-
	\sum_{k=1}^{n}
	A_{\eta}^{k,n}(\xi,z)
	\nonumber\\
	&
	+\!
	\bar{A}_{\eta}^{n+1}(\xi)
	+\!
	B_{\eta}^{n+1}(\xi,z) \!-\! B_{\eta}^{n}(\xi,z)
\end{align}
for $\eta\in V_{\alpha}(\Theta)$, $\xi\in{\cal M}_{c}({\cal X})$,
$z,z',z''\in{\cal Z}$, $n\geq 1$.

\begin{lemma}\label{lemma3.1}
Let Assumptions \ref{a11} -- \ref{a12}, \ref{a4} and \ref{a5} hold.
Then, there exist a function $\phi_{\eta}$
mapping $\eta\in\mathbb{C}^{d}$ to $\mathbb{C}$
and real numbers $\delta_{9},\gamma_{4}\in(0,1)$,
$C_{7}\in[1,\infty)$ such that
\begin{align}\label{l3.1.1*}
	\left|
	\Phi_{\eta}^{n}(\xi,z) - \phi_{\eta}
	\right|
	\leq
	C_{7}n\gamma_{4}^{n}
\end{align}
for all $\eta\in V_{\delta_{9}}(\Theta)$, $\xi\in V_{\delta_{9}}({\cal P}({\cal X}) )$,
$z\in{\cal Z}$, $n\geq 1$.
\end{lemma}

\begin{proof}
Throughout the proof, the following notations is used.
$\gamma_{4}$, $\delta_{9}$ are the real numbers defined by
$\gamma_{4}=\max\{\beta^{1/2},\gamma_{3}^{1/2} \}$,
$\delta_{9}=\min\{\delta_{7},\delta_{8},(1-\gamma_{4})/L\}$
($\beta$, $\delta_{7}$, $\delta_{8}$, $\gamma_{3}$, $L$ are specified in
Assumption \ref{a5} and Lemmas \ref{lemma1.4}, \ref{lemma1.6}).
$\eta$ is any element in $V_{\delta_{9}}(\Theta)$, while
$\theta$ is any element of $\Theta$ satisfying
$\|\eta-\theta\|<\delta_{9}$.
$\xi$, $\xi'$, $\xi''$ are any elements of $V_{\delta_{9}}({\cal P}({\cal X}) )$,
while $z$, $z'$, $z''$ are any elements in ${\cal Z}$.
$B$ is any element of ${\cal B}({\cal Z})$.
$n$, $k$ are any integers satisfying $n\geq k\geq 1$.

Owing to Assumptions \ref{a4} and \ref{a5}, we have
\begin{align}\label{l3.1.101}
	|T_{\eta}|(z,B)
	\leq &
	T_{\theta}(z,B)
	+
	|T_{\eta}-T_{\theta}|(z,B)
	\nonumber\\
	\leq &
	1 + L\|\eta-\theta\|
	\nonumber\\
	< &
	1+L\delta_{9}
	\leq
	\frac{1}{\gamma_{4}}
\end{align}
as $L\delta_{9}\leq 1-\gamma_{4}\leq 1/\gamma_{4}-1$.
Consequently, Assumption \ref{a4} yields
\begin{align}\label{l3.1.103}
	|\tau_{\eta}|(B)
	\leq
	|T_{\eta}-\tau_{\eta}|(z,B)
	+
	|T_{\eta}|(z,B)
	\leq
	L+\frac{1}{\gamma_{4}}.
\end{align}

Let $\tilde{C}_{1}=4C_{4}C_{6}$ where
$C_{4}$, $C_{6}$ are specified in Lemmas \ref{lemma1.3}, \ref{lemma1.6}.
Then, due to to Lemmas \ref{lemma1.21}, \ref{lemma1.3}, \ref{lemma1.6}, we have
\begin{align}\label{l3.1.3}
	&
	\left|
	\Phi_{\eta}\left(
	F_{\eta}^{n-k+1}(\xi,z_{k:n}),z_{n+1}
	\right)
	\!-\!
	\Phi_{\eta}\left(
	F_{\eta}^{n-k}(\xi,z_{k+1:n}),z_{n+1}
	\right)
	\right|
	\nonumber\\
	&\leq \!
	C_{4}\tilde{\psi}(z_{n+1} )
	\left\|
	F_{\eta}^{n-k+1}(\xi,z_{k:n})
	-
	F_{\eta}^{n-k}(\xi,z_{k+1:n})
	\right\|
	\nonumber\\
	&=\!
	C_{4}\tilde{\psi}(z_{n+1} )
	\left\|
	F_{\eta}^{n-k}\!\left(F_{\eta}^{1}(\xi,z_{k}),z_{k+1:n}\right)
	\!-\!
	F_{\eta}^{n-k}(\xi,z_{k+1:n})
	\right\|
	\nonumber\\
	&\leq \!
	C_{4}C_{6}\gamma_{3}^{n-k}\tilde{\psi}(z_{n+1} )
	\left\|F_{\eta}^{1}(\xi,z_{k})-\xi\right\|
	\nonumber\\
	&\leq
	\tilde{C}_{1}\gamma_{4}^{2(n-k)}\tilde{\psi}(z_{n+1} )
\end{align}
for $z_{1},\dots,z_{n+1}\in{\cal Z}$.\footnote
{To get the first two relations in (\ref{l3.1.3}),
use Lemmas \ref{lemma1.21}, \ref{lemma1.3}, and
notice that inclusions $\eta\in V_{\delta_{4}}(\Theta)$,
$F_{\eta}^{n-k+1}(\xi',z_{k:n})\in V_{\delta_{4}}({\cal P}({\cal X}) )$,
$F_{\eta}^{n-k}(\xi'',z_{k+1:n})\in V_{\delta_{4}}({\cal P}({\cal X}) )$
follow from Lemma \ref{lemma1.6} and
$\eta\in V_{\delta_{9}}(\Theta)\subseteq V_{\delta_{5}}(\Theta)$,
$\xi',\xi''\in V_{\delta_{9}}({\cal P}({\cal X}) )
\subseteq V_{\delta_{6}}({\cal P}({\cal X}) )$.
To get the third relation in (\ref{l3.1.3}),
use Lemma \ref{lemma1.6} and notice that
$F_{\eta}^{1}(\xi,z_{k})\in V_{\delta_{6}}({\cal P}({\cal X}) )$
follows from Lemma \ref{lemma1.3} and
$\eta\in V_{\delta_{9}}(\Theta)\subseteq V_{\delta_{7}}(\Theta)$,
$\xi\in V_{\delta_{9}}({\cal P}({\cal X}) )
\subseteq V_{\delta_{8}}({\cal P}({\cal X}) )$.
To get the last relation in (\ref{l3.1.3}),
use inequality
$\|F_{\eta}^{1}(\xi,z_{k})-\xi\|\leq \|F_{\eta}^{1}(\xi,z_{k})\|+\|\xi\|
\leq 2+\delta_{6}+\delta_{9} \leq 4$. }
Similarly, owing to Lemmas \ref{lemma1.3}, \ref{lemma1.6}, we have
\begin{align}\label{l3.1.1}
	&
	\left|
	\Phi_{\eta}\left(
	F_{\eta}^{n}(\xi',z_{1:n}),z_{n+1}
	\right)
	-
	\Phi_{\eta}\left(
	F_{\eta}^{n}(\xi'',z_{1:n}),z_{n+1}
	\right)
	\right|
	\nonumber\\
	&\leq
	C_{4}\tilde{\psi}(z_{n+1} )
	\left\|
	F_{\eta}^{n}(\xi',z_{1:n})
	-
	F_{\eta}^{n}(\xi'',z_{1:n})
	\right\|
	\nonumber\\
	&\leq
	C_{4}C_{6}\gamma_{3}^{n}\tilde{\psi}(z_{n+1} )\|\xi'-\xi''\|
	\nonumber\\
	&\leq
	\tilde{C}_{1}\gamma_{4}^{2n}\tilde{\psi}(z_{n+1} ).
\end{align}

Let $\tilde{C}_{2}=2\tilde{C}_{1}L^{2}/\gamma_{4}^{3}$.
Then, using Assumption \ref{a5} and (\ref{l3.1.101}), (\ref{l3.1.1}), we conclude
\begin{align}\label{l3.1.5}
	\left|
	\Phi_{\eta}^{n}(\xi',z)
	-
	\Phi_{\eta}^{n}(\xi'',z)
	\right|
	\leq &
	\begin{aligned}[t]
	&
	\tilde{C}_{1} \gamma_{4}^{2n}
	\int\cdots\int \tilde{\psi}(z_{n+1} )
	\\
	&\cdot |T_{\eta}|(z_{n},{\rm d}z_{n+1})\cdots|T_{\eta}|(z,{\rm d}z_{1})
	\end{aligned}
	\nonumber\\
	\leq &
	\tilde{C}_{1}L\gamma_{4}^{n}
	\leq
	\tilde{C}_{2}\gamma_{4}^{n}.
\end{align}
Similarly, relying on Assumption \ref{a5} and
(\ref{l3.1.101}), (\ref{l3.1.3}), we deduce
\begin{align}\label{l3.1.7}
	\left|
	A_{\eta}^{k,n}(\xi,z)
	\right|
	\leq &
	\begin{aligned}[t]
	&
	\tilde{C}_{1} \gamma_{4}^{2(n-k)}
	\int\cdots\int\int \tilde{\psi}(z_{n+1} )
	\\
	&\cdot
	|T_{\eta}|(z_{n},{\rm d}z_{n+1})\cdots|T_{\eta}|(z_{k},{\rm d}z_{k+1})
	\\
	&\cdot
	|T_{\eta}^{k}-\tau_{\eta}|(z,{\rm d}z_{k})
	\end{aligned}
	\nonumber\\
	\leq &
	\tilde{C}_{1}L^{2}\beta^{k}\gamma_{4}^{n-k}
	\leq
	\tilde{C}_{2}\gamma_{4}^{n}.
\end{align}
Moreover, using Assumption \ref{a5} and (\ref{l3.1.103}), (\ref{l3.1.3}), we get
\begin{align}\label{l3.1.9}
	\left|
	\bar{A}_{\eta}^{n}(\xi)
	\right|
	\leq &
	\begin{aligned}[t]
	&
	\tilde{C}_{1} \gamma_{4}^{2(n-1)}
	\int\cdots\int\int \tilde{\psi}(z_{n+1})
	\\
	&\cdot |T_{\eta}|(z_{n},{\rm d}z_{n+1})\cdots|T_{\eta}|(z_{0},{\rm d}z_{1})
	|\tau_{\eta}|({\rm d}z_{0})
	\end{aligned}
	\nonumber\\
	\leq &
	\tilde{C}_{1}L\left(L+\frac{1}{\gamma_{4}} \right)\gamma_{4}^{n-2}
	\leq
	\tilde{C}_{2}\gamma_{4}^{n}.
\end{align}

Let $\tilde{C}_{3}=C_{4}L^{2}$, $\tilde{C}_{4}=4(\tilde{C}_{2}+\tilde{C}_{3} )$.
Then, owing to Assumption \ref{a5} and Lemma \ref{lemma1.3}, we have
\begin{align}\label{l3.1.23}
	\left|
	B_{\eta}^{n}(\xi,z)
	\right|
	\leq &
	C_{4}
	\int\int \tilde{\psi}(z'') |T_{\eta}|(z',{\rm d}z'')
	|T_{\eta}^{n}-\tau_{\eta}|(z,{\rm d}z')
	\nonumber\\
	\leq &
	C_{4}L^{2}\beta^{n}
	\leq
	\tilde{C}_{3} \gamma_{4}^{n}.
\end{align}
Consequently, (\ref{3.5*}), (\ref{l3.1.7}), (\ref{l3.1.9}) yield
\begin{align}\label{l3.1.31}
	&
	\left|
	\Phi_{\eta}^{n+1}(\xi,z)
	-
	\Phi_{\eta}^{n}(\xi,z)
	\right|
	\nonumber\\
	&
	\begin{aligned}[t]
	\leq &
	\sum_{k=1}^{n+1}
	\left|
	A_{\eta}^{k,n+1}(\xi,z)
	\right|
	+
	\sum_{k=1}^{n}
	\left|
	A_{\eta}^{k,n}(\xi,z)
	\right|
	\\
	&+
	\left| \bar{A}_{\eta}^{n+1}(\xi) \right|
	+
	\left| B_{\eta}^{n+1}(\xi,z) \right|
	+
	\left| B_{\eta}^{n}(\xi,z) \right|
	\end{aligned}
	\nonumber\\
	&\leq
	2\tilde{C}_{2} (n+1) \gamma_{4}^{n} + 2\tilde{C}_{3}\gamma_{4}^{n}
	\leq 
	\tilde{C}_{4} n \gamma_{4}^{n}.
\end{align}
Hence, we have
\begin{align}\label{l3.1.33}
	&
	\sum_{n=1}^{\infty}
	\left|
	\Phi_{\eta}^{n+1}(\xi,z)
	-
	\Phi_{\eta}^{n}(\xi,z)
	\right|
	\nonumber\\
	&\leq
	\tilde{C}_{4} \sum_{n=1}^{\infty} n\gamma_{4}^{n}
	\leq 
	\frac{\tilde{C}_{4} }{(1-\gamma_{4} )^{2} }
	<
	\infty.
\end{align}
Now, combining (\ref{3.3*}), (\ref{l3.1.7}), (\ref{l3.1.23}), we get
\begin{align*}
	\left|
	\Phi_{\eta}^{n}(\xi,z')
	-
	\Phi_{\eta}^{n}(\xi,z'')
	\right|
	\!\leq\! &
	\sum_{k=1}^{n}\!
	\left|
	A_{\eta}^{k,n}(\xi,z')
	\right|
	\!+\!
	\sum_{k=1}^{n}\!
	\left|
	A_{\eta}^{k,n}(\xi,z'')
	\right|
	\\
	&
	+
	\left|
	B_{\eta}^{n}(\xi,z')
	\right|
	+
	\left|
	B_{\eta}^{n}(\xi,z'')
	\right|
	\\
	\leq &
	2\tilde{C}_{2}n\gamma_{4}^{n} + 2\tilde{C}_{3}\gamma_{4}^{n}.
\end{align*}
Then, (\ref{l3.1.5}) implies
\begin{align}\label{l3.1.35}
	\left|
	\Phi_{\eta}^{n}(\xi',z')
	\!-\!
	\Phi_{\eta}^{n}(\xi'',z'')
	\right|
	\leq &
	\left|
	\Phi_{\eta}^{n}(\xi',z')
	-
	\Phi_{\eta}^{n}(\xi'',z')
	\right|
	\nonumber\\
	&+
	\left|
	\Phi_{\eta}^{n}(\xi'',z')
	-
	\Phi_{\eta}^{n}(\xi'',z'')
	\right|
	\nonumber\\
	\leq &
	\tilde{C}_{2}(2n+1)\gamma_{4}^{n} \!+\! 2\tilde{C}_{3}\gamma_{4}^{n}
	\leq\! 
	\tilde{C}_{4}n\gamma_{4}^{n}.
\end{align}

Let $C_{7}=\tilde{C}_{4}/(1-\gamma_{4})^{2}$.
Moreover, let
\begin{align*}
	\phi_{\eta}(\xi,z)
	=
	\Phi_{\eta}^{0}(\xi,z)
	+
	\sum_{n=0}^{\infty}
	\left(
	\Phi_{\eta}^{n+1}(\xi,z)
	-
	\Phi_{\eta}^{n}(\xi,z)
	\right).
\end{align*}
Then, due to (\ref{l3.1.33}),
$\phi_{\eta}(\xi,z)$ is well-defined.
Now, (\ref{l3.1.31}) implies
\begin{align}\label{l3.1.37}
	\left|
	\Phi_{\eta}^{n}(\xi,z)
	-
	\phi_{\eta}(\xi,z)
	\right|
	\leq &
	\sum_{k=n}^{\infty}
	\left|
	\Phi_{\eta}^{k+1}(\xi,z)
	-
	\Phi_{\eta}^{k}(\xi,z)
	\right|
	\nonumber\\
	\leq &
	\tilde{C}_{4}
	\sum_{k=n}^{\infty} k\gamma_{4}^{k}
	\leq
	C_{7}n\gamma_{4}^{n}.
\end{align}
Consequently, (\ref{l3.1.35}) yields
\begin{align*}
	\left|
	\phi_{\eta}(\xi',z')
	-
	\phi_{\eta}(\xi'',z'')
	\right|
	\leq &
	\left|
	\Phi_{\eta}^{n}(\xi',z')
	-
	\Phi_{\eta}^{n}(\xi'',z'')
	\right|
	\\
	&+
	\left|
	\Phi_{\eta}^{n}(\xi',z')
	-
	\phi_{\eta}(\xi',z')
	\right|
	\\
	&+
	\left|
	\Phi_{\eta}^{n}(\xi'',z'')
	-
	\phi_{\eta}(\xi'',z'')
	\right|
	\\
	\leq &
	3C_{7}n\gamma_{4}^{n}.
\end{align*}
Therefore,
$\phi_{\eta}(\xi',z')=\phi_{\eta}(\xi'',z'')$
for any $\xi',\xi''\in V_{\delta_{9}}({\cal P}({\cal X}) )$,
$z',z''\in{\cal Z}$.
Hence, there exists a function $\phi_{\eta}$
which maps $\eta\in\mathbb{C}^{d}$ to $\mathbb{C}$
and satisfies
$\phi_{\eta}=\phi_{\eta}(\xi,z)$
for all $\eta\in V_{\delta_{9}}(\Theta)$,
$\xi\in V_{\delta_{9}}({\cal P}({\cal X}) )$,
$z\in{\cal Z}$.
Then, using (\ref{l3.1.37}), we conclude that
(\ref{l3.1.1*}) holds for $\eta\in V_{\delta_{9}}(\Theta)$,
$\xi\in V_{\delta_{9}}({\cal P}({\cal X}) )$,
$z\in{\cal Z}$.
\end{proof}

\begin{lemma}\label{lemma3.2}
(i) Let Assumptions \ref{a11} -- \ref{a12} and \ref{a2} hold.
Then, integral
\begin{align}\label{l3.2.1*}
	\int\cdots\int
	\Phi_{\eta}\left(
	F_{\eta}^{n}(\lambda,z_{1:n}), z_{n+1}
	\right)
	S(z_{n},{\rm d}z_{n+1})\cdots S(z,{\rm d}z_{1})
\end{align}
is analytic in $\eta$ for all $\eta\in V_{\delta_{5}}(\Theta)$, $\lambda\in{\cal P}({\cal X})$,
$z\in{\cal Z}$, $n\geq 1$
($\delta_{5}$ is specified in Lemmas \ref{lemma1.4}, \ref{lemma1.6}).

(ii) Let Assumptions \ref{a11} -- \ref{a3} hold. Then, integral
\begin{align}\label{l3.2.3*}
	\int\cdots\int
	\Phi_{\eta}\left(
	F_{\eta}^{n}(\lambda,z_{1:n}), z_{n+1}
	\right)
	S_{\eta}(z_{n},{\rm d}z_{n+1})\cdots S_{\eta}(z,{\rm d}z_{1})
\end{align}
is analytic in $\eta$ for all $\eta\in V_{\delta_{5}}(\Theta)$, $\lambda\in{\cal P}({\cal X})$,
$z\in{\cal Z}$, $n\geq 1$.
\end{lemma}

\begin{proof}
Throughout the proof, the following notations is used.
$\tilde{\phi}(z)$ is the function defined by $\tilde{\phi}(z)=\phi(y)$
for $x\in{\cal X}$, $y\in{\cal Y}$ and $z=(y,x)$.
$\eta$ is any element of $V_{\delta_{5}}(\Theta)$,
while $\lambda$ is any element in ${\cal P}({\cal X} )$.
$\{x_{n} \}_{n\geq 0}$, $\{y_{n} \}_{n\geq 0}$ are any sequences in ${\cal X}$, ${\cal Y}$
(respectively), while $\{z_{n} \}_{n\geq 0}$ is the sequence defined by
$z_{n} = (y_{n}, x_{n} )$ for $n\geq 0$
(notice that $\{z_{n} \}_{n\geq 0}$ can be any sequence in ${\cal Z}$).
$n\geq 1$ is any integer.

Using Lemmas \ref{lemma1.21}, \ref{lemma1.6}, we conclude
\begin{align*}
	\Phi_{\eta}\left(
	F_{\eta}^{n}(\lambda,z_{1:n}), z_{n+1}
	\right)
	=&
	\Phi_{\eta,y_{n+1}}\left(
	F_{\eta,\boldsymbol y}^{0:n}(\lambda)
	\right)
	\\
	=&
	\log
	\left\langle R_{\eta,\boldsymbol y}^{n:n+1}
	\left(F_{\eta,\boldsymbol y}^{0:n}(\lambda) \right) \right\rangle
	\\
	=&
	\log
	\left(
	\frac{\left\langle R_{\eta,\boldsymbol y}^{0:n+1}
	\left(\lambda\right) \right\rangle}
	{\left\langle R_{\eta,\boldsymbol y}^{0:n}
	\left(\lambda\right) \right\rangle}
	\right),
\end{align*}
where $\boldsymbol y = \{y'_{k} \}_{k\geq 1}$ is any sequence in ${\cal Y}$
satisfying $y'_{k}=y_{k}$ for $1\leq k\leq n+1$.
Combining this with Lemmas \ref{lemma2.2}, \ref{lemma1.2}, \ref{lemma1.6},
we deduce that
$\Phi_{\eta}\left(F_{\eta}^{n}(\lambda,z_{1:n}), z_{n+1} \right)$,
$u_{\eta}^{n}(z_{0:n})$
are analytic in $\eta$ for each $\eta\in V_{\delta_{5}}(\Theta)$.
Moreover, due to Lemmas \ref{lemma2.2}, \ref{lemma1.3}, \ref{lemma1.6}, we have
\begin{align}
	&\label{l3.2.1'}
	\left|
	\Phi_{\eta}\left(
	F_{\eta}^{n}(\lambda,z_{1:n}), z_{n+1}
	\right)
	\right|
	\leq C_{4}\tilde{\psi}(z_{n+1}),
	\\
	&\label{l3.2.1''}
	\left|
	u_{\eta}^{n}(z_{0:n})
	\right|
	\leq
	K_{n}\prod_{k=1}^{n} \tilde{\phi}(z_{k})
\end{align}
($\tilde{\psi}(z)$ is specified in (\ref{5.5})).

Owing to Assumption \ref{a2}, we have
\begin{align*}
	&
	\int\cdots\int
	\tilde{\psi}(z_{n+1}) S(z_{n},{\rm d}z_{n+1})\cdots S(z_{0},{\rm d}z_{1})
	\\
	&
	\begin{aligned}
	=&
	\int\int\cdots\int
	(1+\psi(y_{n+1})) Q(x_{n+1},{\rm d}y_{n+1})
	\\
	&\cdot
	P(x_{n},{\rm d}x_{n+1})\cdots P(x_{0},{\rm d}x_{1})
	\end{aligned}
	\\
	&
	\leq
	K+1
	<\infty.
\end{align*}
Consequently, Lemma \ref{lemmaa1} (see Appendix \ref{appendix1}) and (\ref{l3.2.1'}) imply
that integral (\ref{l3.2.1*}) is analytic in $\eta$
for each $\eta\in V_{\delta_{5}}(\Theta)$.

Relying on (\ref{5.71}), it is easy to show
\begin{align*}
	&
	\int\!\cdots\!\int\!
	\Phi_{\eta}\left(
	F_{\eta}^{n}(\lambda,z_{1:n}), z_{n+1}
	\right)
	S_{\eta}(z_{n},{\rm d}z_{n+1})\cdots S_{\eta}(z_{0},{\rm d}z_{1})
	\\
	&
	\begin{aligned}
	=&
	\int\cdots\int
	\Phi_{\eta}\left(
	F_{\eta}^{n}(\lambda,z_{1:n}), z_{n+1}
	\right)
	u_{\eta}^{n+1}(z_{0:n+1} )
	\\
	&\cdot
	(\nu\times\mu)({\rm d}z_{n+1} ) \cdots (\nu\times\mu)({\rm d}z_{1} ).
	\end{aligned}
\end{align*}
Moreover, due to Assumptions \ref{a1}, \ref{a3}, we have
\begin{align*}
	&
	\begin{aligned}
	&
	\int\cdots\int
	\tilde{\psi}(z_{n+1}) \left( \prod_{k=1}^{n+1} \tilde{\phi}(z_{k}) \right)
	\\
	&\cdot
	(\nu\times\mu)({\rm d}z_{n+1})\cdots(\nu\times\mu)({\rm d}z_{1})
	\end{aligned}
	\\
	&
	\begin{aligned}
	=&
	\|\mu\|^{n+1}
	\left(
	\int (1+\psi(y_{n+1})) \phi(y_{n+1}) \nu({\rm d}y_{n+1})
	\right)
	\\
	&\cdot
	\left( \prod_{k=1}^{n} \int \phi(y_{k}) \nu({\rm d}y_{k})\right)
	\end{aligned}
	\\
	&
	<\infty.
\end{align*}
Consequently, Lemma \ref{lemmaa1} (see Appendix \ref{appendix1})
and (\ref{l3.2.1'}), (\ref{l3.2.1''}) imply
that integral (\ref{l3.2.3*}) is analytic in $\eta$
for $\eta\in V_{\delta_{5}}(\Theta)$.
\end{proof}

\begin{proof}[\rm\bf Proof of Theorem \ref{theorem1}]
Let $T_{\eta}(z,{\rm d}z')$ be the kernel defined by $T_{\eta}(z,B)=S(z,B)$
for $\eta\in\mathbb{C}^{d}$, $z\in{\cal Z}$, $B\in{\cal B}({\cal Z})$
($S(z,{\rm d}z')$ is specified in (\ref{5.1'})).
Moreover, let $T_{\eta}^{n}(z,{\rm d}z')$, $\Phi_{\eta}^{n}(\lambda,z)$ have the same meaning as in (\ref{3.7*}).
Then, owing to Lemma \ref{lemma3.2}, $\Phi_{\eta}^{n}(\lambda,z)$
is analytic in $\eta$ for each $\eta\in V_{\delta_{5}}(\Theta)$,
$\lambda\in{\cal P}({\cal X})$, $z\in{\cal Z}$, $n\geq 1$.
Moreover, due to Lemma \ref{lemma2.3}, kernel
$T_{\eta}(z,{\rm d}z')$ (defined here)
satisfies Assumptions \ref{a4}, \ref{a5}.
Combining this with Lemma \ref{lemma3.1}, we deduce that
there exist a function $\phi_{\eta}$
mapping $\eta\in\mathbb{C}^{d}$ to $\mathbb{C}$ and
real numbers $\delta_{9}\in(0,\delta_{5}]$, $\gamma_{4}\in(0,1)$,
$C_{7}\in[1,\infty)$ such that (\ref{l3.1.1*}) holds
for $\eta\in V_{\delta_{9}}(\Theta)$, $\lambda\in{\cal P}({\cal X})$,
$z\in{\cal Z}$, $n\geq 1$.
Since the limit of uniformly convergent analytic functions is also analytic
(see e.g., \cite[Theorem 2.4.1]{taylor}),
$\phi_{\eta}$ is analytic in $\eta$
for each $\eta\in V_{\delta_{9}}(\Theta)$.

In what follows in the proof, $\theta$, $\lambda$, $z$ are any elements of
$\Theta$, ${\cal P}({\cal X} )$, ${\cal Z}$ (respectively),
while $n\geq 1$ is any integer.
It is straightforward to verify
\begin{align*}
	\Phi_{\theta}^{n}(\lambda,z)
	=
	E\left(\left.
	\Phi_{\theta}\left(
	F_{\theta}^{n}(\lambda,Z_{1:n}), Z_{n+1}
	\right)
	\right|
	Z_{0}=z
	\right),
\end{align*}
where $Z_{n}=(Y_{n},X_{n})$.
Therefore, (\ref{3.1*}) yields
\begin{align*}
	E\left(
	\log q_{\theta}^{n}(Y_{1:n}|\lambda)
	\right)
	=
	\sum_{k=0}^{n-1}
	E\left( \Phi_{\theta}^{k}(\lambda,Z_{0}) \right).
\end{align*}
Then, 
Lemma \ref{lemma3.1} implies
\begin{align*}
	\left|
	E\left(\frac{1}{n} \log q_{\theta}(Y_{1:n}|\lambda) \right)
	-
	\phi_{\theta}
	\right|
	\leq &
	\frac{1}{n}
	\sum_{k=0}^{n-1}
	E\left|
	\Phi_{\theta}^{k}(\lambda,Z_{0}) - \phi_{\theta}
	\right|
	\\
	\leq &
	\frac{C_{7}}{n} \sum_{k=0}^{n-1} \gamma_{4}^{k}
	\leq 
	\frac{C_{7}}{n(1-\gamma_{4}) }.
\end{align*}
Consequently, there exists a function $l:\Theta\rightarrow\mathbb{R}$
with the properties specified in the statement of the theorem.
\end{proof}

\begin{proof}[\rm\bf Proof of Theorem \ref{theorem2}]
Let $T_{\eta}(z,{\rm d}z')$ be the kernel defined by $T_{\eta}(z,B)=S_{\eta}(z,B)$
for $\eta\in\mathbb{C}^{d}$, $z\in{\cal Z}$, $B\in{\cal B}({\cal Z})$
($S_{\eta}(z,{\rm d}z')$ is specified in (\ref{5.1''})).
Moreover, let $T_{\eta}^{n}(z,{\rm d}z')$, $\Phi_{\eta}^{n}(\lambda,z)$ have the same meaning as in (\ref{3.7*}).
Then, due to Lemma \ref{lemma3.2}, $\Phi_{\eta}^{n}(\lambda,z)$
is analytic in $\eta$ for each $\eta\in V_{\delta_{5}}(\Theta)$,
$\lambda\in{\cal P}({\cal X})$, $z\in{\cal Z}$, $n\geq 1$.
Moreover, Lemma \ref{lemma2.1} implies that Assumptions \ref{a4}, \ref{a5}
hold for kernel $T_{\eta}(z,{\rm d}z')$ (defined here).
Combining this with Lemma \ref{lemma3.1}, we conclude that
there exist a function $\phi_{\eta}$
mapping $\eta\in\mathbb{C}^{d}$ to $\mathbb{C}$ and
real numbers $\delta_{9}\in(0,\delta_{5}]$, $\gamma_{4}\in(0,1)$,
$C_{7}\in[1,\infty)$ such that (\ref{l3.1.1*}) holds
for $\eta\in V_{\delta_{9}}(\Theta)$, $\lambda\in{\cal P}({\cal X})$,
$z\in{\cal Z}$, $n\geq 1$.
As the limit of uniformly convergent analytic functions is also analytic
(see e.g., \cite[Theorem 2.4.1]{taylor}),
$\phi_{\eta}$ is analytic in $\eta$
for each $\eta\in V_{\delta_{9}}(\Theta)$.

In the rest of the proof, $\theta$, $\lambda$, $z$ are any elements of
$\Theta$, ${\cal P}({\cal X} )$, ${\cal Z}$ (respectively),
while $n\geq 1$ is any integer.
It is easy to show
\begin{align*}
	\Phi_{\theta}^{n}(\lambda,z)
	=
	E\left(\left.
	\Phi_{\theta}\left(
	F_{\theta}^{n}\big(\lambda,Z_{1:n}^{\theta,\lambda} \big), Z_{n+1}^{\theta,\lambda}
	\right)
	\right|
	Z_{0}^{\theta,\lambda}=z
	\right),
\end{align*}
where $Z_{n}^{\theta,\lambda}=
\left(Y_{n}^{\theta,\lambda},X_{n}^{\theta,\lambda}\right)$.
Then, (\ref{3.1*}) yields
\begin{align*}
	E\left(
	\log q_{\theta}^{n}\big(Y_{1:n}^{\theta,\lambda}\big|\lambda\big)
	\right)
	=
	\sum_{k=0}^{n-1}
	E\left( \Phi_{\theta}^{k}\big(\lambda,Z_{0}^{\theta,\lambda}\big) \right).
\end{align*}
Therefore, 
Lemma \ref{lemma3.1} implies
\begin{align*}
	\left|
	E\left(\frac{1}{n} \log q_{\theta}\big(Y_{1:n}^{\theta,\lambda}\big|\lambda\big) \right)
	-
	\phi_{\theta}
	\right|
	\!\leq\! &
	\frac{1}{n}
	\sum_{k=0}^{n-1}
	E\left|
	\Phi_{\theta}^{k}\big(\lambda,Z_{0}^{\theta,\lambda}\big) - \phi_{\theta}
	\right|
	\\
	\leq &
	\frac{C_{7}}{n} \sum_{k=0}^{n-1} \gamma_{4}^{k}
	\leq 
	\frac{C_{7}}{n(1-\gamma_{4}) }.
\end{align*}
Consequently, there exists a function $h:\Theta\rightarrow\mathbb{R}$
with the properties specified in the statement of the theorem.
\end{proof}

\section{Proof of Corollaries \ref{corollary1b} -- \ref{corollary2c}}\label{section4*}

\begin{proof}[\rm\bf Proof of Corollaries \ref{corollary1b} and \ref{corollary2b}]
Let $\tilde{\Theta}$ be any non-empty bounded open set satisfying
$\text{cl}\tilde{\Theta}\subset\Theta$.
As $\text{cl}\tilde{\Theta}$, ${\cal X}$ are compact sets,
Assumption \ref{b2} and Lemma \ref{lemmaa2} (see Appendix \ref{appendix1}) imply that there exist functions
$\{\hat{a}_{\eta}^{i}(x)\}_{1\leq i\leq N_{x}}$,
$\{\hat{b}_{\eta}^{j}(x) \}_{1\leq j\leq N_{y}}$ with the following properties:

(i) $\{\hat{a}_{\eta}^{i}(x) \}_{1\leq i\leq N_{x}}$,
$\{\hat{b}_{\eta}^{j}(x) \}_{1\leq j\leq N_{y}}$ map
$\eta\in\mathbb{C}^{d}$, $x\in\mathbb{C}^{d_{x}}$ to $\mathbb{C}$.

(ii) $\hat{a}_{\theta}^{i}(x)=a_{\theta}^{i}(x)$,
$\hat{b}_{\theta}^{j}(x)=b_{\theta}^{j}(x)$ for $\theta\in\tilde{\Theta}$,
$x\in{\cal X}$, $1\leq i\leq N_{x}$, $1\leq j\leq N_{y}$.

(iii) There exists a real number $\alpha_{1}\in(0,1)$ such that
$\hat{a}_{\eta}^{i}(x)$, $\hat{b}_{\eta}^{j}(x)$ are analytic in $(\eta,x)$
for $\eta\in V_{\alpha_{1}}(\tilde{\Theta})$, $x\in V_{\alpha_{1}}({\cal X})$,
$1\leq i\leq N_{x}$, $1\leq j\leq N_{y}$.

Owing to Assumption \ref{b2}, $\{\hat{a}_{\theta}^{i}(x) \}_{1\leq i\leq N_{x}}$,
$\{\hat{b}_{\theta}^{j}(x) \}_{1\leq j\leq N_{y}}$ are positive and uniformly
bounded away from zero for $\theta\in\text{cl}\tilde{\Theta}$, $x\in{\cal X}$.
Then, due to (iii), there exist real numbers $\alpha\in(0,\alpha_{1})$,
$\beta\in(0,1)$ such that
\begin{align}
	\label{c1b.1'}
	\text{Re}\left\{ \hat{a}_{\eta}^{i}(x) \right\}
	\geq \beta,
	&\;\;\;\;\;\;\;
	|\hat{a}_{\eta}^{i}(x) |\leq \frac{1}{\beta},
	\\
	\label{c1b.1''}
	\text{Re}\left\{ \hat{b}_{\eta}^{j}(x) \right\}
	\geq \beta,
	&\;\;\;\;\;\;\;
	|\hat{b}_{\eta}^{j}(x) |\leq \frac{1}{\beta}
\end{align}
for $\eta\in V_{\alpha}(\tilde{\Theta})$, $x\in V_{\alpha}({\cal X})$,
$1\leq i\leq N_{x}$, $1\leq j\leq N_{y}$.

Let $\hat{p}_{\eta}(x'|x)$, $\hat{q}_{\eta}(y|x)$ be the functions defined by
\begin{align*}
	&
	\hat{p}_{\eta}(x'|x)
	=
	\sum_{i=1}^{N_{x}} \hat{a}_{\eta}^{i}(x) v_{i}(x'),
	\\
	&
	\hat{q}_{\eta}(y|x)
	=
	\sum_{j=1}^{N_{y}} \hat{b}_{\eta}^{j}(x) w_{j}(y)
\end{align*}
for $\eta\in\mathbb{C}^{d}$, $x,x'\in{\cal X}$, $y\in{\cal Y}$,
while
\begin{align*}
	&
	r_{\theta}(y,x'|x)
	=
	q_{\theta}(y|x')p_{\theta}(x'|x),
	\\
	&
	\hat{r}_{\eta}(y,x'|x)
	=
	\hat{q}_{\eta}(y|x')\hat{p}_{\eta}(x'|x)
\end{align*}
for the same $\eta,x,x',y$ and $\theta\in\Theta$.
Then, owing to (ii), (iii),
$\hat{r}_{\eta}(y,x'|x)$ is analytic in $\eta$
for each $\eta\in V_{\alpha}(\tilde{\Theta})$,
$x,x'\in{\cal X}$, $y\in{\cal Y}$.
For similar reasons,
$\hat{r}_{\theta}(y,x'|x)=r_{\theta}(y,x'|x)$
for the same $x,x',y$ and $\theta\in\tilde{\Theta}$.
Moreover, Assumption \ref{b3} and (\ref{c1b.1'}) imply
\begin{align}
	&\label{c1b.3'}
	|\hat{p}_{\eta}(x'|x)|
	\geq
	\sum_{i=1}^{N_{x}}
	\text{Re}\left\{ \hat{a}_{\eta}^{i}(x) \right\} v_{i}(x')
	\geq
	\beta\varepsilon N_{x},
	\\
	&\label{c1b.3''}
	|\hat{p}_{\eta}(x'|x)|
	\leq
	\sum_{i=1}^{N_{x}} |\hat{a}_{\eta}^{i}(x)| v_{i}(x')
	\leq
	\frac{N_{x}}{\beta\varepsilon}
\end{align}
for $\eta\in V_{\alpha}(\tilde{\Theta})$, $x,x'\in{\cal X}$.
Similarly, (\ref{c1b.1''}) yields
\begin{align}
	&\label{c1b.5'}
	|\hat{q}_{\eta}(y|x)|
	\geq
	\sum_{j=1}^{N_{y}}
	\text{Re}\left\{ \hat{b}_{\eta}^{j}(x) \right\} w_{j}(y)
	\geq
	\beta\sum_{j=1}^{N_{y}} w_{j}(y),
	\\
	&\label{c1b.5''}
	|\hat{q}_{\eta}(y|x)|
	\leq
	\sum_{j=1}^{N_{y}} |\hat{b}_{\eta}^{j}(x)| w_{j}(y)
	\leq
	\frac{1}{\beta} \sum_{j=1}^{N_{y}} w_{j}(y)
\end{align}
for the same $\eta,x$ and $y\in{\cal Y}$.

Let $\tilde{C}_{1}=\beta^{-2}\varepsilon^{-1}N_{x}$,
$\tilde{C}_{2}=\tilde{C}_{1}N_{x}$,
$\gamma=\beta^{4}\varepsilon^{2}$.
Moreover, let $\phi(y)$, $\psi(y)$ be the functions defined by
\begin{align*}
	&
	\phi(y)
	=
	\tilde{C}_{1} \sum_{j=1}^{N_{y}} w_{j}(y),
	\\
	&
	\psi(y)
	=
	\tilde{C}_{2}\left(1 + \sum_{j=1}^{N_{y}} |\log w_{j}(y)| \right)
\end{align*}
for $y\in{\cal Y}$.
Then, combining (\ref{c1b.3'}) -- (\ref{c1b.5''}), we get
\begin{align}\label{c1b.7}
	\gamma\phi(y)\leq |\hat{r}_{\eta}(y,x'|x)| \leq \frac{\phi(y)}{\gamma}
\end{align}
for $\eta\in V_{\alpha}(\tilde{\Theta})$, $x,x'\in{\cal X}$, $y\in{\cal Y}$.
We also get
\begin{align*}
	\log\phi(y)
	\leq &
	\log(\tilde{C}_{1}N_{x} )
	+
	\max_{1\leq j\leq N_{y} }\log w_{j}(y)
	\\
	\leq &
	\tilde{C}_{1}N_{x}
	\left(1 + \sum_{j=1}^{N_{y}} |\log w_{j}(y)|  \right),
	\\
	\log\phi(y)
	\geq &
	\log(\tilde{C}_{1}N_{x} )
	+
	\min_{1\leq j\leq N_{y} }\log w_{j}(y)
	\\
	\geq &
	-\tilde{C}_{1}N_{x}
	\left(1 + \sum_{j=1}^{N_{y}} |\log w_{j}(y)|  \right).
\end{align*}
Therefore, we have
\begin{align}\label{c1b.9}
	|\log\phi(y)|\leq \psi(y).
\end{align}
Since $\int\phi(y)\nu({\rm d}y)=\tilde{C}_{1}N_{y}<\infty$,
(\ref{c1b.7}), (\ref{c1b.9}) imply that
Assumptions \ref{a11} -- \ref{a12} follow from Assumptions \ref{c1} -- \ref{c3}
when $\Theta$ is restricted to $\tilde{\Theta}$
(i.e., when $\Theta$ is replaced with $\tilde{\Theta}$).

Owing to Assumption \ref{b5}, we have
\begin{align*}
	\int\psi(y)Q(x,{\rm d}y)
	=&
	\tilde{C}_{2} + \sum_{j=1}^{N_{y}}|\log w_{j}(y)|Q(x,{\rm d}y)
	\\
	\leq &
	\tilde{C}_{2}+KN_{y}
	<\infty.
\end{align*}
Hence, Assumption \ref{a2} results from Assumption \ref{b5}.
Moreover, due to Assumption \ref{b4}, we have
\begin{align*}
	\int \psi(y) \phi(y) \nu({\rm d}y)
	= &
	\tilde{C}_{1}\tilde{C}_{2}
	\sum_{j,k=1}^{N_{y} }
	\int |\log w_{j}(y) | w_{k}(y) \nu({\rm d}y)
	\\
	&+
	\tilde{C}_{1}\tilde{C}_{2}N_{y}
	<\infty.
\end{align*}
Thus, Assumption \ref{a3} results from Assumption \ref{b4}.

Using Theorems \ref{theorem1}, \ref{theorem2}, we conclude that there exist
functions $\tilde{l},\tilde{h}:\tilde{\Theta}\rightarrow\mathbb{R}$
such that $\tilde{l}(\theta)$, $\tilde{h}(\theta)$ are real-analytic in $\theta$
and satisfy $\lim_{n\rightarrow\infty}l_{n}(\theta,\lambda)=\tilde{l}(\theta)$,
$\lim_{n\rightarrow\infty}h_{n}(\theta,\lambda)=\tilde{h}(\theta)$
for each $\theta\in\tilde{\Theta}$, $\lambda\in{\cal P}({\cal X})$
($l_{n}(\theta,\lambda)$, $h_{n}(\theta,\lambda)$ have the same meaning as in
(\ref{1.1})).
Consequently, Corollaries \ref{corollary1c}, \ref{corollary2c} hold. We use here the representation $\Theta=\bigcup_{n=1}^{\infty}\tilde{\Theta}_{n}$,
where $\{\tilde{\Theta}_{n} \}_{n\geq 1}$ is a sequence of non-empty open balls satisfying
$\text{cl}\tilde{\Theta}_{n}\subset\Theta$ for $n\geq 1$.
\end{proof}

\begin{proof}[\rm\bf Proof of Corollaries \ref{corollary1c} and \ref{corollary2c}]
Let $\tilde{\Theta}$ be a non-empty bounded open set satisfying
$\text{cl}\tilde{\Theta}\subset\Theta$.
As $\text{cl}\tilde{\Theta}$, ${\cal X}$, ${\cal Y}$ are compact
and $A_{\theta}(x)$, $B_{\theta}^{-1}(x)$, $C_{\theta}(x)$, $D_{\theta}^{-1}(x)$
are continuous in $(\theta,x)$,
it follows from Assumption \ref{c3} that
there exists a real number $r\in[1,\infty)$ such that
\begin{align}
	&\label{c1c.1'}
	\left\|
	B_{\theta}^{-1}(x) \left( x'-A_{\theta}(x) \right)
	\right\|
	\leq
	r,
	\\
	&\label{c1c.1''}
	\left\|
	D_{\theta}^{-1}(x) \left( y-C_{\theta}(x) \right)
	\right\|
	\leq
	r
\end{align}
for $\theta\in\text{cl}\tilde{\Theta}$, $x,x'\in{\cal X}$, $y\in{\cal Y}$.

Let $\tilde{\cal X}=\{x\in\mathbb{R}^{d_{x}}: \|x\|\leq r \}$,
$\tilde{\cal Y}=\{y\in\mathbb{R}^{d_{y}}: \|y\|\leq r \}$.
Since $\text{cl}\tilde{\Theta}$, ${\cal X}$, $\tilde{\cal X}$, $\tilde{\cal Y}$
are compact sets, Assumptions \ref{c2}, \ref{c3} and Lemma \ref{lemmaa2} (see Appendix \ref{appendix1})
imply that there exist functions
$\hat{A}_{\eta}(x)$, $\hat{B}_{\eta}(x)$, $\hat{C}_{\eta}(x)$, $\hat{D}_{\eta}(x)$
and $\hat{v}(x)$, $\hat{w}(y)$ with the following properties:

(i) $\hat{A}_{\eta}(x)$, $\hat{B}_{\eta}(x)$, $\hat{C}_{\eta}(x)$, $\hat{D}_{\eta}(x)$
map $\eta\in\mathbb{C}^{d}$, $x\in\mathbb{C}^{d_{x}}$ to
$\mathbb{C}^{d_{x}}$, $\mathbb{C}^{d_{x}\times d_{x}}$,
$\mathbb{C}^{d_{y}}$, $\mathbb{C}^{d_{y}\times d_{y}}$ (respectively),
while $\hat{v}(x)$, $\hat{w}(y)$ map $x\in\mathbb{C}^{d_{x}}$, $y\in\mathbb{C}^{d_{y}}$
to $\mathbb{C}$.

(ii) $\hat{A}_{\theta}(x)=A_{\theta}(x)$, $\hat{B}_{\theta}(x)=B_{\theta}(x)$,
$\hat{C}_{\theta}(x)=C_{\theta}(x)$, $\hat{D}_{\theta}(x)=D_{\theta}(x)$
for $\theta\in\tilde{\Theta}$, $x\in{\cal X}$,
and $\hat{v}(x)=v(x)$, $\hat{w}(y)=w(y)$ for $x\in\tilde{\cal X}$, $y\in\tilde{\cal Y}$.

(iii) There exists a real number $\alpha_{1}\in(0,1)$ such that
$\hat{A}_{\eta}(x)$, $\hat{B}_{\eta}(x)$, $\hat{C}_{\eta}(x)$, $\hat{D}_{\eta}(x)$
are analytic in $(\eta,x)$ for $\eta\in V_{\alpha_{1}}(\tilde{\Theta})$,
$x\in V_{\alpha_{1}}({\cal X})$.

(iv) There exists a real number $\alpha_{2}\in(0,1)$ such that
$\hat{v}(x)$, $\hat{w}(y)$
are analytic in $x$, $y$ (respectively) for $x\in V_{\alpha_{2}}(\tilde{\cal X})$,
$y\in V_{\alpha_{2}}(\tilde{\cal Y})$.

Since $|\text{det}\hat{B}_{\theta}(x)|$, $|\text{det}\hat{D}_{\theta}(x)|$ are uniformly
bounded away from zero for $\theta\in\text{cl}\tilde{\Theta}$, $x\in{\cal X}$,
Assumption \ref{c3} and (iii) imply that there exists a real number $\alpha_{3}\in(0,\alpha_{1})$
such that $\text{det}\hat{B}_{\eta}(x)\neq 0$, $\text{det}\hat{D}_{\eta}(x)\neq 0$
for $\eta\in V_{\alpha_{3}}(\tilde{\Theta})$, $x\in V_{\alpha_{3}}({\cal X})$.
Therefore,
\begin{align*}
	&\hat{B}_{\eta}^{-1}(x)\left( x'-\hat{A}_{\eta}(x) \right),
	\\
	&\hat{D}_{\eta}^{-1}(x)\left( y-\hat{C}_{\eta}(x) \right)
\end{align*}
are well-defined and analytic in $(\eta,x,x',y)$ for
$\eta\in V_{\alpha_{3}}(\tilde{\Theta})$, $x\in V_{\alpha_{3}}({\cal X})$,
$x'\in\mathbb{C}^{d_{x}}$, $y\in\mathbb{C}^{d_{y}}$.
As $\text{cl}\tilde{\Theta}$, ${\cal X}$, ${\cal Y}$ are compact sets,
it follows from (\ref{c1c.1'}), (\ref{c1c.1''}) that
there exists a real number $\alpha_{4}\in(0,\alpha_{3})$ such that
\begin{align*}
	&
	\left\|
	\text{Re}\left\{ \hat{B}_{\eta}^{-1}(x)\left( x'-\hat{A}_{\eta}(x) \right) \right\}
	\right\|
	<
	r+\frac{\alpha_{2}}{2},
	\\
	&
	\left\|
	\text{Im}\left\{ \hat{B}_{\eta}^{-1}(x)\left( x'-\hat{A}_{\eta}(x) \right) \right\}
	\right\|
	<
	\frac{\alpha_{2}}{2},
	\\
	&
	\left\|
	\text{Re}\left\{ \hat{D}_{\eta}^{-1}(x)\left( y-\hat{C}_{\eta}(x) \right) \right\}
	\right\|
	<
	r+\frac{\alpha_{2}}{2},
	\\
	&
	\left\|
	\text{Im}\left\{ \hat{D}_{\eta}^{-1}(x)\left( y-\hat{C}_{\eta}(x) \right) \right\}
	\right\|
	<
	\frac{\alpha_{2}}{2}
\end{align*}
for $\eta\in V_{\alpha_{4}}(\tilde{\Theta})$, $x,x'\in V_{\alpha_{4}}({\cal X})$,
$y\in V_{\alpha_{4}}({\cal Y})$.
Hence, we have
\begin{align*}
	&
	\hat{B}_{\eta}^{-1}(x)\left( x'-\hat{A}_{\eta}(x) \right)
	\in V_{\alpha_{2}}(\tilde{\cal X}),
	\\
	&
	\hat{D}_{\eta}^{-1}(x)\left( y-\hat{C}_{\eta}(x) \right)
	\in V_{\alpha_{2}}(\tilde{\cal Y})
\end{align*}
for the same $\eta,x,x',y$.
Consequently, (iv) yields that
\begin{align}
	&\label{c1c.3'}
	\hat{v}\left( \hat{B}_{\eta}^{-1}(x)\left( x'-\hat{A}_{\eta}(x) \right) \right),
	\\
	&\label{c1c.3''}
	\hat{w}\left( \hat{D}_{\eta}^{-1}(x)\left( y-\hat{C}_{\eta}(x) \right) \right)
\end{align}
are analytic in $(\eta,x,x',y)$
for $\eta\in V_{\alpha_{4}}(\tilde{\Theta})$, $x,x'\in V_{\alpha_{4}}({\cal X})$,
$y\in V_{\alpha_{4}}({\cal Y})$.
Since functions (\ref{c1c.3'}), (\ref{c1c.3''}) are positive and uniformly bounded away from zero
for $\eta\in\text{cl}\tilde{\Theta}$, $x,x'\in{\cal X}$, $y\in{\cal Y}$,
Assumption \ref{c2} implies that there exist real numbers
$\alpha\in(0,\alpha_{4})$, $\beta\in(0,1)$ such that
\begin{align}
	&\label{c1c.5'}
	\text{Re}\left\{ \hat{v}\left( \hat{B}_{\eta}^{-1}(x)\left( x'-\hat{A}_{\eta}(x) \right) \right) \right\}
	\geq
	\beta,
	\\
	&\label{c1c.5''}
	\left|
	\hat{v}\left( \hat{B}_{\eta}^{-1}(x)\left( x'-\hat{A}_{\eta}(x) \right) \right)
	\right|
	\leq
	\frac{1}{\beta},
	\\
	&\label{c1c.7'}
	\text{Re}\left\{ \hat{w}\left( \hat{D}_{\eta}^{-1}(x)\left( y-\hat{C}_{\eta}(x) \right) \right) \right\}
	\geq
	\beta,
	\\
	&\label{c1c.7''}
	\left|
	\hat{w}\left( \hat{D}_{\eta}^{-1}(x)\left( y-\hat{C}_{\eta}(x) \right) \right)
	\right|
	\leq
	\frac{1}{\beta}
\end{align}
for $\eta\in V_{\alpha}(\tilde{\Theta})$, $x,x'\in V_{\alpha}({\cal X})$,
$y\in V_{\alpha}({\cal Y})$.

Owing to (\ref{c1c.5'}), (\ref{c1c.5''}), we have
\begin{align}
	&\label{c1c.21}
	\left|
	\int_{\cal X} \hat{v}\left( \hat{B}_{\eta}^{-1}(x)\left( x'-\hat{A}_{\eta}(x) \right) \right) {\rm d}x'
	\right|
	\nonumber\\
	&\geq
	\int_{\cal X}
	\text{Re}\left\{ \hat{v}\left( \hat{B}_{\eta}^{-1}(x)\left( x'-\hat{A}_{\eta}(x) \right) \right) \right\}
	{\rm d}x'
	\geq
	\beta \text{m}({\cal X})
	> 0,
	\\
	&\label{c1c.23}
	\left|
	\int_{\cal X} \hat{v}\left( \hat{B}_{\eta}^{-1}(x)\left( x'-\hat{A}_{\eta}(x) \right) \right) {\rm d}x'
	\right|
	\nonumber\\
	&\leq
	\int_{\cal X}
	\left| \hat{v}\left( \hat{B}_{\eta}^{-1}(x)\left( x'-\hat{A}_{\eta}(x) \right) \right) \right|
	{\rm d}x'
	\leq
	\frac{\text{m}({\cal X})}{\beta}
\end{align}
for $\eta\in V_{\alpha}(\tilde{\Theta})$, $x\in V_{\alpha}({\cal X})$,
where $\text{m}({\cal X})$ is the Lebesgue measure of ${\cal X}$.
Similarly, due to (\ref{c1c.7'}), (\ref{c1c.7''}), we have
\begin{align}
	&\label{c1c.25}
	\left|
	\int_{\cal Y} \hat{w}\left( \hat{D}_{\eta}^{-1}(x)\left( y-\hat{C}_{\eta}(x) \right) \right) {\rm d}y
	\right|
	\nonumber\\
	&\geq
	\int_{\cal Y}
	\text{Re}\left\{ \hat{w}\left( \hat{D}_{\eta}^{-1}(x)\left( y-\hat{C}_{\eta}(x) \right) \right) \right\}
	{\rm d}y
	\geq
	\beta \text{m}({\cal Y})
	> 0,
	\\
	&\label{c1c.27}
	\left|
	\int_{\cal Y} \hat{w}\left( \hat{D}_{\eta}^{-1}(x)\left( y-\hat{C}_{\eta}(x) \right) \right) {\rm d}y
	\right|
	\nonumber\\
	&\leq
	\int_{\cal Y}
	\left| \hat{w}\left( \hat{D}_{\eta}^{-1}(x)\left( y-\hat{C}_{\eta}(x) \right) \right) \right|
	{\rm d}y
	\leq
	\frac{\text{m}({\cal Y})}{\beta}
\end{align}
for the same $\eta, x$,
where $\text{m}({\cal Y})$ is the Lebesgue measure of ${\cal Y}$.
Further to this, Lemma \ref{lemmaa1} (see Appendix \ref{appendix1})
and (\ref{c1c.5''}), (\ref{c1c.7''}) imply that
\begin{align}
	&\label{c1c.29'}
	\int_{\cal X} \hat{v}\left( \hat{B}_{\eta}^{-1}(x)\left( x'-\hat{A}_{\eta}(x) \right) \right) {\rm d}x',
	\\
	&\label{c1c.29''}
	\int_{\cal Y} \hat{w}\left( \hat{D}_{\eta}^{-1}(x)\left( y-\hat{C}_{\eta}(x) \right) \right) {\rm d}y
\end{align}
are analytic in $(\eta,x)$ for $\eta\in V_{\alpha}(\tilde{\Theta})$, $x\in V_{\alpha}({\cal X})$.

In the rest of the proof, the following notations is used.
$\hat{p}_{\eta}(x'|x)$, $\hat{q}_{\eta}(y|x)$ are the functions defined by
$\hat{p}_{\eta'}(x'|x) = 0$, $\hat{q}_{\eta'}(y|x) = 0$ and
\begin{align*}
	&
	\hat{p}_{\eta}(x'|x)
	=
	\frac{\hat{v}\left( \hat{B}_{\eta}^{-1}(x)\left( x'-\hat{A}_{\eta}(x) \right) \right) }
	{\int_{\cal X} \hat{v}\left( \hat{B}_{\eta}^{-1}(x)\left( x''-\hat{A}_{\eta}(x) \right) \right) {\rm d}x'' },
	\\
	&
	\hat{q}_{\eta}(y|x)
	=
	\frac{\hat{w}\left( \hat{D}_{\eta}^{-1}(x)\left( y-\hat{C}_{\eta}(x) \right) \right) }
	{\int_{\cal Y} \hat{w}\left( \hat{D}_{\eta}^{-1}(x)\left( y'-\hat{C}_{\eta}(x) \right) \right) {\rm d}y' }
\end{align*}
for $\eta\in V_{\alpha}(\tilde{\Theta})$, $\eta'\in \mathbb{C}^{d}\setminus V_{\alpha}(\tilde{\Theta})$,
$x,x'\in{\cal X}$, $y\in{\cal Y}$,
while
$r_{\theta}(y,x'|x)$, $\hat{r}_{\eta}(y,x'|x)$ are the functions defined by
\begin{align*}
	&
	r_{\theta}(y,x'|x)
	=
	q_{\theta}(y|x') p_{\theta}(x'|x),
	\\
	&
	\hat{r}_{\eta}(y,x'|x)
	=
	\hat{q}_{\eta}(y|x')\hat{p}_{\eta}(x'|x)
\end{align*}
for the same $x,x',y$ and $\theta\in\Theta$, $\eta\in\mathbb{C}^{d}$.

As functions (\ref{c1c.3'}), (\ref{c1c.3''}) and integrals (\ref{c1c.29'}), (\ref{c1c.29''}) are analytic in
$(\eta,x,x',y)$ for $\eta\in V_{\alpha}(\tilde{\Theta})$,
$x,x'\in V_{\alpha}({\cal X})$, $y\in V_{\alpha}({\cal Y})$,
it follows from (\ref{c1c.21}), (\ref{c1c.25}) that $\hat{r}_{\eta}(y,x'|x)$
is well-defined and analytic in $\eta$ for the same $\eta,x,x',y$.
Similarly, (\ref{c1c.21}) -- (\ref{c1c.27}) imply that there exists a real number $\gamma\in(0,1)$
such that
$\gamma \leq |\hat{r}_{\eta}(y,x'|x)| \leq 1/\gamma$
for $\eta\in V_{\alpha}(\tilde{\Theta})$,
$x,x'\in{\cal X}$, $y\in{\cal Y}$.
Further to this, (ii) yields
$\hat{r}_{\theta}(y,x'|x)=r_{\theta}(y,x'|x)$
for the same $x,x',y$ and $\theta\in\tilde{\Theta}$.
Consequently, Assumptions \ref{a11} -- \ref{a2} follow from Assumptions \ref{c1} -- \ref{c4}
when $\Theta$ is restricted to $\tilde{\Theta}$
(i.e., when $\Theta$ is replaced with $\tilde{\Theta}$).
Then, using Theorems \ref{theorem1}, \ref{theorem2}, we conclude that there exist
functions $\tilde{l},\tilde{h}:\tilde{\Theta}\rightarrow\mathbb{R}$
such that $\tilde{l}(\theta)$, $\tilde{h}(\theta)$ are real-analytic in $\theta$
and satisfy $\lim_{n\rightarrow\infty}l_{n}(\theta,\lambda)=\tilde{l}(\theta)$,
$\lim_{n\rightarrow\infty}h_{n}(\theta,\lambda)=\tilde{h}(\theta)$
for $\theta\in\tilde{\Theta}$, $\lambda\in{\cal P}({\cal X})$
($l_{n}(\theta,\lambda)$, $h_{n}(\theta,\lambda)$ have the same meaning as in
(\ref{1.1})).
Consequently, Corollaries \ref{corollary1c}, \ref{corollary2c} hold
(we use here representation $\Theta=\bigcup_{n=1}^{\infty}\tilde{\Theta}_{n}$,
where $\{\tilde{\Theta}_{n} \}_{n\geq 1}$ is a sequence of non-empty open balls satisfying
$\text{cl}\tilde{\Theta}_{n}\subset\Theta$ for $n\geq 1$).
\end{proof}

\refstepcounter{appendixcounter}\label{appendix1}
\section*{Appendix \arabic{appendixcounter} }

This section contains some auxiliary results which are relevant for the
proof of Lemmas \ref{lemma2.4}, \ref{lemma2.2}, \ref{lemma1.2}, \ref{lemma3.2}
and Corollaries \ref{corollary1b} -- \ref{corollary2c}.
Here, we rely on the following notations.
$d_{w}\geq 1$ and $d_{z}\geq 1$ are integers, while
$A$ is a bounded convex set in $\mathbb{C}^{d_{w}}$.
$F(w,z)$ is a function mapping $w\in\mathbb{C}^{d_{w}}$, $z\in\mathbb{R}^{d_{z}}$
to $\mathbb{C}$, while $\lambda({\rm d}z)$ is a measure on $\mathbb{R}^{d_{z}}$.
$f(w)$ is the function defined by
\begin{align*}
	f(w) = \int F(w,z) \lambda({\rm d}z)
\end{align*}
for $w\in\mathbb{C}^{d_{w}}$.

\begin{lemmaappendix}\label{lemmaa1}
Assume the following:

(i) There exists a real number $\delta\in(0,1)$
such that $F(w,z)$ is analytic in $w$
for each $w\in V_{\delta}(A)$, $z\in\mathbb{R}^{d_{z}}$.

(ii) There exists a function $\phi:\mathbb{R}^{d_{z}}\rightarrow[1,\infty)$
such that
$|F(w,z)|\leq\phi(z)$ for all $w\in V_{\delta}(A)$, $z\in\mathbb{R}^{d_{z}}$.
\newline
Then, we have
\begin{align*}
	|F(w',z)-F(w'',z) |
	\leq
	\frac{d_{w}\phi(z)\|w'-w''\|}{\delta}
\end{align*}
for all $w',w''\in V_{\delta}(A)$, $z\in\mathbb{R}^{d_{z}}$.
Moreover, if $\int \phi(z)\lambda({\rm d}z) < \infty$, then
$f(w)$ is well-defined and analytic for all $w\in V_{\delta}(A)$.
\end{lemmaappendix}

\begin{proof}
Owing to Cauchy's inequality (see e.g., \cite[Proposition 2.1.3]{taylor}) and (i), (ii),
we have
\begin{align}\label{la1.1}
	\|\nabla_{x} F(w,z) \|
	\leq
	\frac{d_{w}\phi(z)}{\delta}
\end{align}
for $w\in V_{\delta}(A)$, $z\in\mathbb{R}^{d_{z}}$.
Consequently, we get
\begin{align*}
	&
	|F(w',z) - F(w'',z) |
	\\
	&=
	\left|
	\int_{0}^{1}
	\left( \nabla_{w} F(tw'+(1-t)w'',z ) \right)^{T} (w'-w'') {\rm d}t
	\right|
	\\
	&\leq
	\int_{0}^{1}
	\left\| \nabla_{w} F(tw'+(1-t)w'',z ) \right\| \|w'-w'' \| {\rm d}t
	\\
	&\leq
	\frac{d_{w}\phi(z)\|w'-w''\|}{\delta}
\end{align*}
for $w',w''\in V_{\delta}(A)$, $z\in\mathbb{R}^{d_{z}}$, 
as $tw'+(1-t)w''\in V_{\delta}(A)$ results from $t\in[0,1]$ and
the convexity of $V_{\delta}(A)$.
Moreover, if $\int\phi(z)\lambda({\rm d}z)<\infty$,
the dominated convergence theorem and (\ref{la1.1}) imply that
$f(w)$ is well-defined and differentiable for $w\in V_{\delta}(A)$.
Consequently, $f(w)$ is analytic for $w\in V_{\delta}(A)$.
\end{proof}

In the rest of this appendix, we use the following notations.
$B$ is a compact set in $\mathbb{R}^{d_{w}}$,
while $g(w)$ is a function mapping $w\in\mathbb{R}^{d_{w}}$ to $\mathbb{R}$
($d_{w}$ is specified at the beginning in the appendix).

\begin{lemmaappendix}\label{lemmaa2}
Assume that there exists an open set $C$ in $\mathbb{R}^{d_{w}}$
such that $B\subset C$ and $g(w)$ is real-analytic on $C$.
Then, there exists a function $\hat{g}(w)$ with the following properties:

(i) $\hat{g}(w)$ maps $w\in\mathbb{C}^{d_{w}}$ to $\mathbb{C}$.

(ii) $\hat{g}(w)=g(w)$ for all $w\in B$.

(iii) There exists a real number $\delta\in(0,1)$ such that $\hat{g}(w)$ is analytic
on $V_{\delta}(B)$.
\end{lemmaappendix}

\begin{proof}
First, we assume that $B$ is connected (latter, this assumption is dropped).
As $g(w)$ is real-analytic on $C$,
$g(w)$ has an analytic continuation in an open vicinity of any point in $C$.
Hence, there exist functions $\hat{g}(w,v)$, $\delta(v)$ with the following properties:

(iv) $\hat{g}(w,v)$, $\delta(v)$ map $w\in\mathbb{C}^{d_{w}}$, $v\in C$
to $\mathbb{C}$, $(0,1)$ (respectively).

(v) $\hat{g}(w,v)=g(w)$ for $w\in V_{\delta(v)}(v) \cap \mathbb{R}^{d_{w}}$, $v\in C$.

(vi) $\hat{g}(w,v)$ is analytic in $w$ for $w\in V_{\delta(v)}(v)$, $v\in C$.
\newline
Since $B$ is compact, there exist an integer $M\geq 1$ and points $\{v_{i} \}_{1\leq i\leq M}$
such that $v_{i}\in B$ for $1\leq i\leq M$
and $B\subset \bigcup_{i=1}^{M} V_{\delta(v_{i})}(v_{i})$.
Let $\hat{g}_{i}(w)=\hat{g}(w,v_{i})$, $V_{i}=V_{\delta(v_{i})}(v_{i})$
for $w\in\mathbb{C}^{d_{w}}$, $1\leq i\leq M$.
As $B$ is connected, for each $1\leq i\leq M$, there exists $1\leq j\leq M$, $j\neq i$
such that
$V_{i}\cap V_{j}\cap \mathbb{R}^{d_{w}} \neq \emptyset$.
Moreover,
if $V_{i}\cap V_{j}\cap \mathbb{R}^{d_{w}} \neq \emptyset$,
then $V_{i}\cap V_{j}\cap \mathbb{R}^{d_{w}}$ is a non-empty open set
and $\hat{g}_{i}(w)=\hat{g}_{j}(w)=g(w)$ for $w\in V_{i}\cap V_{j}\cap \mathbb{R}^{d_{w}}$.
Then, by the uniqueness of analytic continuation
(see e.g., \cite[Corollary 1.2.6]{krantz&park}),
for each $1\leq i\leq M$,
there exist $1\leq j\leq M$, $j\neq i$
and a function $\hat{g}_{ij}(w)$ with the following properties:

(vii) $\hat{g}_{ij}(w)$ maps $w\in\mathbb{C}^{d_{w}}$ to $\mathbb{C}$.

(viii) $\hat{g}_{ij}(w)$ is analytic on $V_{i}\cup V_{j}$.

(ix) $\hat{g}_{ij}(w)=\hat{g}_{i}(w)$ for $w\in V_{i}$
and $\hat{g}_{ij}(w)=\hat{g}_{j}(w)$ for $w\in V_{j}$.
\newline
Following these arguments, we conclude that there exists a function $\hat{g}(w)$
with the following properties:

(x) $\hat{g}(w)$ maps $w\in\mathbb{C}^{d_{w}}$ to $\mathbb{C}$.

(xi) $\hat{g}(w)$ is analytic on $\bigcup_{i=1}^{M} V_{i}$.

(xii) $\hat{g}(w)=\hat{g}_{i}(w)$ for $w\in V_{i}$, $1\leq i\leq M$.
\newline
Now, we drop the assumption that $B$ is connected
(i.e., $B$ is any compact set in $\mathbb{R}^{d_{w}}$).
Since $B$ is compact, there exist an integer $N\geq 1$ and open sets
$\{W_{i} \}_{1\leq i\leq N}$ in $\mathbb{R}^{d_{w}}$
such that $W_{i}\subseteq C$, $B\cap W_{i}\neq \emptyset$, $W_{i}\cap W_{j} = \emptyset$
for $1\leq i,j\leq N$, $i\neq j$ and $B\subset \bigcup_{i=1}^{N} W_{i}$.
Let $B_{i}=B\cap W_{i}$ for $1\leq i\leq N$.
Hence, $\{B_{i} \}_{1\leq i\leq N}$ are connected components of $B$,
and thus, $\{B_{i} \}_{1\leq i\leq N}$ are compact and disjoint.
Then, according to what has already been shown,
there exist open sets $\{U_{i} \}_{1\leq i\leq N}$ in $\mathbb{C}^{d_{w} }$
and functions $\{\hat{g}_{i}(w) \}_{1\leq i\leq N}$
with the following properties:

(xiii) $B_{i}\subset U_{i}$, $U_{i}\cap U_{j} = \emptyset$ for $1\leq i,j\leq N$, $i\neq j$.

(xiv) $\hat{g}_{i}(w)$ maps $w\in\mathbb{C}^{d_{w}}$ to $\mathbb{C}$ for $1\leq i\leq N$.

(xv) $\hat{g}_{i}(w)=g(w)$ for $w\in B_{i}$, $1\leq i\leq N$.

(xvi) $\hat{g}_{i}(w)$ is analytic on $U_{i}$ for $1\leq i\leq N$.
\newline
Let $\hat{g}(w)$ be the function defined by
$\hat{g}(w)=\hat{g}_{i}(w)$ for $w\in U_{i}$, $1\leq i\leq N$
and $\hat{g}(w)=0$ for $w\not\in \bigcup_{i=1}^{N} U_{i}$.
Due to (xiii), $\hat{g}(w)$ is well-defined.
As $B$ is compact and $B\subset \bigcup_{i=1}^{N} U_{i}$ (owing to (xiii)),
there exists a real number $\delta\in(0,1)$
such that $B\subset V_{\delta}(B) \subset \bigcup_{i=1}^{N} U_{i}$.
Then, (xv), (xvi) imply that $\hat{g}(w)$ is analytic on $V_{\delta}(B)$
and satisfies $\hat{g}(w)=g(w)$ for $w\in B$.
\end{proof}

\refstepcounter{appendixcounter}\label{appendix2}
\section*{Appendix \arabic{appendixcounter} }

In this section, we show how Theorem \ref{theorem2} can be applied to
finite-state hidden Markov models.
We also provide a link between Theorem \ref{theorem2} and the results of
\cite{han&marcus3}.
Here, we assume that ${\cal X}$ has a finite number of elements.
We also assume ${\cal X}=\{1,\dots,N\}$
and $\mu(x)=1$ for each $x\in{\cal X}$
(in this case, $p_{\theta}(x'|x)$ is the conditional probability
of $X_{n+1}^{\theta,\lambda}=x'$ given $X_{n}^{\theta,\lambda}=x$).
Further to this, we introduce the following assumptions.

\begin{assumptionappendix}\label{ab1}
$p_{\theta}(x'|x)$ and $q_{\theta}(y|x)$ are real-analytic in $\theta$
for each $\theta\in\Theta$, $x,x'\in{\cal X}$, $y\in{\cal Y}$.
Moreover, $p_{\theta}(x'|x)$ and $q_{\theta}(y|x)$ have complex-valued continuations
$\hat{p}_{\eta}(x'|x)$ and $\hat{q}_{\eta}(y|x)$ with the following properties:

(i) $\hat{p}_{\eta}(x'|x)$ and $\hat{q}_{\eta}(y|x)$ map $\eta\in\mathbb{C}^{d}$,
$x,x'\in{\cal X}$, $y\in{\cal Y}$ to $\mathbb{C}$.

(ii) $\hat{p}_{\theta}(x'|x)=p_{\theta}(x'|x)$ and $\hat{q}_{\theta}(y|x)=q_{\theta}(y|x)$
for all $\theta\in\Theta$, $x,x'\in{\cal X}$, $y\in{\cal Y}$.

(iii) There exists a real number $\delta\in(0,1)$ such that
$\hat{p}_{\eta}(x'|x)$ and $\hat{q}_{\eta}(y|x)$ are analytic in $\eta$
for each $\eta\in V_{\delta}(\Theta)$,
$x,x'\in{\cal X}$, $y\in{\cal Y}$.

(iv) There exists a real number $\varepsilon\in(0,1)$ such that
$\varepsilon\leq|\hat{p}_{\eta}(x'|x)|\leq 1/\varepsilon$
for all $\eta\in V_{\delta}(\Theta)$,
$x,x'\in{\cal X}$, $y\in{\cal Y}$.
\end{assumptionappendix}

\begin{assumptionappendix}\label{ab2}
There exists a real number $\alpha\in(0,1)$ and a vector $\hat{\theta}\in\Theta$
with the following properties:

(i) $q_{\hat{\theta}}(y|x)\neq 0$,
$q_{\theta}(y|x)/q_{\hat{\theta}}(y|x)\geq\alpha$
and $|\hat{q}_{\eta}(y|x)/q_{\hat{\theta}}(y|x)|\leq 1/\alpha$
for all $\theta\in\Theta$, $\eta\in V_{\delta}(\Theta)$,
$x\in{\cal X}$, $y\in{\cal Y}$.

(ii) $\int |\log q_{\hat{\theta}}(y|x')| \, q_{\hat{\theta}}(y|x) \nu({\rm d}y) < \infty$
for all $x,x'\in{\cal X}$.
\end{assumptionappendix}

\begin{assumptionappendix}\label{ab3}
There exists a real number $\beta\in(0,1)$, a vector $\hat{x}\in{\cal X}$
and functions $\tilde{\phi},\tilde{\psi}:{\cal Y}\rightarrow(0,\infty)$
with the following properties:

(i) $\hat{q}_{\eta}(y|\hat{x})\neq 0$,
$|\hat{q}_{\eta}(y|x)/\hat{q}_{\eta}(y|\hat{x})|\leq 1/\beta$
for all $\eta\in V_{\delta}(\Theta)$,
$x\in{\cal X}$, $y\in{\cal Y}$.

(ii) $|\hat{q}_{\eta}(y|\hat{x})|\leq\tilde{\phi}(y)$
and $|\log|\hat{q}_{\eta}(y|\hat{x})||\leq\tilde{\psi}(y)$
for all $\eta\in V_{\delta}(\Theta)$, $y\in{\cal Y}$.

(iii) $\int \tilde{\phi}(y) \nu({\rm d}y) < \infty$ and $\int \tilde{\psi}(y) \tilde{\phi}(y) \nu({\rm d}y) < \infty$.
\end{assumptionappendix}

Assumptions \ref{ab1} -- \ref{ab3} are a particular case of Assumptions \ref{a11} -- \ref{a3}
(see Corollary \ref{corollarya2.1} and its proof).
At the same time, Assumptions \ref{ab1} -- \ref{ab3} include, as a special case,
all conditions which the results of \cite{han&marcus3} are based on.\footnote
{Assumptions \ref{ab1} and \ref{ab2} follow (respectively)
from \cite[Conditions (a), (c.i)]{han&marcus3}
and \cite[Conditions (b), (c.iii), Equation (11)]{han&marcus3},
while Assumption \ref{ab3} results from
one of \cite[Conditions (c.iii), (d.i), Equation (7)]{han&marcus3}
and \cite[Conditions (c.ii), (c.iii), (d.i), Equation (8)]{han&marcus3}. }
Further to this, Assumptions \ref{ab1} -- \ref{ab3} considerably simplify the
conditions adopted in \cite{han&marcus3}.

\begin{corollaryappendix}\label{corollarya2.1}
Let Assumption \ref{ab1} and one of Assumptions \ref{ab2}, \ref{ab3} hold.
Then, all conclusions of Theorem \ref{theorem2} are true.
\end{corollaryappendix}

\begin{proof}
It is sufficient to show that Assumptions \ref{a11} -- \ref{a3} follow from
Assumption \ref{ab1} and one of Assumptions \ref{ab2}, \ref{ab3}.

(i) In this part of the proof, we demonstrate that Assumption \ref{a11} holds under
Assumption \ref{ab1}.
Let $\lambda_{\theta}({\rm d}x|y)$ be the measure on ${\cal X}$ defined by
\begin{align*}
	\lambda_{\theta}(B|y)
	=
	\sum_{x\in{\cal X}} q_{\theta}(y|x)I_{B}(x)\mu(x)
\end{align*}
for $\theta\in\Theta$, $y\in{\cal Y}$, $B\subseteq{\cal X}$.
Then, Assumption \ref{ab1} implies
\begin{align*}
	\sum_{x'\in B} r_{\theta}(y,x'|x)I_{B}(x')\mu(x')
	\geq &
	\varepsilon\sum_{x'\in B} q_{\theta}(y|x')I_{B}(x')\mu(x')
	\\
	=&
	\lambda_{\theta}(B|y),
	\\
	\sum_{x'\in B} r_{\theta}(y,x'|x)I_{B}(x')\mu(x')
	\leq &
	\frac{1}{\varepsilon}\sum_{x'\in B} q_{\theta}(y|x')I_{B}(x')\mu(x')
	\\
	=&
	\frac{\lambda_{\theta}(B|y) }{\varepsilon}
\end{align*}
for the same $\theta,y,B$ and $x\in{\cal X}$.
Hence, Assumption \ref{a11} holds.

(ii) In the next part of the proof, we show that Assumptions \ref{a1} -- \ref{a3}
follow from Assumptions \ref{ab1}, \ref{ab2}.
Let $\tilde{C}_{1}=\alpha^{-1}\varepsilon^{-1}$,
$\gamma=\alpha^{2}\varepsilon^{2}$,
while $\hat{r}_{\eta}(y,x'|x)$ and $\varphi_{\eta}(y)$ are the functions defined by
\begin{align*}
	&\hat{r}_{\eta}(y,x'|x)
	=
	\hat{q}_{\eta}(y|x')\hat{p}_{\eta}(x'|x),
	\\
	&\varphi_{\eta}(y)
	=
	\tilde{C}_{1}\sum_{x\in{\cal X}}q_{\hat{\theta}}(y|x)
\end{align*}
for $\eta\in\mathbb{C}^{d}$, $x,x'\in{\cal X}$, $y\in{\cal Y}$.
Since $\varphi_{\eta}(y)$ is constant in $\eta$,
it follows from Assumption \ref{ab1} that
$\hat{r}_{\eta}(y,x'|x)$ and $\varphi_{\eta}(y)$ are analytic in $\eta$
for each $\eta\in V_{\delta}(\Theta)$, $x,x'\in{\cal X}$, $y\in{\cal Y}$.
Moreover, Assumptions \ref{ab1}, \ref{ab2} yield
$\varphi_{\eta}(y)\neq 0$ and
\begin{align}\label{ca2.1.21}
	|\hat{r}_{\eta}(y,x'|x)|
	\leq
	\frac{|\hat{q}_{\eta}(y|x')|}{\varepsilon}
	\leq
	\frac{q_{\hat{\theta}}(y|x')|}{\alpha\varepsilon}
	=
	|\varphi_{\eta}(y)|
\end{align}
for the same $\eta,x,x',y$.
Assumptions \ref{ab1}, \ref{ab2} also imply
\begin{align}\label{ca2.1.23}
	\sum_{x'\in{\cal X}} r_{\theta}(y,x'|x) \mu(x')
	\geq
	\varepsilon
	\sum_{x'\in{\cal X}} q_{\theta}(y|x')
	\geq &
	\alpha\varepsilon
	\sum_{x'\in{\cal X}} q_{\hat{\theta}}(y|x')
	\nonumber\\
	= &
	\gamma|\varphi_{\theta}(y)|
\end{align}
for $\theta\in\Theta$, $x\in{\cal X}$, $y\in{\cal Y}$.

Let $\tilde{C}_{2}=\tilde{C}_{1}N$, while $\phi(y)$ and $\psi(y)$ are the functions defined by
\begin{align*}
	&
	\phi(y)=\tilde{C}_{1} \sum_{x\in{\cal X}} q_{\hat{\theta}}(y|x),
	\\
	&
	\psi(y)=\tilde{C}_{2}\left(1 + \sum_{x\in{\cal X}} |\log q_{\hat{\theta}}(y|x) | \right)
\end{align*}
for $y\in{\cal Y}$.
Then, due to Part (ii) of Assumption \ref{ab2}, we have
$\int \phi(y) \nu({\rm d}y) = \tilde{C}_{1}N < \infty$ and
\begin{align}\label{ca2.1.25}
	\int\! \psi(y) \phi(y) \nu({\rm d}y)
	\!=&
	\tilde{C}_{1}\tilde{C}_{2} \!\!
	\sum_{x,x'\in{\cal X}}
	\int |\log q_{\hat{\theta}}(y|x')| q_{\hat{\theta}}(y|x) \nu({\rm d}y)
	\nonumber\\
	&
	+
	\tilde{C}_{1}\tilde{C}_{2}N
	<\infty.
\end{align}
We also have
\begin{align*}
	\log|\varphi_{\eta}(y)|
	\leq &
	\log(\tilde{C}_{1}N)
	+
	\max_{x\in{\cal X}} \log q_{\hat{\theta}}(y|x)
	\\
	\leq &
	\tilde{C}_{1}N
	\left(1 + \sum_{x\in{\cal X}} |\log q_{\hat{\theta}}(y|x)| \right),
	\\
	\log|\varphi_{\eta}(y)|
	\geq &
	\log(\tilde{C}_{1}N)
	+
	\min_{x\in{\cal X}} \log q_{\hat{\theta}}(y|x)
	\\
	\geq &
	-\tilde{C}_{1}N
	\left(1 + \sum_{x\in{\cal X}} |\log q_{\hat{\theta}}(y|x)| \right)
\end{align*}
for $\eta\in V_{\delta}(\Theta)$, $y\in{\cal Y}$.
Consequently, we get
\begin{align}\label{ca2.1.27}
	|\varphi_{\eta}(y)|
	\leq
	\phi(y),
	\;\;\;\;\;\;\;
	|\log|\varphi_{\eta}(y)||
	\leq
	\psi(y)
\end{align}
for the same $\eta$, $y$.
Then, using (\ref{ca2.1.21}) -- (\ref{ca2.1.27}),
we conclude that Assumptions \ref{a1} -- \ref{a3} hold.

(iii) In this part of the proof, we show that Assumptions \ref{a1} -- \ref{a3}
follow from Assumptions \ref{ab1}, \ref{ab3}.
Let $\tilde{C}=\beta^{-1}\varepsilon^{-1}$,
$\gamma=\beta\varepsilon^{2}$,
while $\hat{r}_{\eta}(y,x'|x)$ and $\varphi_{\eta}(y)$ are the functions defined by
\begin{align*}
	&\hat{r}_{\eta}(y,x'|x)
	=
	\hat{q}_{\eta}(y|x')\hat{p}_{\eta}(x'|x),
	\\
	&\varphi_{\eta}(y)
	=
	\tilde{C}\hat{q}_{\eta}(y|\hat{x} )
\end{align*}
for $\eta\in\mathbb{C}^{d}$, $x,x'\in{\cal X}$, $y\in{\cal Y}$.
Then, due to Assumption \ref{ab1},
$\hat{r}_{\eta}(y,x'|x)$ and $\varphi_{\eta}(y)$ are analytic in $\eta$
for each $\eta\in V_{\delta}(\Theta)$, $x,x'\in{\cal X}$, $y\in{\cal Y}$.
Moreover, Assumptions \ref{ab1}, \ref{ab3} yield
$\varphi_{\eta}(y)\neq 0$ and
\begin{align}\label{ca2.1.1}
	|\hat{r}_{\eta}(y,x'|x)|
	\leq
	\frac{|\hat{q}_{\eta}(y|x')|}{\varepsilon}
	\leq
	\frac{|\hat{q}_{\eta}(y|\hat{x})|}{\beta\varepsilon}
	=
	|\varphi_{\eta}(y)|
\end{align}
for $\eta,x,x',y$.
Assumptions \ref{ab1}, \ref{ab3} also imply
\begin{align}\label{ca2.1.3}
	\sum_{x'\in{\cal X}} r_{\theta}(y,x'|x) \mu(x')
	\geq
	q_{\theta}(y|\hat{x}) p_{\theta}(\hat{x}|x)
	\geq &
	\varepsilon q_{\theta}(y|\hat{x})
	\nonumber\\
	=&
	\gamma|\varphi_{\theta}(y)|
\end{align}
for $\theta\in\Theta$, $x\in{\cal X}$, $y\in{\cal Y}$.

Let $\phi(y)$ and $\psi(y)$ be the functions defined by
\begin{align*}
	\phi(y)=\tilde{C}\tilde{\phi}(y),
	\;\;\;\;\;\;\;
	\psi(y)=\tilde{C}\big(1+\tilde{\psi}(y) \big)
\end{align*}
for $y\in{\cal Y}$.
Then, Assumption \ref{ab3} yields
\begin{align}\label{ca2.1.5}
	|\varphi_{\eta}(y)|
	=
	\tilde{C}|\hat{q}_{\eta}(y|\hat{x})|
	\leq
	\tilde{C}\tilde{\phi}(y)
	=
	\phi(y)
\end{align}
for $\eta\in V_{\delta}(\Theta)$, $y\in{\cal Y}$.
Assumption \ref{ab3} also implies
\begin{align}\label{ca2.1.7}
	|\log|\varphi_{\eta}(y)||
	\leq
	\log\tilde{C}
	+
	|\log|\hat{q}_{\eta}(y|\hat{x})||
	\leq &
	\tilde{C}(1+\tilde{\psi}(y) )
	\nonumber\\
	=&
	\psi(y)
\end{align}
for the same $\eta,y$.
Then, using Part (iii) of Assumption \ref{ab3},
and (\ref{ca2.1.1}) -- (\ref{ca2.1.7}),
we conclude that Assumptions \ref{a1} -- \ref{a3} hold.
\end{proof}

In rest of the section, we explain how Theorem \ref{theorem2} can further be extended
in the context of finite-state hidden Markov models.
Here, we rely on the following notations.
${\cal P}^{N}$ is the set of $N$-dimensional probability vectors,
while $e$ is the $N$-dimensional vector whose all elements are one.
For $\theta\in\Theta$, $y\in{\cal Y}$,
$R_{\theta}(y)$ is the $N\times N$ matrix whose $(x',x)$-entry is $r_{\theta}(y,x'|x)$,
where $r_{\theta}(y,x'|x)$ has the same meaning as in Section \ref{section1}.
$G_{\theta}(\lambda,y)$ and $h_{\theta}(\lambda,y)$ are the functions defined by
\begin{align*}
	G_{\theta}(\lambda,y)
	=
	\frac{R_{\theta}(y)\lambda}{e^{T}R_{\theta}(y)\lambda},
	\;\;\;\;\;\;\;
	h_{\theta}(\lambda,y)
	=
	\log\left(e^{T}R_{\theta}(y)\lambda \right)
\end{align*}
for $\theta\in\Theta$, $\lambda\in{\cal P}^{N}$, $y\in{\cal Y}$.
Regarding functions $r_{\theta}(y,x'|x)$, $G_{\theta}(\lambda,y)$ and $h_{\theta}(\lambda,y)$,
we assume the following.

\begin{assumptionappendix}\label{aa1}
There exist a real number $\varepsilon\in(0,1)$ and a function $s_{\theta}(y,x)$
mapping $\theta\in\Theta$, $x\in{\cal X}$, $y\in{\cal Y}$ to $[0,\infty)$
such that
\begin{align*}
	\varepsilon s_{\theta}(y,x')
	\leq
	r_{\theta}(y,x'|x)
	\leq
	\frac{s_{\theta}(y,x')}{\varepsilon}
\end{align*}
for all $\theta\in\Theta$, $x,x'\in{\cal X}$, $y\in{\cal Y}$.
\end{assumptionappendix}

\begin{assumptionappendix}\label{aa2}
$G_{\theta}(\lambda,y)$ and $h_{\theta}(\lambda,y)$ are real-analytic in $(\theta,\lambda)$
for all $\theta\in\Theta$, $\lambda\in{\cal P}^{N}$, $y\in{\cal Y}$.
Moreover, $G_{\theta}(\lambda,y)$ and $h_{\theta}(\lambda,y)$
have complex-valued continuations $\hat{G}_{\eta}(\xi,y)$ and $\hat{h}_{\eta}(\xi,y)$
with the following properties:

(i) $\hat{G}_{\eta}(\xi,y)$ and $\hat{h}_{\eta}(\xi,y)$ map $\eta\in\mathbb{C}^{d}$,
$\xi\in\mathbb{C}^{N}$, $y\in{\cal Y}$ to $\mathbb{C}^{N}$ and $\mathbb{C}$ (respectively).

(ii) $\hat{G}_{\theta}(\lambda,y)=G_{\theta}(\lambda,y)$
and $\hat{h}_{\theta}(\lambda,y)=h_{\theta}(\lambda,y)$
for all $\theta\in\Theta$, $\lambda\in{\cal P}^{N}$, $y\in{\cal Y}$.

(iii) There exists a real number $\delta\in(0,1)$ such that
$\hat{G}_{\eta}(\xi,y)$ and $\hat{h}_{\eta}(\xi,y)$ are analytic in $(\eta,\xi)$
for each $\eta\in V_{\delta}(\Theta)$, $\xi\in V_{\delta}({\cal P}^{N})$, $y\in{\cal Y}$.

(iv) There exist a real number $K\in[1,\infty)$ and a function $\tilde{\psi}:{\cal Y}\rightarrow[1,\infty)$
such that
$\int\!\exp\big(\tilde{\psi}(y)\big)\tilde{\psi}(y)\nu({\rm d}y)$ $<\infty$ and
\begin{align*}
	\|\hat{G}_{\eta}(\xi,y)\|
	\leq
	K,
	\;\;\;\;\;\;\;
	|\hat{h}_{\eta}(\xi,y)|
	\leq
	\tilde{\psi}(y)
\end{align*}
for all $\eta\in V_{\delta}(\Theta)$, $\xi\in V_{\delta}({\cal P}^{N})$,
$y\in{\cal Y}$.
\end{assumptionappendix}

Assumption \ref{aa1} corresponds to the stability of the hidden Markov model
$\left\{(X_{n}^{\theta,\lambda},Y_{n}^{\theta,\lambda})\right\}_{n\geq 0}$
and its optimal filter,
while Assumption \ref{aa2} is related to the parameterization
of the model $\left\{(X_{n}^{\theta,\lambda},Y_{n}^{\theta,\lambda})\right\}_{n\geq 0}$.
Assumptions \ref{aa1} and \ref{aa2} are the same as the (corresponding) assumptions
adopted in \cite{tadic2}.
Further to this, Assumptions \ref{aa1} and \ref{aa2} include, as a particular case,
all conditions which the results of \cite{han&marcus3} are based on.

\begin{theoremappendix}
Let Assumptions \ref{aa1} and \ref{aa2} hold.
Then, all conclusions of Theorem \ref{theorem2} are true.
\end{theoremappendix}

\begin{proof}
Let $e_{i}$ be the $i$-th standard unit vector in $\mathbb{R}^{N}$,
where $1\leq i\leq N$.
Moreover, let $\hat{r}_{\eta}(y,x'|x)$ be the function defined by
\begin{align*}
	\hat{r}_{\eta}(y,x'|x)
	=
	e_{x'}^{T} \hat{G}_{\eta}(e_{x},y) \exp\big(\hat{h}_{\eta}(e_{x},y)\big)
\end{align*}
for $\eta\in\mathbb{C}^{d}$, $x,x'\in{\cal X}$, $y\in{\cal Y}$,
while $\varphi_{\eta}(y)$, $\phi(y)$ and $\psi(y)$ are the functions defined by
\begin{align*}
	\varphi_{\eta}(y)
	=
	\phi(y)
	=
	K\exp\big(\tilde{\psi}(y)\big),
	\;\;\;\;\;\;\;
	\psi(y)
	=
	2K\tilde{\psi}(y)
\end{align*}
for the same $\eta,y$.
Then, it is straightforward to demonstrate that Assumptions \ref{a11}, \ref{a1} and \ref{a3} hold.

Let $\hat{s}_{\eta}(x)$ be the function defined by
\begin{align*}
	\hat{s}_{\eta}(x)
	=
	\sum_{x'\in{\cal X}}\int\hat{r}_{\eta}(y,x'|x)\nu({\rm d}y)
\end{align*}
for $\eta\in\mathbb{C}^{d}$, $x\in{\cal X}$,
while $\tilde{r}_{\eta}(y,x'|x)$ is the function be defined by
\begin{align*}
	\tilde{r}_{\eta}(y,x'|x)
	=
	\begin{cases}
	\hat{r}_{\eta}(y,x'|x) /\hat{s}_{\eta}(x),
	&\text{ if } \hat{s}_{\eta}(x)\neq 0
	\\
	0, &\text{ otherwise }
	\end{cases}
\end{align*}
for the same $\eta,x$ and $x'\in{\cal X}$, $y\in{\cal Y}$.
Moreover, let $T_{\eta}(z,B)$ be the kernel defined by
\begin{align*}
	T_{\eta}(z,B)
	=
	\sum_{x'\in{\cal X}}\int I_{B}(y',x') \tilde{r}_{\eta}(y',x'|x) \nu({\rm d}y')
\end{align*}
for $\eta\in\mathbb{C}^{d}$, $x\in{\cal X}$, $y\in{\cal Y}$, a Borel-set $B\subseteq{\cal Y}\times{\cal X}$
and $z=(y,x)$.
Since Assumptions \ref{a11}, \ref{a1} and \ref{a3} hold,
Lemma \ref{lemma2.1} implies that Assumptions \ref{a4}, \ref{a5} hold, too.

Let $\Phi_{\eta}(\xi,z)$ be the function defined by
\begin{align*}
	\Phi_{\eta}(\xi,z)
	=
	\hat{h}_{\eta}(\xi,y)
\end{align*}
for $\eta\in\mathbb{C}^{d}$, $\xi\in{C}^{N}$, $x\in{\cal X}$, $y\in{\cal Y}$
and $z=(y,x)$,
while $\left\{ F_{\eta,\boldsymbol y}^{m:n}(\xi) \right\}_{n\geq m\geq 0}$
are the functions recursively defined by
$F_{\eta,\boldsymbol y}^{m:m}(\xi)=\xi$ and
\begin{align*}
	F_{\eta,\boldsymbol y}^{m:n+1}(\xi)
	=
	\hat{G}_{\eta}\left(F_{\eta,\boldsymbol y}^{m:n}(\xi), y_{n+1} \right)
\end{align*}
for the same $\eta,\xi$
and a sequence $\boldsymbol y = \{y_{n} \}_{n\geq 1}$ in ${\cal Y}$.
Then, owing to Lemma \ref{lemmaa1}, the conclusions of Lemma \ref{lemma1.3} hold.
Moreover, due to \cite[Lemma 3]{tadic2},
the conclusions of Lemma \ref{lemma1.6} also hold
provided that elements of $\mathbb{C}^{N}$ are interpreted as complex measures on ${\cal X}$.
Combining Assumptions \ref{a4}, \ref{a5} and the conclusions of Lemmas \ref{lemma1.3}, \ref{lemma1.6},
we get the conclusions of Lemmas \ref{lemma3.1}, \ref{lemma3.2}.
Then, as a direct consequence of the conclusions of Lemmas \ref{lemma3.1}, \ref{lemma3.2},
we get the conclusions of Theorem \ref{theorem2}.
\end{proof}

\end{document}